\documentclass[10pt, letterpaper]{article}
\usepackage[latin9]{inputenc}
\usepackage[T1]{fontenc}
\usepackage[margin=1in]{geometry}
\usepackage{setspace,mathdots} 
\usepackage{verbatim}
\usepackage{float}
\usepackage{wrapfig,multirow}
\usepackage{sidecap}

\usepackage[usenames]{color}
\usepackage{prettyref,soul,xspace}

\usepackage{fancyhdr,lastpage}
\usepackage{graphicx,amsmath,amssymb,enumerate}

\usepackage{amsmath,amsfonts,amssymb,amsthm,graphicx,graphics,subfigure,epsfig,enumerate,bm}
\usepackage{ccaption}

  \captionnamefont{\itshape}
  \captiontitlefont{\itshape}
  \captiondelim{ . }

\usepackage{algorithm, multirow}
\usepackage{algorithmic}
\usepackage{graphicx}
\usepackage{amsmath}
\usepackage{amssymb}
\usepackage{wrapfig}

\newcommand{\bp}{{\boldsymbol p}}
\newcommand{\bx}{{\boldsymbol x}}

\newcommand{\by}{{\boldsymbol y}}
\newcommand{\be}{{\boldsymbol e}}

\newcommand{\argmax}{\mathop{\rm argmax}}
\newcommand{\argmin}{\mathop{\rm argmin}}

\newcommand{\bc}{{\boldsymbol c}}

\newcommand{\bu}{{\boldsymbol u}}
\newcommand{\bg}{{\boldsymbol g}}
\newcommand{\bv}{{\boldsymbol v}}
\newcommand{\bI}{{\boldsymbol I}}

\newcommand{\bW}{{\boldsymbol W}}
\newcommand{\bw}{{\boldsymbol w}}

\newrefformat{eq}{(\ref{#1})}
\newrefformat{sec}{Section~\ref{#1}}
\newrefformat{algo}{Algorithm~\ref{#1}}
\newrefformat{fig}{Fig.~\ref{#1}}
\newrefformat{tab}{Table~\ref{#1}}
\newrefformat{rmk}{Remark~\ref{#1}}
\newrefformat{clm}{Claim~\ref{#1}}
\newrefformat{def}{Definition~\ref{#1}}
\newrefformat{cor}{Corollary~\ref{#1}}
\newrefformat{lmm}{Lemma~\ref{#1}}
\newrefformat{lemma}{Lemma~\ref{#1}}
\newrefformat{prop}{Proposition~\ref{#1}}

\newcommand\reviseA[1]{{\color{black}#1}}

\allowdisplaybreaks

\newtheorem{theorem}{Theorem}[section]

\usepackage{algorithm, multirow}
\usepackage{algorithmic}
\usepackage{graphicx}
\usepackage{amsmath}
\usepackage{amssymb}
\usepackage{wrapfig}

\newcommand{\ba}{{\boldsymbol a}}

\newcommand{\cA}{{\mathcal A}}

\usepackage{sidecap,comment}
\sidecaptionvpos{figure}{c}

\usepackage[usenames]{color}
\usepackage{prettyref,soul,xspace}

\theoremstyle{definition}

\usepackage[
            CJKbookmarks=true,
            bookmarksnumbered=true,
            bookmarksopen=true,
            colorlinks=true,
            citecolor=red,
            linkcolor=blue,
            anchorcolor=red,
            urlcolor=blue
            ]{hyperref}

\title{Stable Separation and Super-Resolution of Mixture Models}
\author{ Yuanxin Li and Yuejie Chi \\
Department of Electrical and Computer Engineering \\
The Ohio State University, Columbus, Ohio, 43210\\
Email: \{li.3822, chi.97\}@osu.edu\thanks{This paper has been presented in part at 2015 International Symposium on Information Theory (ISIT) and 2015 International Conference on Sampling Theory and Applications (SampTA).}}

\begin{document}
\theoremstyle{plain}\newtheorem{lemma}{\textbf{Lemma}}\newtheorem{corollary}{\textbf{Corollary}}\newtheorem{assumption}{\textbf{Assumption}}\newtheorem{prop}{\textbf{Proposition}}
\newtheorem{example}{\textbf{Example}}
\theoremstyle{remark}\newtheorem{remark}{\textbf{Remark}}

\maketitle

\begin{abstract}
We consider simultaneously identifying the membership and locations of point sources that are convolved with different band-limited point spread functions, from the observation of their superpositions. This problem arises in three-dimensional super-resolution single-molecule imaging, neural spike sorting, multi-user channel identification, among other applications. We propose a novel algorithm, based on convex programming, and establish its near-optimal performance guarantee for exact recovery in the noise-free setting by exploiting the spectral sparsity of the point source models as well as the incoherence between point spread functions. Furthermore, robustness of the recovery algorithm in the presence of bounded noise is also established. Numerical examples are provided to demonstrate the effectiveness of the proposed approach. 
\end{abstract}

\textbf{Keywords:}
super-resolution, parameter estimation, atomic norm minimization, mixture models

\section{Introduction}
In many emerging applications in applied science and engineering, the acquired signal at the sensor can be regarded as a noisy superposition of returns from multiple {\em modalities}, where the return from each modality is a band-limited observation of a point source signal captured through a low-pass point spread function, governed by either the underlying physical field or the system design. Mathematically, we consider the following {\em parametric mixture model} of the acquired signal, $y(t)$, given as
\begin{equation}\label{demixing}
y(t) = \sum_{i=1}^I x_i( t )* g_i( t) + w (t)  = \sum_{i=1}^I \left( \sum_{k=1}^{K_i} a_{ik}g_i(t- \tau_{ik}) \right) + w(t), 
\end{equation}
where $*$ denotes the convolution operator, $w(t)$ is an additive noise, and $I$ is the total number of modalities. Moreover,  
\begin{equation*}
x_i(t) = \sum_{k=1}^{K_i}a_{ik}\delta(t-\tau_{ik})
\end{equation*}
is the point source signal observed from the $i$th modality, and $g_i(t)$ is the corresponding point spread function. For the $i$th modality, let ${\tau}_{ik}\in [0,1)$ and $a_{ik}\in\mathbb{C}$ be the location and the amplitude of the $k$th point source, $1\leq k\leq K_i$, respectively, where the locations of point sources $\tau_{ik}$'s are continuous-valued and can lie anywhere in the parameter space, at nature's will. The point source model can be used to model a variety of physical phenomena occurring in a wide range of practical problems, such as the activation pattern of fluorescence in single-molecule imaging \cite{rust2006sub}, sparse channel impulse response in multi-path fading environments, the locations of pollution plants in urban areas, firing times of neurons, and many more. 

Our goal is to {\em stably invert} for the field parameters, i.e.  the parameters of the point source models, of each modality from the acquired signal reflecting the ensemble behavior of all modalities, even in the presence of noise. This allows us to {\em separate} the contributions of each modality to the acquired signal. Moreover, typically we are interested in super resolution, i.e. resolving the parameters at a resolution much higher than the native resolution of the acquired signal, determined by the Rayleigh limit, or in other words, the reciprocal of the bandwidth of the point spread functions.

\subsection{Motivating Applications}
The mixture model \eqref{demixing} is motivated by the modeling and analysis of many practical problems, such as three-dimensional super-resolution single-molecule imaging \cite{huang2008three,huang20153d}, spike sorting in neural recording
\cite{ekanadham2014unified,knudson2014inferring}, multi-user multi-path channel identification \cite{chi2013compressive,applebaum2012asynchronous}, and blind calibration of time-interleaved analog-to-digital converters \cite{lu2010multichannel,li2015blind}. We describe several example applications below.  

\vspace{0.05in}
\noindent\textbf{Three-dimensional super-resolution single-molecule imaging:} By employing photoswitchable fluorescent molecules, the imaging process of single-molecule microscopies (Stochastic Optical Reconstruction Microscopy (STORM) \cite{rust2006sub} or Photo Activated Localization Microscopy (PALM) \cite{betzig2006imaging}) is divided into many frames, where in each frame, a {\em sparse} number of fluorophores (point sources) are randomly activated, localized at a resolution below the diffraction limit, and deactivated. The final image is thus obtained by superimposing the localization outcomes of all the frames. This principle can be extended to reconstruct a 3-D object from 2-D image frames, for example, by introducing a cylindrical lens to modulate the ellipticity of the point spread function based on the depth of the fluorescent object in 3-D STORM \cite{huang2008three}. Therefore, the acquired image in each frame can be regarded as a {\em superposition} of returns from multiple depth layers, where the return from each layer corresponds to the convolution outcome of the fluorophores in that depth layer with the depth-dependent point spread function, as modeled in \eqref{demixing}. The goal is thus to recover the locations and depth membership of each point source given the image frame. 

\vspace{0.05in}
\noindent\textbf{Spike sorting for neural recording:} Neurons in the brain communicate by firing action potentials, i.e. spikes, and it is possible to capture their communications through the use of a microelectrode, which records simultaneous activities of multiple neurons within a local neighborhood. Spike sorting \cite{lewicki1998review}, thus, refers to the grouping of spikes according to each neuron, from the recording of the microelectrode. Interestingly, it is possible to model the spike fired by each neuron with a {\em characteristic} shape \cite{gerstein1964simultaneous}. The neural recording can thus be modeled as a {\em superposition} of returns from multiple neurons, as in \eqref{demixing}, where the return from each neuron corresponds to the convolution of its characteristic spike shape with the sequence of its firing times. A similar problem also arises in DNA sequencing, please refer to \cite{li2000parametric}.

\vspace{0.05in}
\noindent\textbf{Multi-path identification in random-access channels:} In multi-user multiple access model \cite{applebaum2012asynchronous}, each active user transmits a signature waveform modulated via a signature sequence, which can be designed to optimize performance and the base station receives a {\em superposition} of returns from active users, as in \eqref{demixing}, where the received signal from each active user corresponds to the convolution of its signature waveform with the unknown {\em sparse} multi-path channel from the user to the base station. The goal is to identify the set of active users, as well as their channel states, from the received signal at the base station.

\subsection{Related Work and Our Contributions}
There is an extensive research literature \cite{stoica1997introduction} on inverting \eqref{demixing} when there is only a single modality with $I=1$, where conventional approaches for parameter estimation such as matched filtering, MUSIC \cite{Schmidt1986MUSIC}, matrix pencil \cite{HuaSarkar1990}, to more recent approaches based on the trigonometric polynomial frame \cite{mhaskar2000detection} or total variation minimization \cite{candes2014towards}, can be applied. However, these approaches can not be applied directly when multiple modalities exist in the observed signal, due to the {\em mutual interference}. To the best of the authors' knowledge, methods for inverting \eqref{demixing} with multiple modalities have been extremely limited. Sparse recovery algorithms have been proposed to estimate the mixture model in \cite{romberg2009multiple,chi2013compressive,applebaum2012asynchronous} with a discretized set of delays, but the performance may degenerate when the actual delays do not belong to the discrete grid \cite{Chi2011sensitivity}. Even when all the point sources indeed lie on the grid, existing work suggests that the sample complexity, or the bandwidth of the acquire signal, may have to grow logarithmically with the size of the grid, which is undesirable. More recently, \cite{ekanadham2014unified,knudson2014inferring} have proposed heuristic sparse recovery algorithms to estimate the continuous-valued delays in the mixture model for spike sorting, however no performance guarantees are available. Finally, an algebraic approach has been proposed in \cite{lu2010multichannel}, but it is sensitive to noise due to the nature of the employed root-finding procedure and does not extend well to a large number of modalities due to the prohibitive sample complexity.

In this paper, we study the problem of super-resolving the mixture model \eqref{demixing} when there are two modalities, i.e. $I=2$. The methodology in this paper can be extended straightforwardly to the analysis of the case $I>2$ and is left for future work. We start by recognizing that in the Fourier domain, the observed signal can be regarded as a linear combination of two spectrally-sparse signals, each composed of a small number of distinct complex sinusoids. The atomic norm \cite{Chandrasekaran2010convex, chandrasekaran2012convex} of spectrally-sparse signals is developed and proposed as an efficient convex optimization framework to motivate parsimonious structures \cite{candes2014towards, Chandrasekaran2010convex, chandrasekaran2012convex,tang2014CSoffgrid, fernandez2016super} recently, which can be computed efficiently via semidefinite programming. We then separate and recover the two signals by motivating their spectral structures using atomic norm minimization, in addition to satisfying the observation constraints. The proposed algorithm, denoted by {\em AtomicDemix}, is reminiscent of the algorithms for sparse error correction \cite{nguyen2013exact}, robust principal component analysis \cite{CanLiMaWri09}, demixing of sines and spikes \cite{chen2014robust,fernandez2016super}, and source separation  \cite{mccoy2014convexity}, where one aims to separate two low-dimensional signals with incoherent structures via convex optimization.

The separation and identification of the two point source signals, using the proposed {\em AtomicDemix} algorithm, is made possible with two additional natural conditions. The first condition is that the point source signal of each modality satisfies a mild separation condition, such that the locations of the point sources are separated by at least four times the Rayleigh limit; this is the same separation condition required by Cand\`es and Fernandez-Granda \cite{candes2014towards} even with $I=1$ when applying total variation minimization for super-resolution. The second condition is that the point spread functions of different modalities have to be sufficiently incoherent, which is supplied in our theoretical analysis by assuming they are {\em randomly} generated from a uniform distribution on the complex unit circle. Define $K_{\max}=\max \left\{K_1,K_2\right\}$. Our main results are summarized as below:
\begin{itemize}
\item For the noise-free case, we demonstrate that, provided that the coefficients of the point sources have symmetric random signs, that is to say the signs of the coefficients of the point sources are randomly generated from a symmetric distribution on the complex unit circle, as soon as the number of measurements $M$, or equivalently, the bandwidth of the point spread functions, is on the order $M/\log M = O( K_{\max}\log (K_1+K_2) )$, AtomicDemix exactly recovers the point source model of each modality with high probability. Since at least an order of $O(K_1+K_2)$ measurements is necessary, our sample complexity is near-optimal up to logarithmic factors. When the coefficients of the point sources have arbitrary signs, we establish a similar performance guarantee with a higher sample complexity, on the order of $M=O(K_{\max}^2\log (K_1+K_2))$.
\item For the noisy case, when the coefficients of the point sources have arbitrary signs, under same conditions that guarantee exact recovery in the noise-free case, we establish that AtomicDemix is stable in the presence of possibly adversarial bounded noise.
\item The point sources of each modality can be localized from the dual solution of the proposed algorithms, without estimating or knowing the model order a priori. Numerical examples are provided to corroborate the theoretical analysis, with comparisons against the standard Cram\'er-Rao Bound (CRB) for parameter estimation.
\end{itemize}

\subsection{Organization and Notations}
The rest of this paper is organized as follows. We specify the problem formulation and main results in Section~\ref{sec::problem_formulation}. Numerical experiments are provided to corroborate the theoretical analysis in Section~\ref{sec::numerical_results}. Section~\ref{sec::proof_theorem_noise_free} and Section~\ref{sec::proof_theorem_noisy} provide detailed proof procedures of our main results for the noise-free case and the noisy case, respectively. Finally, the paper is concluded in Section~\ref{sec::conclusion} with discussions on extensions and future work. 

Throughout the paper, $\left(\cdot\right)^{T}$ and $\left(\cdot\right)^H$ denote the transpose and Hermitian transpose, respectively, and $\bar{\left(\cdot\right)}$ denotes the (element-wise) conjugate of a complex scalar or vector. \reviseA{For a function $f(\tau)$ with variable $\tau$, we denote its first-order derivative and second-order derivative by $f'(\tau)$ and $f''(\tau)$, respectively. We also use $f^{(l)}(\tau)$ to represent its $l$th-order derivative.}
\reviseA{The quantity $\sqrt{-1}$ is denoted by $j$.} Besides, we use $C$ with different superscripts and subscripts to represent constants, whose values may change from line to line.

\section{Problem Formulation and Main Results} \label{sec::problem_formulation}

\subsection{Observation Model}
Due to hardware and physical limits, the resolution of the sensor suite is limited by the {\em diffraction limit} or {\em Rayleigh limit}, which heuristically is often referred to as half the width of the mainlobe of $g_{i}({t})$'s. Alternatively, in the frequency domain,  we say $g_{i}({t})$'s are band-limited with {\em cut-off frequency} $2M$. Denote the discrete-time Fourier transform of $g_{i}(t)$ as
\begin{equation}
g_{i,n} = \int_{-\infty}^{\infty} g_i(t)e^{-j2\pi nt}dt,
\end{equation}
then $g_{i,n} = 0$ whenever $n\notin\Omega_M=\left\{-2M,\dots,0,\dots,2M\right\}$. Taking the discrete-time Fourier transform of \eqref{demixing}, the measurements can be represented as, in the Fourier domain,
\begin{equation}\label{demixing_freq}
y_{n} = \sum_{i=1}^{I}g_{i,n}\cdot\left(\sum_{k=1}^{K_{i}}a_{ik}e^{-j2\pi n\tau_{ik}} \right) + w_{n}, \quad n\in\Omega_M,
\end{equation}
where the noise $w_n$ is
\begin{equation*}
w_n = \int_{-\infty}^{\infty} w(t)e^{-j2\pi nt} dt, \quad n\in\Omega_M.
\end{equation*}

When $I=2$, the measurements \eqref{demixing_freq} in the Fourier domain can be equivalently formulated as
\begin{equation}\label{mixing_freq}
y_n =  g_{1,n}  \cdot\left( \sum_{k=1}^{K_1} a_{1k} e^{-j2\pi n\tau_{1k} } \right) +  g_{2,n}  \cdot\left( \sum_{k=1}^{K_2} a_{2k} e^{-j2\pi n\tau_{2k} }\right) + w_{n}, \quad  n\in\Omega_M.
\end{equation}
The measurements $y_{n}$'s in \eqref{mixing_freq} can be considered as a linear combination of two spectrally-sparse signals, with $g_{i,n}$'s determining the combination coefficients. In vector form, we have
\begin{equation}\label{vector_two}
\by = \bg_1 \odot \bx_1^{\star} +\bg_2 \odot \bx_2^{\star} + \boldsymbol{w}, 
\end{equation}
where $\by=\left[y_{-2M},\ldots,y_{0},\ldots,y_{2M}\right]^{T}$, $\boldsymbol{w} = [w_{-2M},\ldots, w_0, \ldots, w_{2M} ]^T$, $\bg_i = [g_{i,-2M}, \ldots, g_{i,0},\ldots, g_{i,2M}]^T$ for $i=1,2$, and $\odot$ denotes the Hadamard element-wise product operator. Furthermore, let $\bx_{1}^{\star}\in\mathbb{C}^{4M+1}$ and $\bx_{2}^{\star}\in\mathbb{C}^{4M+1}$ denote two spectrally-sparse signals, each composed of a small number of distinct complex harmonics, represented as 
\begin{equation} 
\bx_{1}^{\star}=\sum_{k=1}^{K_{1}}a_{1k}\bc\left(\tau_{1k}\right), \quad \mathrm{and} \quad \bx_{2}^{\star}=\sum_{k=1}^{K_{2}}a_{2k}\bc\left(\tau_{2k}\right),
\end{equation}
where $K_{1}$ is the spectral sparsity of $\bx_{1}^{\star}$ and $K_{2}$ is the spectral sparsity of $\bx_{2}^{\star}$. The atom $\bc\left(\tau\right)$ is defined as 
\begin{equation*}
\bc(\tau) = \left[e^{-j2\pi(-2M)\tau},\ldots, 1,\ldots, e^{-j2\pi(2M)\tau}\right]^T,
\end{equation*}
which corresponds to a point source at the location $\tau\in[0,1)$. Further denote the location set of point sources in $\bx_{1}^{\star}$ and $\bx_{2}^{\star}$ by $\Upsilon_{1}=\left\{\tau_{11},\dots,\tau_{1K_{1}}\right\}$ and $\Upsilon_{2}=\left\{\tau_{21},\dots,\tau_{2K_{2}}\right\}$, respectively. The goal is thus to recover $\Upsilon_1$ and $\Upsilon_2$, and their corresponding amplitudes, from the observation \eqref{vector_two}.

Intuitively, it is impossible to separate the two modalities if $\bg_1$ and $\bg_2$ are highly coherent. In this paper, we assume the entries of the point spread functions $g_{i,n}$'s are i.i.d. generated from a uniform distribution on the complex unit circle. This randomness assumption is reasonable when $g_{i,n}$'s can be designed, such as the spreading sequences in multi-user communications, and provides the incoherence between different modalities that is necessary for separation. Multiplying both sides of \eqref{mixing_freq} with $\bar{g}_{1,n}$, and with slight abuse of notation, \eqref{vector_two} can be rewritten as 
\begin{equation}\label{equ_noisy}
\by =   \bx_1^{\star} + \boldsymbol{g} \odot \bx_2^{\star} + \boldsymbol{w}\in\mathbb{C}^{4M+1},
\end{equation}
where $\bg=\left[g_{-2M},\dots,g_{0},\dots,g_{2M}\right]^{T}\in\mathbb{C}^{4M+1}$ with $g_{n}=g_{2,n}\bar{g}_{1,n}$ uniformly drawn from the unit complex circle. In the noisy case, we consider the scenario where $\boldsymbol{w}$ is bounded as $\left\Vert\boldsymbol{w}\right\Vert_{2}^{2}\le\sigma_{w}^{2}$. 

\subsection{AtomicDemix -- A Convex Programming for Demixing}\label{sec::main_results}

Define the atomic norm \cite{Chandrasekaran2010convex,chandrasekaran2012convex,tang2014CSoffgrid} of $\boldsymbol{x}\in\mathbb{C}^{N}$ with respect to the atoms $\bc(\tau)$ as 
\begin{equation*}
\left\Vert\bx\right\Vert_{\mathcal{A}}=\inf_{a_{k}\in\mathbb{C},\tau_{k}\in[0,1)}\left\{\sum_{k}\left\vert a_{k}\right\vert \Big| \bx=\sum_{k} a_{k} \bc\left(\tau_{k}\right) \right\},
\end{equation*}
which can be regarded as the tightest convex relaxation of counting the smallest number of atoms $\boldsymbol{c}(\tau)$ that is needed to represent a signal $\boldsymbol{x}$. Therefore, we seek to recover the signals $\bx_1$ and $\bx_2$ by motivating their spectral sparsity via minimizing the sum of their atomic norms, with respect to the observation constraint in the noise-free case where $\boldsymbol{w}=0$: 
\begin{equation}\label{convex_demixing_rewritten}
\{\hat{\bx}_1, \hat{\bx}_2\} = \argmin_{\bx_1,\bx_2} \| \bx_1\|_{\cA} +  \|\bx_2\|_{\cA}, \quad \mathrm{s.t.} \quad \by= \bx_1 + \boldsymbol{g} \odot \bx_2.
\end{equation}
In the noisy case, we propose a regularized atomic norm minimization algorithm as
\begin{equation}\label{algorithm_noisy_model}
\left\{\hat{\bx}_{1},\hat{\bx}_{2}\right\}=\argmin_{\bx_{1},\bx_{2}}\frac{1}{2}\left\Vert\by-\bx_{1}-\bg\odot\bx_{2}\right\Vert_{2}^{2}+\lambda_{w}\left(\left\Vert\bx_{1}\right\Vert_{\mathcal{A}}+\left\Vert\bx_{2}\right\Vert_{\mathcal{A}}\right),
\end{equation}
where $\lambda_{w}$ is the regularization parameter to balance the data fitting term and the structural promoting term, to be determined later. The above algorithms are referred to as {\em AtomicDemix}. Interestingly, the atomic norm $\|\bx_i\|_{\cA}$ can be equivalently characterized via semidefinite programming \cite{tang2014CSoffgrid}, therefore the proposed algorithms can be solved efficiently using off-the-shelf solvers.

\subsection{Performance Guarantee in the Noise-free Case}
Recall $K_{\max} = \max \left\{K_1,K_2\right\}$. Define the separation of the point source signal of the $i$th modality as
\begin{equation}\label{separation} 
\Delta_i =  \min_{k\neq t} \left\vert\tau_{ik} - \tau_{it}\right\vert,
\end{equation}
which is understood as the wrapped-around distance on $[0,1)$, and the minimum separation of the point source signals of all modalities as $\Delta =\min_i \Delta_i$. We have the following performance guarantee for the noise-free algorithm \eqref{convex_demixing_rewritten}, whose proof is provided in Section~\ref{sec::proof_theorem_noise_free}. 

\begin{theorem}[Noise-free Case]\label{theorem_main}
Assume that $g_{n}=e^{j2\pi\phi_n}$'s are i.i.d. randomly generated from a uniform distribution on the complex unit circle with $\phi_n\sim\mathcal{U}[0,1]$, and that the minimum separation satisfies $\Delta \ge 1/M$. \reviseA{Let $\eta\in(0,1)$}, then there exists a numerical constant $C$ such that
\begin{equation}\label{sample_deterministic_sign}
M\ge C \max\left\{\log^{2}{\left(\frac{M\left( {K_{1}}+ {K_{2}}\right) }{\eta}\right)}, K_{\max}\log{\left(\frac{M\left(K_{1}+K_{2}\right)}{\eta}\right)}, K_{\max}^{2}\log{\left(\frac{K_{1}+K_{2}}{\eta}\right)} \right\}
\end{equation}
is sufficient to guarantee that $\boldsymbol{x}_{1}^{\star}$ and $\boldsymbol{x}_{2}^{\star}$ are the unique solutions of \eqref{convex_demixing_rewritten} with probability at least $1-\eta$. 

Moreover, if the signs of the coefficients $a_{ik}$'s are i.i.d. generated from a symmetric distribution on the complex unit circle, there exists a numerical constant $C$ such that
\begin{equation}\label{sample_random_sign}
M\ge C\max{\left\{\log^{2}{\left(\frac{M\left( {K_{1}}+ {K_{2}}\right) }{\eta}\right)},K_{\max}\log{\left(\frac{K_{1}+K_{2}}{\eta}\right)}\log{\left(\frac{M\left( {K_{1}}+ {K_{2}}\right) }{\eta}\right)}\right\}}
\end{equation}
is sufficient to guarantee that $\boldsymbol{x}_{1}^{\star}$ and $\boldsymbol{x}_{2}^{\star}$ are the unique solutions of \eqref{convex_demixing_rewritten} with probability at least $1-\eta$.
\end{theorem}

Theorem~\ref{theorem_main} provides two sample complexities depending on whether the signs of the coefficients $a_{ik}$'s are random. Given random signs of $a_{ik}$'s, Theorem~\ref{theorem_main} indicates that as soon as the number of measurements $M$ is on the order $M/\log M = O(K_{\max}\log (K_1+K_2) )$, AtomicDemix exactly recovers the point source models with high probability. This suggests that the performance of AtomicDemix is near-optimal in terms of the sample complexity as at least $O(K_1+K_2)$ measurements are necessary to identify the unknown parameters. Without requiring random signs of $a_{ik}$'s, the sample complexity is slightly higher, roughly dominated by the last term on the order of $M=O(K_{\max}^2\log(K_1+K_2))$.

\begin{remark}
The separation condition $\Delta\ge 1/M$ is a sufficient condition in Theorem~\ref{theorem_main} to guarantee accurate signal demixing, which is the same as the one required by Cand\`es and Fernandez-Granda in \cite{candes2014towards} even with $I=1$. Our results suggest that the separation condition to achieve super resolution in mixture models is no stronger than that required even in the single modality case, provided the point spread functions are incoherent enough. It is implied in \cite{candes2014towards,moitra2014threshold} that a reasonable separation is also necessary to guarantee stable super-resolution. Interestingly, no separation between point sources from {\em different} modalities is required, as long as their point spread functions are incoherent enough.
\end{remark}

\begin{remark}
Theorem~\ref{theorem_main} assumes $g_{n}$'s are i.i.d. from a uniform distribution on the complex unit circle, which may be relaxed as long as $g_{n}$'s are independently drawn from a distribution satisfying $\mathbb{E}\left[\bar{g}_{n}\right]=\mathbb{E}\left[ \bar{g}_{n}^{-1} \right]=0$ and $C_{1}\le\left\vert g_{n}\right\vert\le C_{2}$ for some constants $0\leq C_1\leq C_2$. Both $\mathrm{sign}\left(a_{1k}\right)$ and $\mathrm{sign}\left(a_{2k}\right)$ are assumed randomly generated, which are reasonable in many applications.
\end{remark}

\begin{remark}
Theorem~\ref{theorem_main} can also be extended into multi-dimensional point source models, following similar techniques in \cite{chi2014compressive}, where the same order of measurements shall be sufficient to localize the point sources under similar mild separation conditions. We leave this extension to interested readers.
\end{remark}

\subsection{Performance Guarantee in the Noisy Case}

In the presence of bounded noise, AtomicDemix in \eqref{algorithm_noisy_model} still stably recovers the point source signals, as established in the following theorem, whose proof is provided in Section~\ref{sec::proof_theorem_noisy}.
\begin{theorem}[Noisy Case]\label{theorem_noisy}
Let $\lambda_{w} = C_{w} \sigma_{w}\sqrt{4M+1}$, for some constant $C_{w} > 1$ large enough. Assume that $g_{n}=e^{j2\pi\phi_n}$'s are i.i.d. randomly generated from a uniform distribution on the complex unit circle with $\phi_n\sim\mathcal{U}[0,1]$, and that the minimum separation satisfies $\Delta \ge 1/M$. \reviseA{Let $\eta\in(0,1)$}, then as long as \eqref{sample_deterministic_sign} holds for some constant $C$, the solution to \eqref{algorithm_noisy_model} satisfies
\begin{equation} \label{inversion_bound}
\frac{1}{\sqrt{4M+1}}\left(\left\Vert\hat{\bx}_{1}-\bx_{1}^{\star}\right\Vert_{2} + \left\Vert\hat{\bx}_{2}-\bx_{2}^{\star}\right\Vert_{2}\right)\leq C_1 \sigma_{w}\sqrt{K_{\max}^3\log{M}} ,
\end{equation}
and 
\begin{equation} \label{denoising_bound}
\frac{1}{\sqrt{4M+1}} \left\Vert( \hat{\bx}_{1}+ \bg \odot \hat{\bx}_{2})- (\bx_{1}^{\star} + \bg \odot \bx_{2}^{\star}) \right\Vert_{2} \leq C_2 \sigma_w \left(\frac{ K_{\max}^3\log{M}}{M}\right)^{1/4} ,
\end{equation}
with probability at least $1-\eta- C_3(M^3\log M)^{-1/2}$, where $C_1$, $C_2$ and $C_3$ are some constants.
\end{theorem}

Theorem~\ref{theorem_noisy} does not make any assumptions on the signs of the coefficients of point sources. It guarantees the stability for inversion in the presence of bounded noise, even when the noise is adversarially generated. When $\sigma_{w}=0$, Theorem~\ref{theorem_noisy} degenerates to the noise-free case, providing a performance guarantee of AtomicDemix in accordance with Theorem~\ref{theorem_main} when the point sources have deterministic coefficients. The first bound \eqref{inversion_bound} concerns signal reconstruction, which guarantees that one can stably separate $\hat{\bx}_1$ and $\hat{\bx}_2$ even in the presence of noise. The second bound \eqref{denoising_bound} concerns denoising, which guarantees that AtomicDemix can output a denoised signal $\hat{\by} = \hat{\bx}_{1}+ \bg \odot \hat{\bx}_{2}$ proportional to the noise level.

\subsection{Localization via Dual Polynomials}
With the demixing results $\hat{\bx}_{1}$ and $\hat{\bx}_{2}$, the source locations $\tau_{ik}$'s of each signal can be estimated accurately by MUSIC \cite{Schmidt1986MUSIC}, ESPRIT \cite{roy1989esprit}, the Prony's method \cite{prony1795essai} or other linear prediction methods. More interestingly, the source locations can be identified directly from the dual solutions of \eqref{convex_demixing_rewritten} and \eqref{algorithm_noisy_model}. The coefficients $\ba_{1}$ and $\ba_{2}$ can then be estimated by least-squares using the estimates of $\tau_{ik}$'s.

We first characterize the dual problem of \eqref{convex_demixing_rewritten} and \eqref{algorithm_noisy_model}. Define the inner product of two vectors as $\langle\boldsymbol{p}, \boldsymbol{x}\rangle=\boldsymbol{x}^H\boldsymbol{p}$ and the real-valued inner product as $\langle\boldsymbol{p}, \boldsymbol{x}\rangle_{\mathbb{R}}=\mathrm{Re}\left(\boldsymbol{x}^H\boldsymbol{p}\right)$, where $\mathrm{Re}(\cdot)$ takes the real part of a complex scaler. The dual norm of $\left\Vert\cdot\right\Vert_{\mathcal{A}}$ can be represented as 
\begin{equation*}
\left\Vert\boldsymbol{p}\right\Vert_{\mathcal{A}}^{\star}=\sup_{\left\Vert\boldsymbol{x}\right\Vert_{\mathcal{A}}\le 1}\ \langle\boldsymbol{p},\boldsymbol{x}\rangle_{\mathbb{R}}= \sup_{\tau\in[0,1)} \left\vert \langle \boldsymbol{p},  \boldsymbol{c}\left(\tau\right) \rangle \right\vert = \sup_{\tau\in[0,1)}\left\vert\sum_{n=-2M}^{2M}p_{n}e^{j2\pi n\tau}\right\vert,
\end{equation*}
where $\bp = \left[ p_{-2M}, \dots, p_{0}, \dots, p_{2M} \right]^{T}$. Then the dual problem of \eqref{convex_demixing_rewritten} can be written as
\begin{equation}\label{convex_demixing_dual}
\hat{\boldsymbol{p}}=\argmax_{\boldsymbol{p}}\;\langle\boldsymbol{p},\boldsymbol{y}\rangle_{\mathbb{R}}, \quad \mathrm{s.t.}\quad \left\Vert\boldsymbol{p}\right\Vert_{\mathcal{A}}^{\star}\le 1,\ \left\Vert\bar{\boldsymbol{g}}\odot\boldsymbol{p}\right\Vert_{\mathcal{A}}^{\star}\le 1,
\end{equation}
whose derivations can be found in Appendix~\ref{proof_dual}. Similarly, by standard Lagrangian calculation the dual problem of \eqref{algorithm_noisy_model} can be obtained as
\begin{equation}\label{convex_demixing_noise_dual}
\hat{\boldsymbol{p}}=\argmax_{\boldsymbol{p}}\ \frac{1}{2}\left(\left\Vert\by\right\Vert_{2}^{2}-\left\Vert\by-\lambda_{w} \boldsymbol{p}\right\Vert_{2}^{2}\right),\quad \mathrm{s.t.}\quad \left\Vert\boldsymbol{p}\right\Vert_{\mathcal{A}}^{\star}\le 1,\ \left\Vert\bar{\boldsymbol{g}}\odot\boldsymbol{p}\right\Vert_{\mathcal{A}}^{\star}\le 1.
\end{equation}

Based on the definition of the dual norm, define the dual polynomials $\hat{P}\left(\tau\right)$ and $\hat{Q}\left(\tau\right)$ generated from the dual solutions of \eqref{convex_demixing_dual} or \eqref{convex_demixing_noise_dual} as
\begin{equation*}
\hat{P}\left(\tau\right)=\sum_{n=-2M}^{2M} \hat{p}_{n}e^{j2\pi n\tau} , \quad  \hat{Q}\left(\tau\right)=\sum_{n=-2M}^{2M}\hat{p}_{n}\bar{g}_{n}e^{j2\pi n\tau}.
\end{equation*}
Then the source locations can be identified as 
\begin{equation*}
\hat{\Upsilon}_{1}=\left\{ \tau\in[0,1): \ \left\vert\hat{P}\left(\tau\right)\right\vert=1\right\}, \quad \mathrm{and} \quad
\hat{\Upsilon}_{2}=\left\{\tau\in [0,1):  \ \left\vert\hat{Q}\left(\tau\right)\right\vert=1\right\}.
\end{equation*}

For the noise-free case, it is straightforward to show that $\Upsilon_{1}\subseteq\hat{\Upsilon}_{1}$ and $\Upsilon_{2}\subseteq\hat{\Upsilon}_{2}$ whenever the optimal primal solution is $\left\{\bx_1^\star, \bx_2^\star\right\}$ in Appendix~\ref{proof_fre_recovery}. Note however in general both $\hat{\Upsilon}_{1}$ and $\hat{\Upsilon}_{2}$ may contain spurious source locations. Interested readers can refer to relevant discussions in \cite[Proposition 2.5]{tang2014CSoffgrid} on when the dual polynomials return exact source locations, which also apply to our proposed algorithms with little modifications.

\section{Numerical Examples}\label{sec::numerical_results}
We carry out a series of numerical simulations to validate the performance of AtomicDemix in both noise-free and noisy cases under different parameter settings.

\subsection{Phase Transitions in the Noise-free Case}
We first examine the phase transition as a function of $(K_1, K_2)$ for a fixed $M$. We vary the spectral sparsity levels of the two modalities as $K_1$ and $K_2$. For each pair of $(K_1, K_2)$, we first randomly generate a pair of point sources $\Upsilon_{1}$ and $\Upsilon_{2}$ that satisfy a separation condition $\Delta \ge 1/\left(2M\right)$, with the coefficients of the point sources i.i.d. drawn from the complex standard Gaussian distribution. For each Monte Carlo trial, we then randomly generate the point spread functions $g_n$'s in the Fourier domain with i.i.d. entries drawn uniformly from the complex unit circle, and perform AtomicDemix by solving \eqref{convex_demixing_rewritten} using CVX \cite{grant2008cvx}. The algorithm is considered successful when the normalized estimate error satisfies $\sum_{i=1}^{2}\left\Vert\hat{\bx}_{i}-\bx_{i}^{\star}\right\Vert_{2}/\left\Vert\bx_{i}^{\star}\right\Vert_{2}\le 10^{-4}$. 
\begin{figure}[ht]
\begin{center}
\begin{tabular}{cc} 
\includegraphics[width=0.47\textwidth]{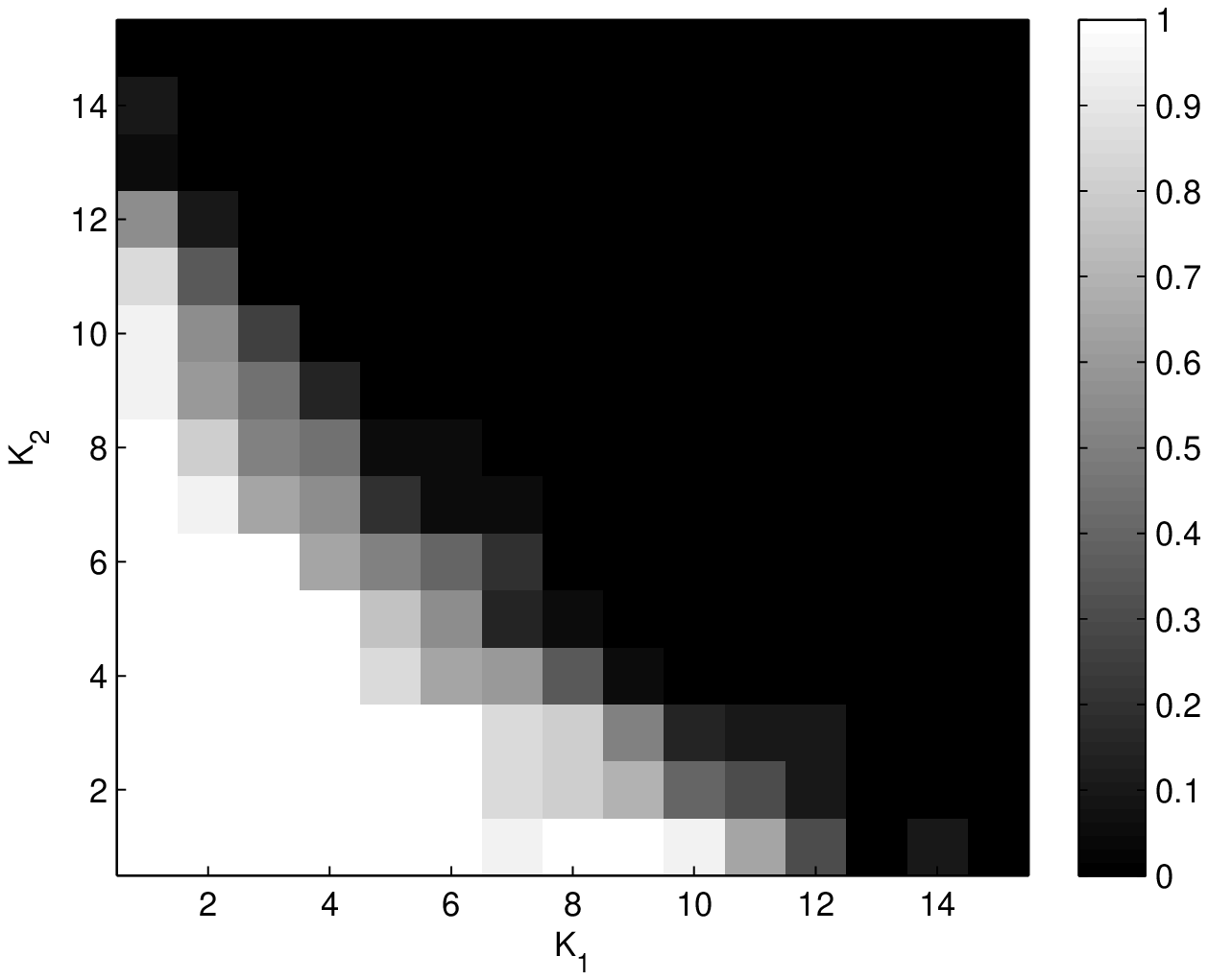} & \includegraphics[width=0.47\textwidth]{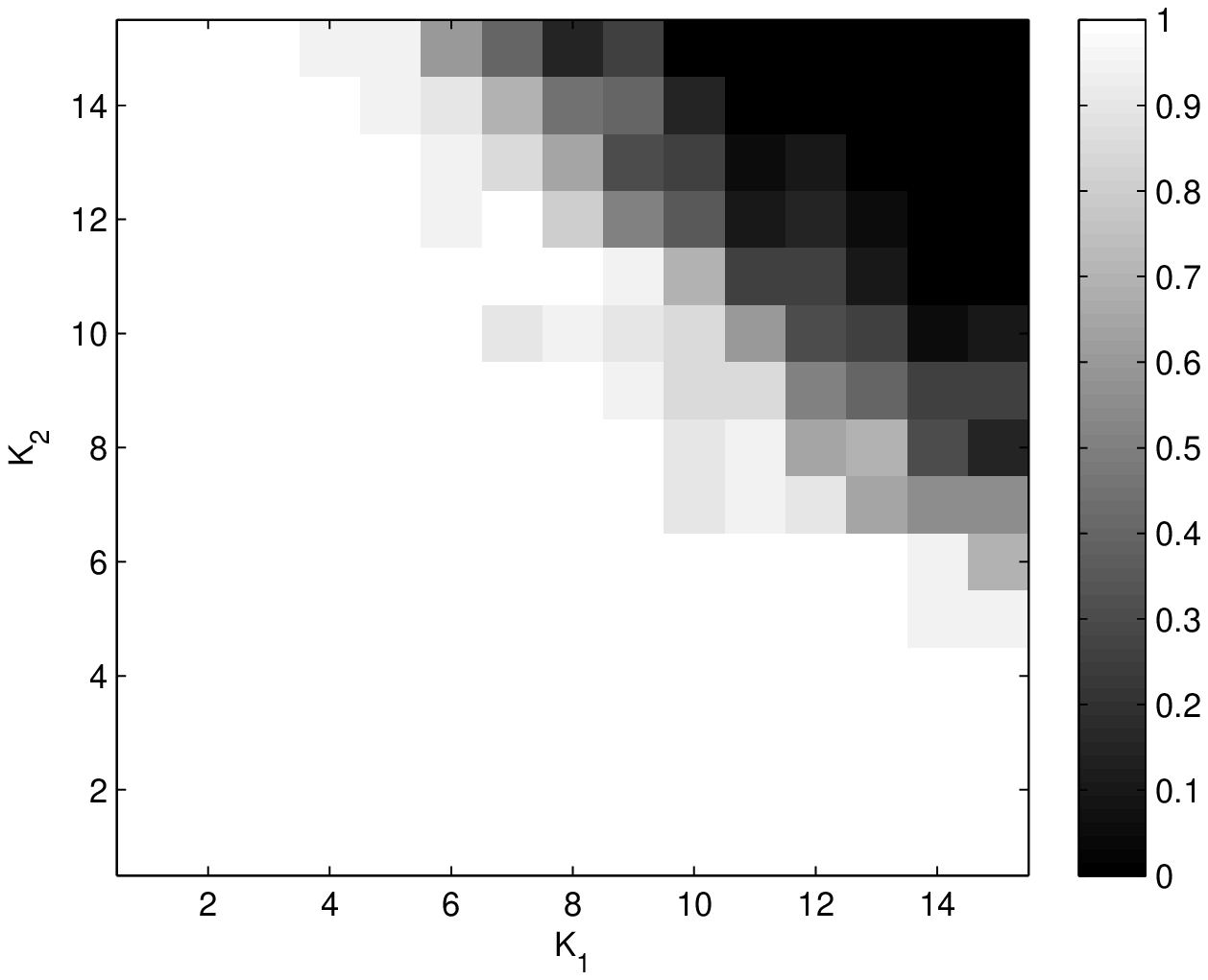} \\
(a) $M = 8$ & (b) $M=16$
\end{tabular}
\end{center}
\caption{Successful rates of AtomicDemix as a function of $(K_1, K_2)$ when (a) $M=8$ and (b) $M = 16$.}\label{fig:phase_tran}
\end{figure}

Fig.~\ref{fig:phase_tran} shows the success rates of AtomicDemix over $20$ Monte Carlo trials for each cell, when $M= 8$ in (a) and $M=16$ in (b), respectively. Fig.~\ref{fig:FixK1equK2_Mchange} (a) shows the success rates of AtomicDemix with respect to $M$ for different values of $K_1 = K_2$, and Fig.~\ref{fig:FixK1equK2_Mchange} (b) shows the success rates of AtomicDemix with respect to $K_1=K_2$ for different values of $M$.
\begin{figure}[htp]
\begin{center}
\begin{tabular}{cc}
\includegraphics[width=0.45\textwidth]{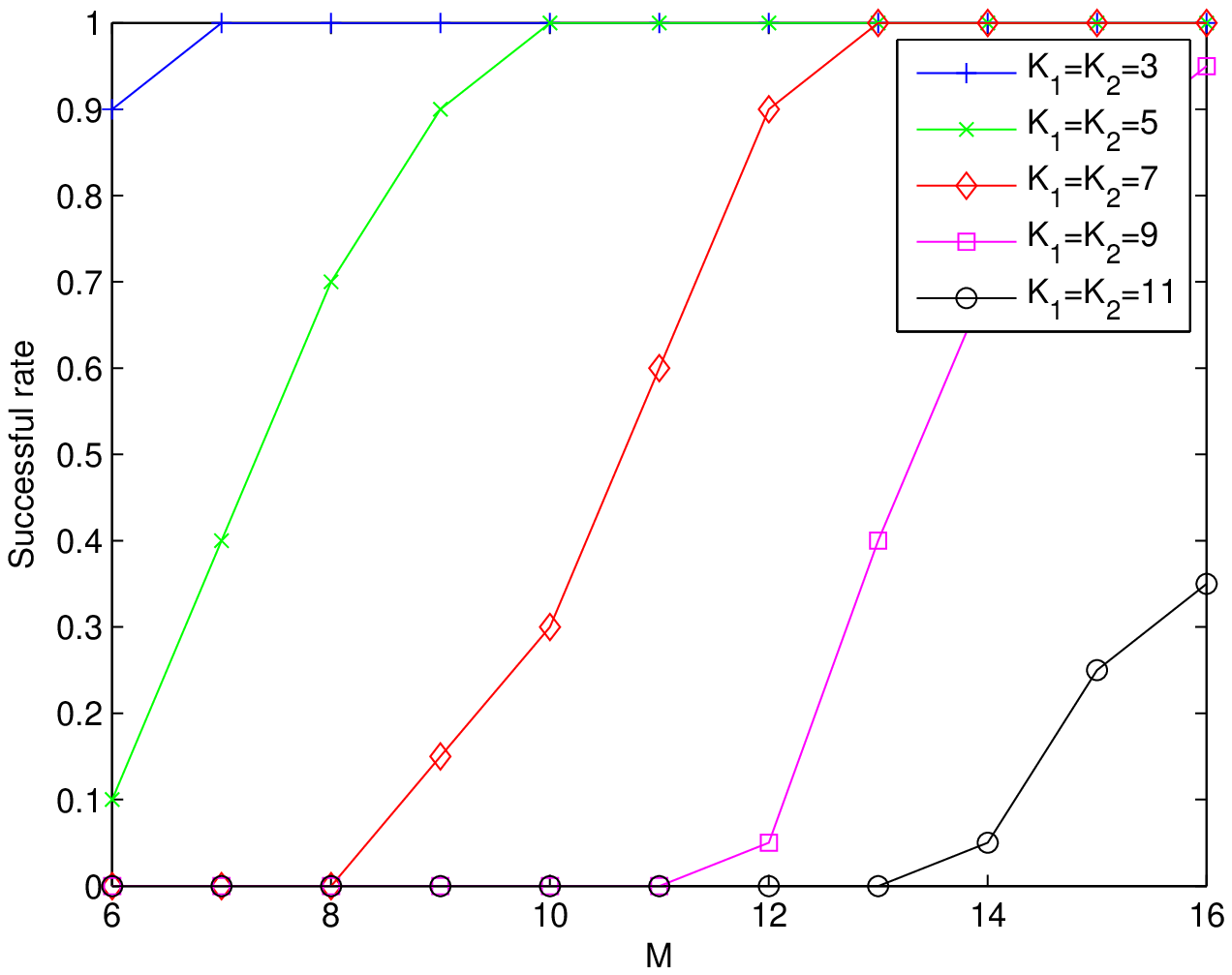} & \includegraphics[width=0.45\textwidth]{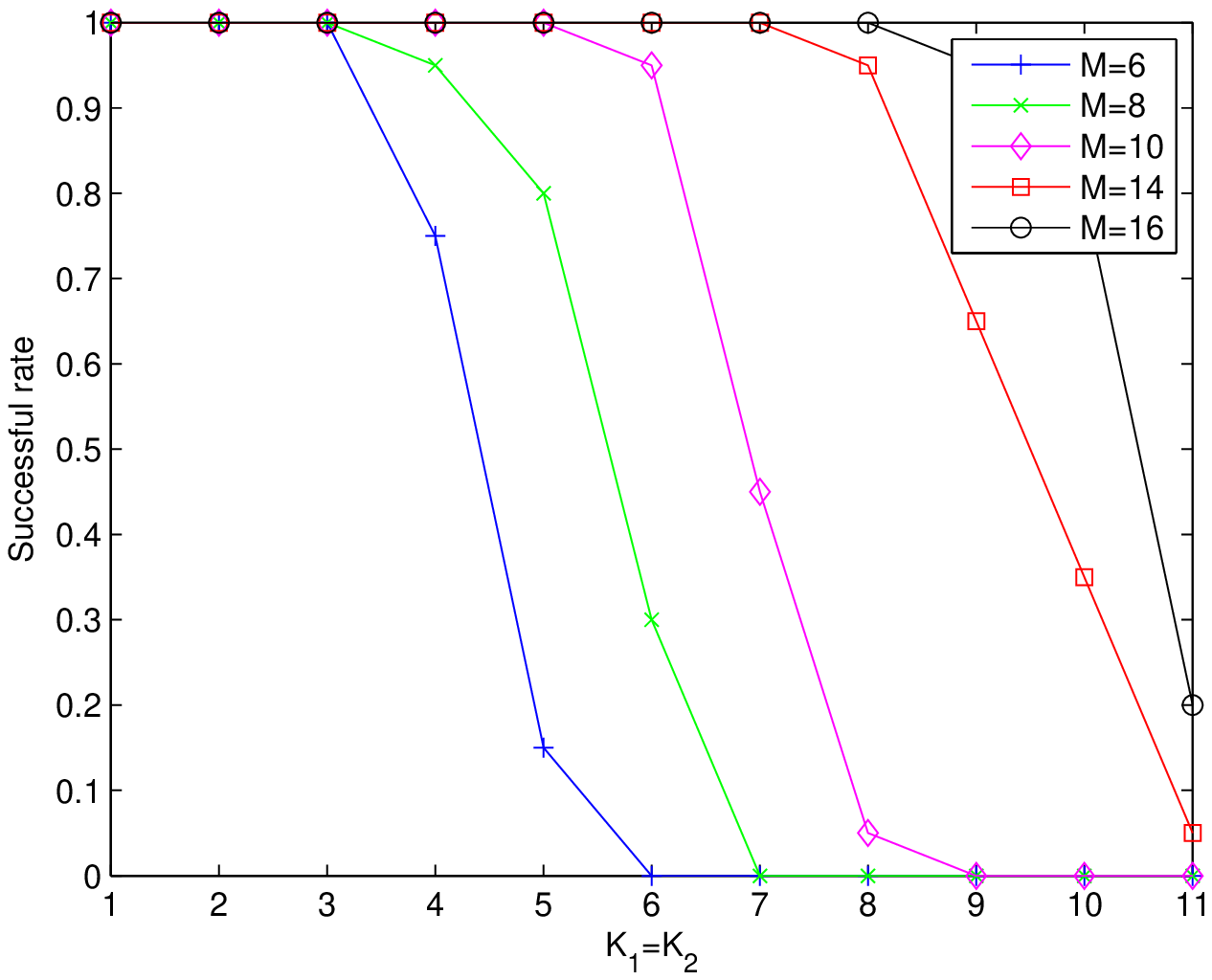}\\
(a) & (b) 
\end{tabular}
\end{center}
\caption{Success rates of AtomicDemix in the noise-free case (a) with respect to $M$ for various $K_1=K_2$ and (b) with respect to $K_1=K_2$ for various $M$. }
\label{fig:FixK1equK2_Mchange}
\end{figure}

\subsection{Point Source Recovery from Dual Polynomials}
As described earlier, the locations of the point sources can be recovered from the dual solutions of the proposed algorithm. Fix $M=16$, $K_{1}=4$ and $K_{2}=3$. We randomly generate a pair of point sources that satisfy a separation condition $\Delta \ge 1/\left(2M\right)$, with the coefficients of the point sources i.i.d. drawn from the complex standard Gaussian distribution. In the noise-free case, the amplitudes of the dual polynomials $\hat{P}\left(\tau\right)$ and $\hat{Q}\left(\tau\right)$ constructed from the solution of \eqref{convex_demixing_dual} are shown in Fig.~\ref{fig:fre_recov_dual} (a), superimposed on the ground truth, indicating the accurate recovery of the point sources.
\begin{figure}[ht]
\begin{center}
\begin{tabular}{c}
\includegraphics[height=1.6in,width=0.5\textwidth]{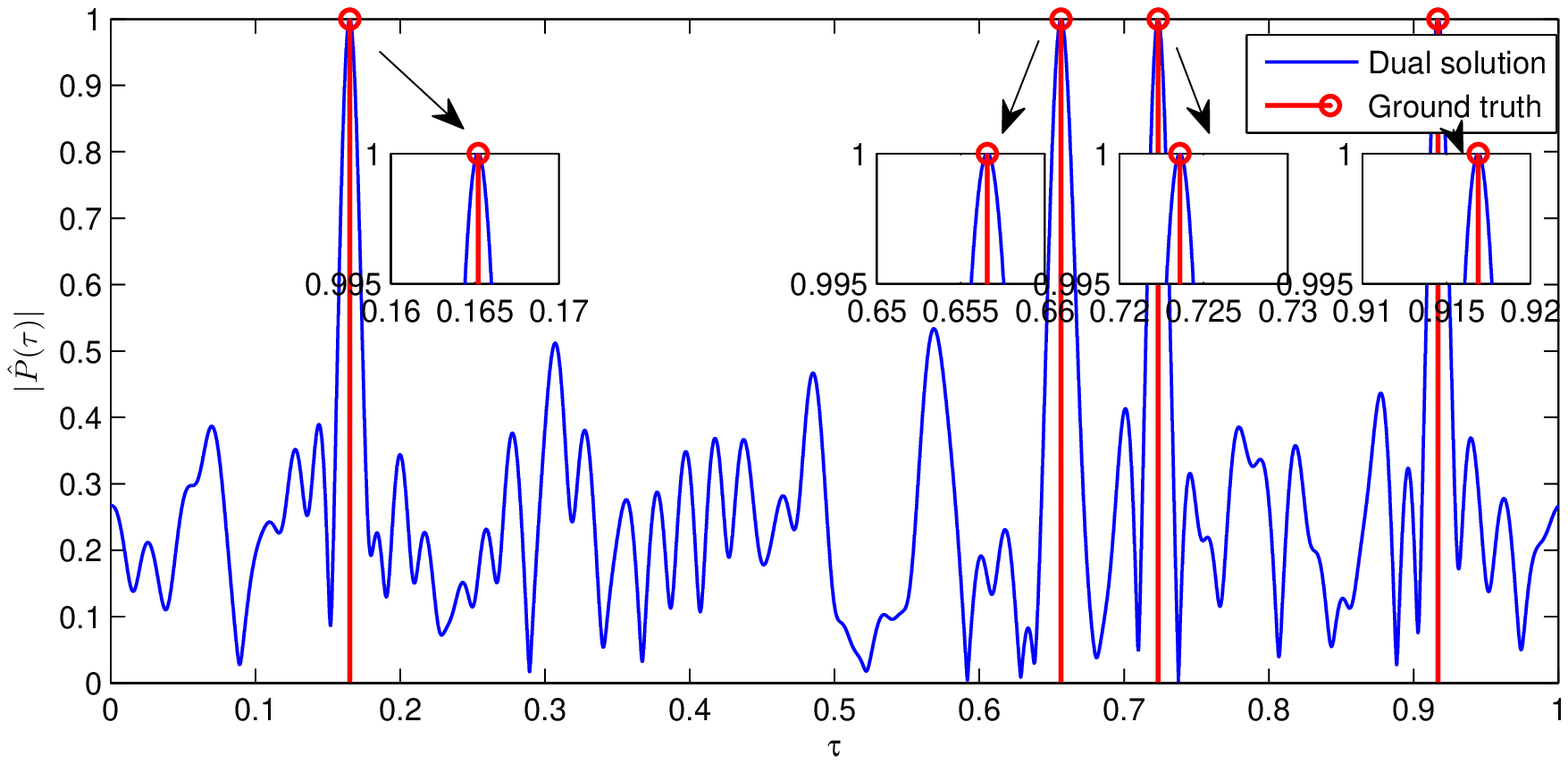} \includegraphics[height=1.6in,width=0.5\textwidth]{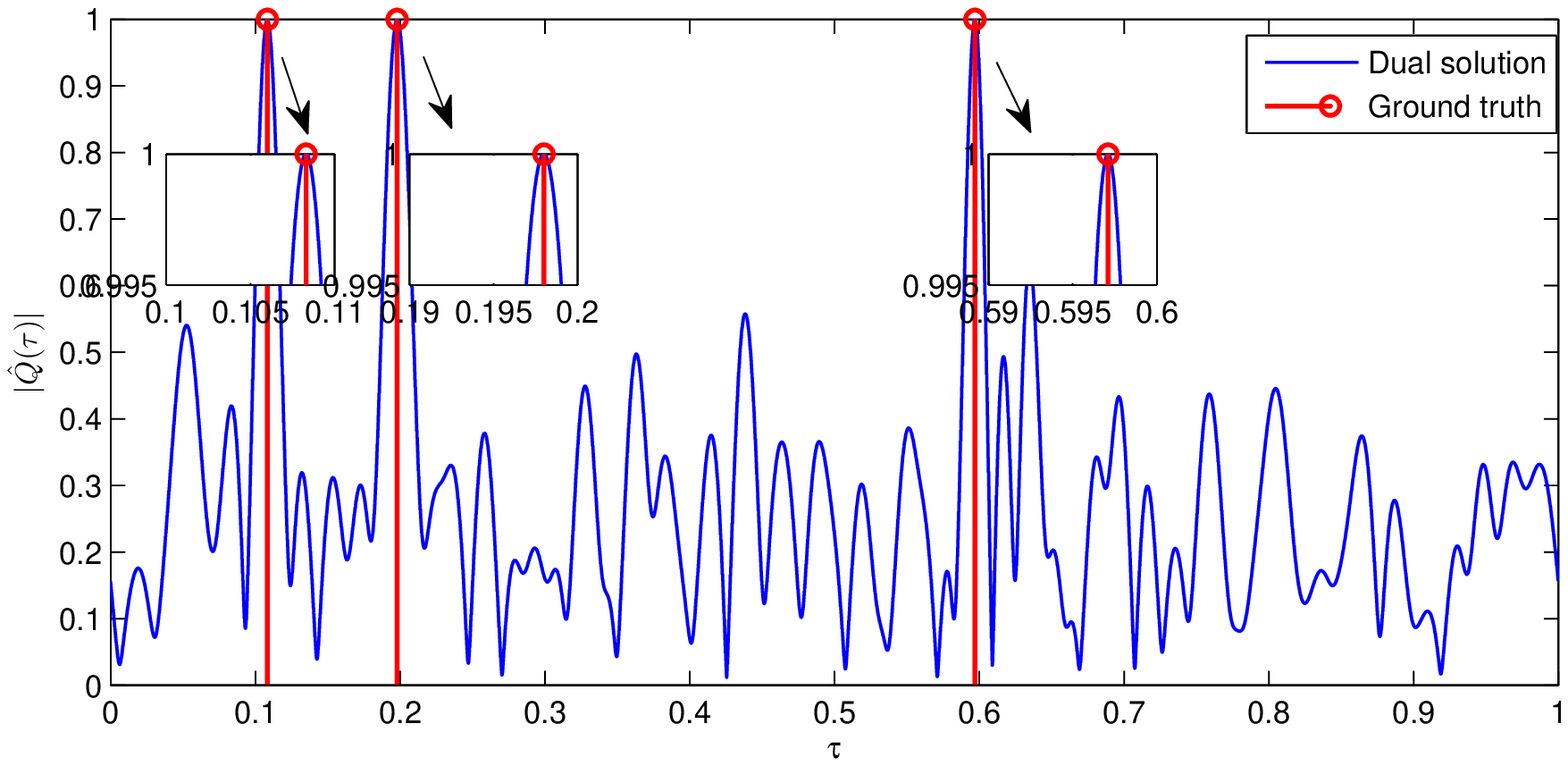} \\
(a) $\hat{P}\left(\tau\right)$ and  $\hat{Q}\left(\tau\right)$, noise-free  \\
\includegraphics[width=0.5\textwidth]{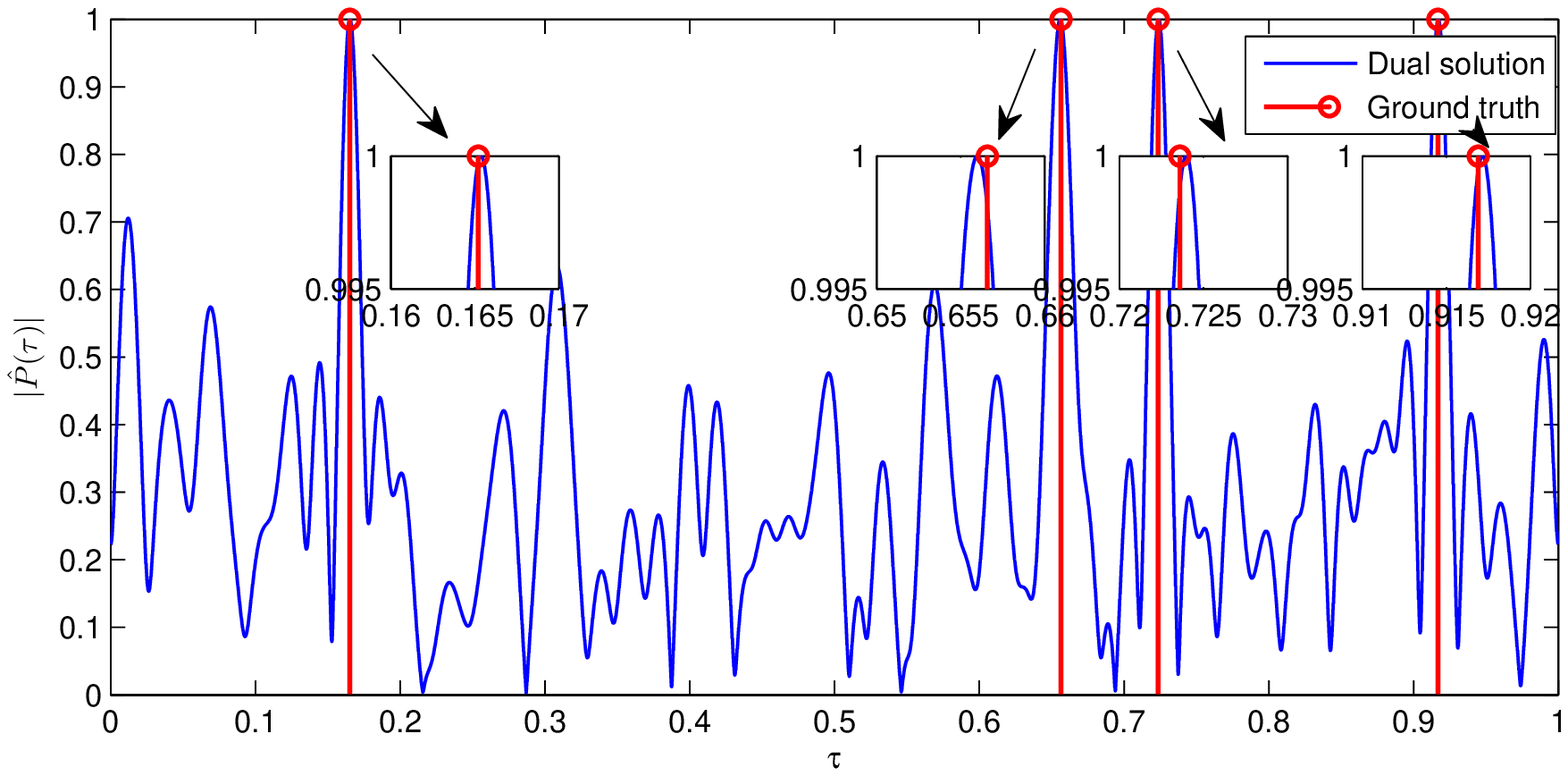} \includegraphics[width=0.5\textwidth]{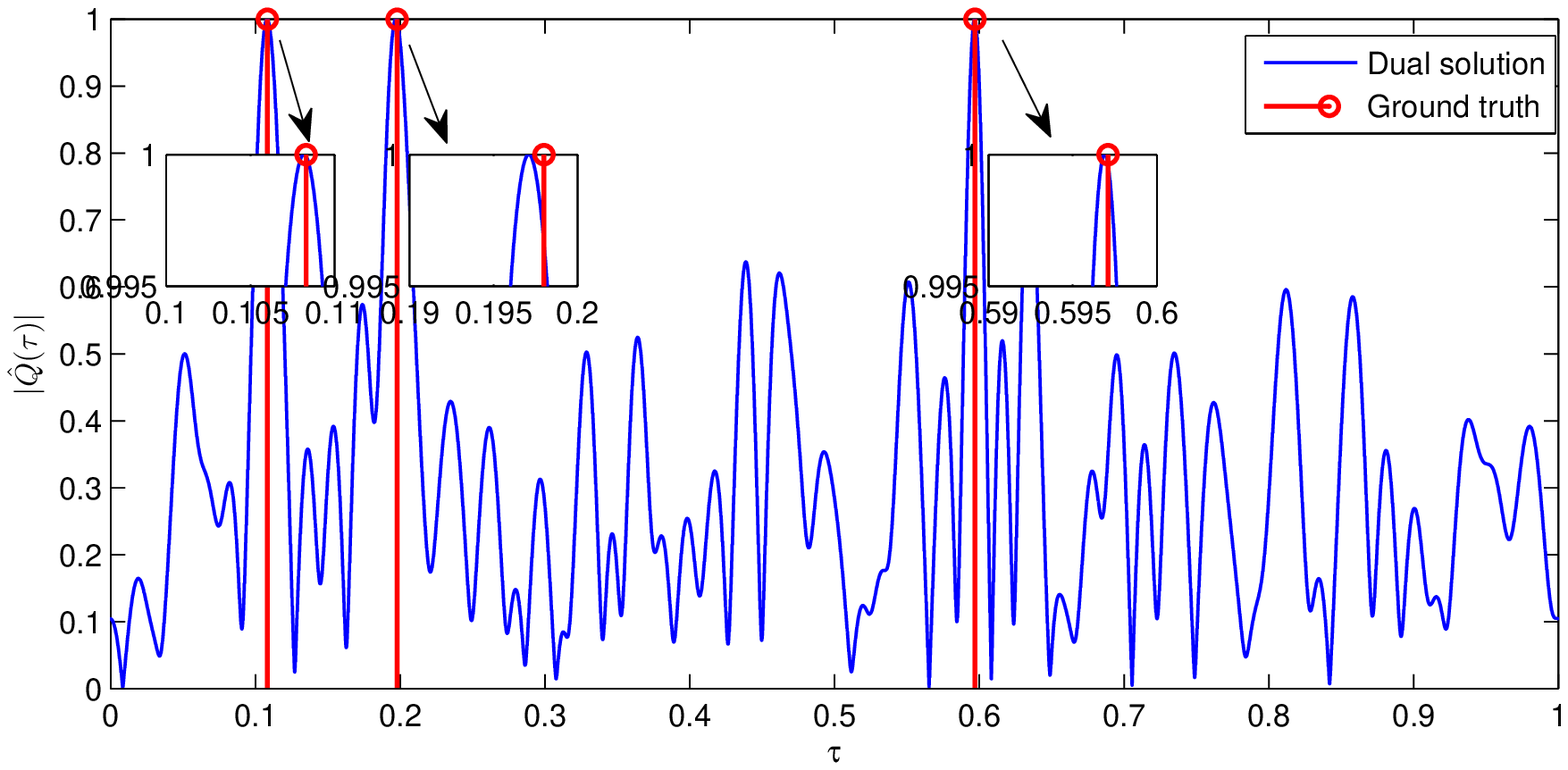} \\
(b) $\hat{P}\left(\tau\right)$ and  $\hat{Q}\left(\tau\right)$, SNR = 16dB  \\
\includegraphics[width=0.5\textwidth]{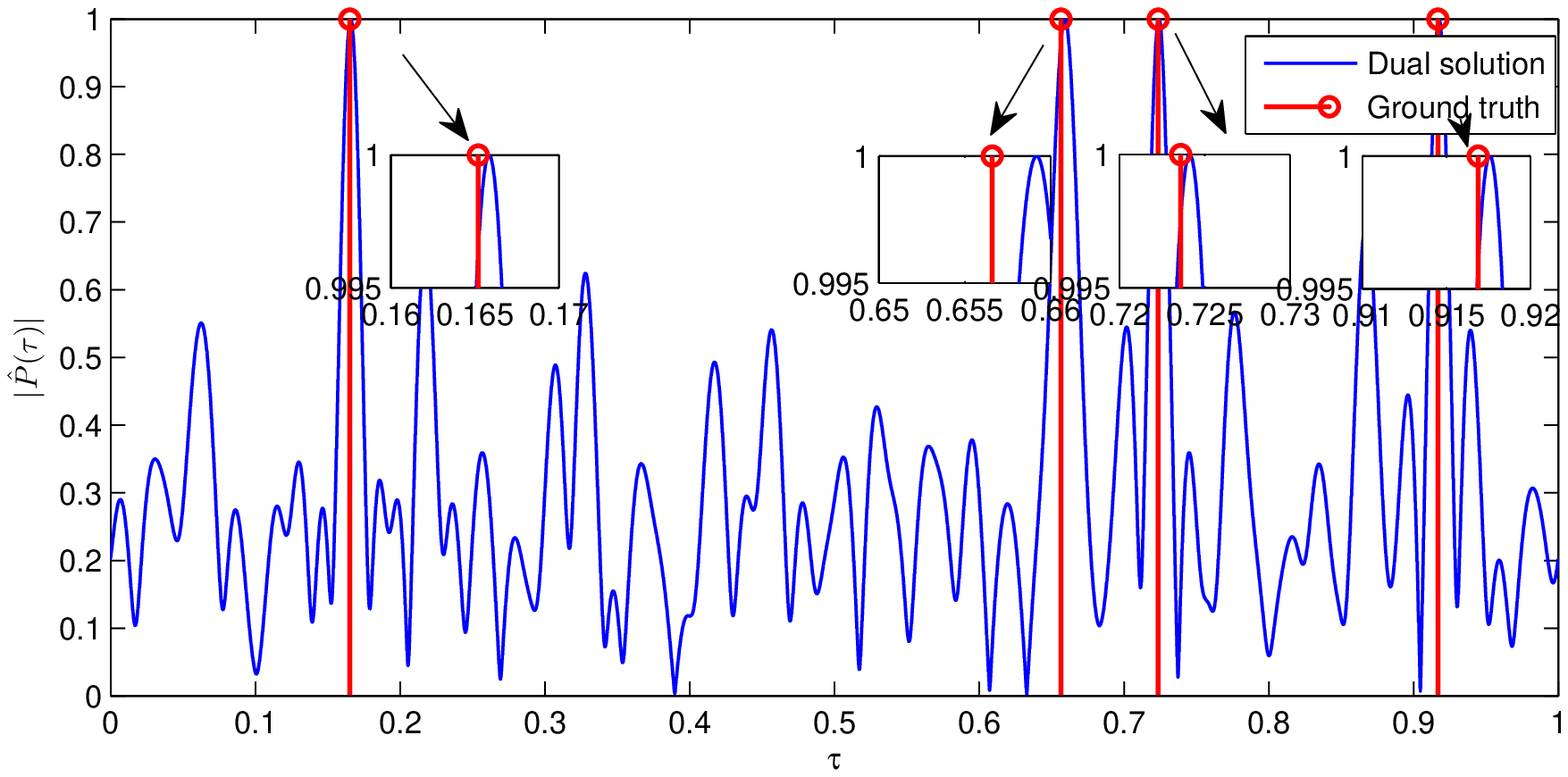}  \includegraphics[width=0.5\textwidth]{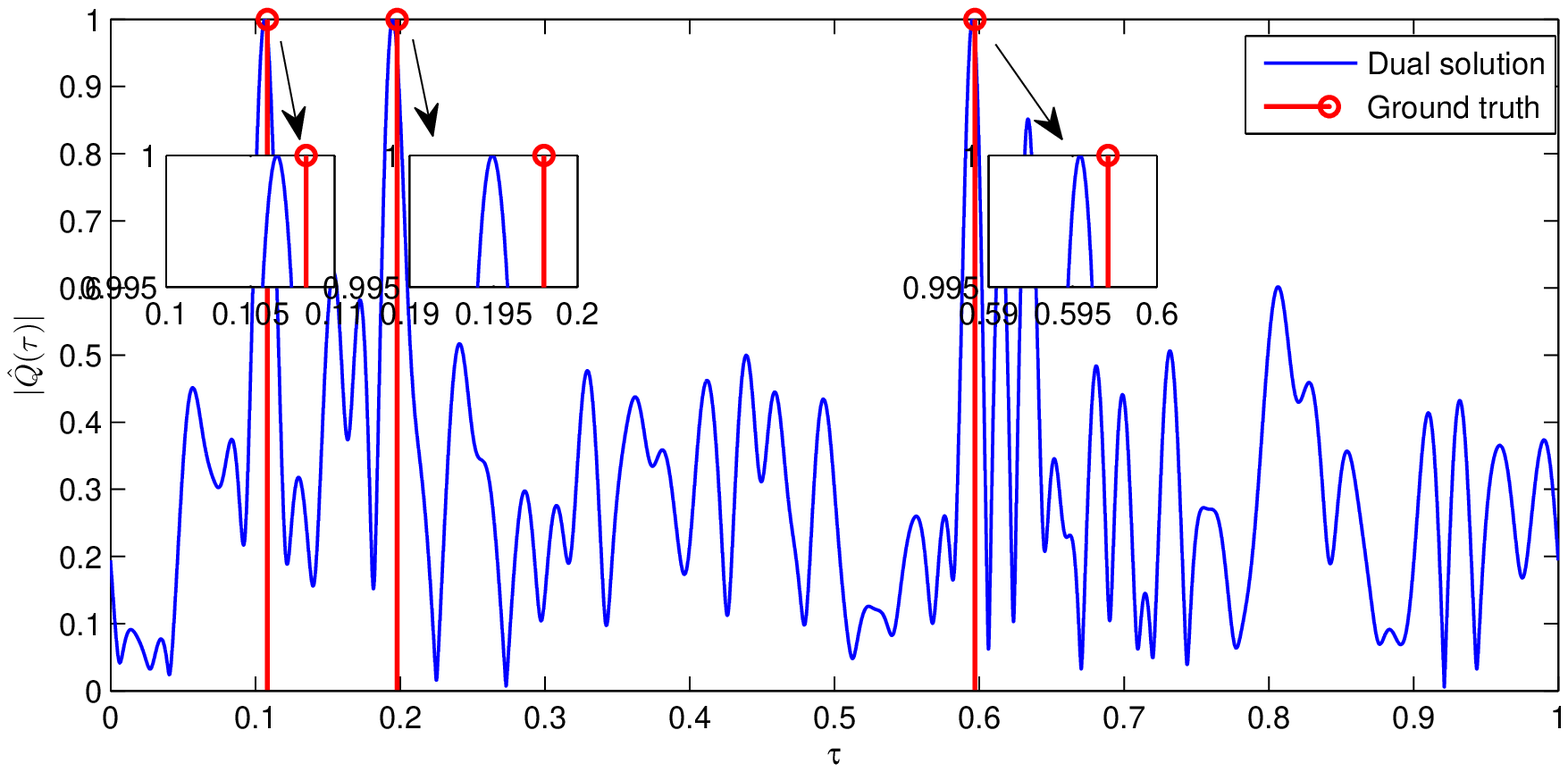} \\
(c)  $\hat{P}\left(\tau\right)$ and $\hat{Q}\left(\tau\right)$, SNR = 5dB 
\end{tabular}
\end{center}
\caption{Point source localization from dual polynomials (a) in absence of noise, (b)  SNR =16dB, and (c) SNR = 5dB, for $M=16$, $K_{1}=4$ and $K_{2}=3$.}\label{fig:fre_recov_dual}
\end{figure}

We then consider the noisy case when the noise is composed of i.i.d. complex Gaussian entries $\mathcal{CN}(0,\sigma^2)$, and set $\lambda_{w} = \sigma  \sqrt{(1+\frac{1}{\log{(4M+1)}}) (4M+1) (\log{\alpha}+\sqrt{2\log{\alpha}}+2+\sqrt{\frac{\pi}{2}})}$, where $\alpha=8\pi (4M+1)\log{(4M+1)}$ based on the discussions in \cite{bhaskar2012atomic, li2016off} or $\lambda_{w} = \sigma \sqrt{4M+1}\sqrt{1.2\log\left(8\pi \left(4M+1\right)\log\left(4M+1\right)\right)}$ for simplicity of use. The amplitudes of the dual polynomials $\hat{P}\left(\tau\right)$ and $\hat{Q}\left(\tau\right)$ are shown in Fig.~\ref{fig:fre_recov_dual} (b) and (c) for SNR = 16 dB and SNR = 5dB, respectively, where the Signal-to-Noise Ratio (SNR) is defined as $\mathrm{SNR}=10\log_{10}{\left(\frac{\left\Vert\bx_{1}^{\star}+\bg\odot\bx_{2}^{\star}\right\Vert_{2}^{2}/\left(4M+1\right)}{\sigma^{2}}\right)}$ dB. It is clear that the source locations can be estimated stably from the dual solutions, and the performance degenerates gracefully with the increase of the noise level.

\subsection{Comparisons with CRB for Point Source Localization}
We further examine the performance of \eqref{algorithm_noisy_model} on estimating the locations of the point sources from noisy measurements by comparing it against the CRB. Specifically, consider the special case with a single point source for each modality, by letting $K_{1}=K_{2}=1$. Denote the point source location in $\bx_{1}^{\star}$ and $\bx_{2}^{\star}$ by $\tau_{1}$ and $\tau_{2}$, respectively. We assume the corresponding amplitude of each point source is known and unity when computing the CRB for estimating $\tau_1$ and $\tau_2$, which can be found as the diagonal entries of the inverse of the following Fisher information matrix:
\begin{equation*}
\boldsymbol{J}\left(\tau_{1},\tau_{2}\right)= \frac{8\pi^{2}}{\sigma^{2}}
\begin{bmatrix}
\sum_{n=-2M}^{2M}n^{2} & \mathrm{Re}\left(\sum_{n=-2M}^{2M}n^{2}\bar{g}_{n} e^{-j2\pi n\left(\tau_{1}-\tau_{2}\right)}\right) \\
\mathrm{Re}\left(\sum_{n=-2M}^{2M}n^{2}\bar{g}_{n}e^{-j2\pi n\left(\tau_{1}-\tau_{2}\right)}\right) & \sum_{n=-2M}^{2M}n^{2}
\end{bmatrix}.
\end{equation*}
For each SNR, we randomly generate $200$ noise realizations and compute the average squared estimate error $\left(\hat{\tau}_{i}-\tau_{i}\right)^{2}$, where $\hat{\tau}_i$ is the dual solution of \eqref{algorithm_noisy_model}, $i=1, 2$. Fig.~\ref{fig:compare_mse_crb} shows the average squared estimate error in comparison with the CRB with respect to SNR when $M=10$ in (a) and $M=16$ in (b). The performance of parameter estimation shows a similar ``thresholding effect'' \cite{tufts1991threshold} as for conventional spectrum estimation algorithms, where the average squared estimate error approaches the CRB as soon as SNR is large enough. Moreover, as we increase $M$, the threshold SNR becomes smaller. Characterizing the exact threshold SNR for AtomicDemix is an interesting future research topic.
\begin{figure}[h]
\begin{center}
\begin{tabular}{cc} 
\hspace{-0.4in}\includegraphics[height=2.5in,width=0.54\textwidth]{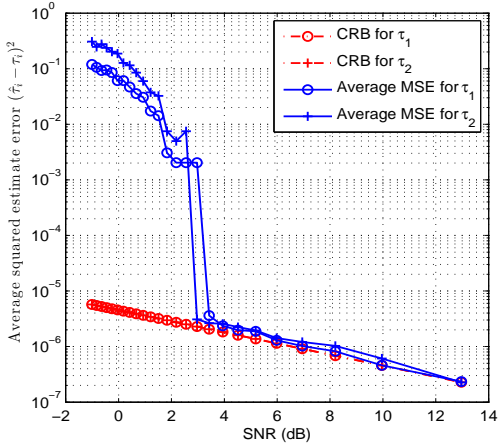} & \hspace{-0.4in}\includegraphics[height=2.5in,width=0.54\textwidth]{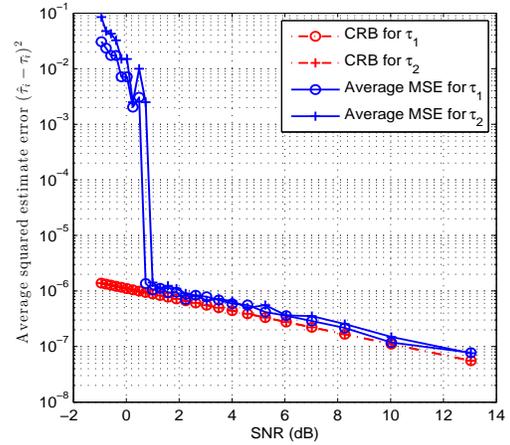} \\
\hspace{-0.4in} (a) $M=10$ & \hspace{-0.4in} (b) $M=16$
\end{tabular}
\end{center}
\caption{The comparisons between the average squared estimate error of point source localization and the corresponding CRB with respect to SNR, when (a) $M=10$, (b) $M=16$.}\label{fig:compare_mse_crb}
\end{figure}

\section{Proof of Theorem~\ref{theorem_main}}\label{sec::proof_theorem_noise_free}

In this section, we proceed to prove Theorem \ref{theorem_main}. We first provide the optimality conditions using dual polynomials to certify the optimality of the solution of \eqref{convex_demixing_rewritten}. Illuminated by \cite{candes2014towards,tang2014CSoffgrid}, where the dual polynomial is constructed using the squared Fej\'er's kernel, we propose a construction of dual polynomials which are composed of a deterministic term and a random perturbation term induced by the interference between modalities. Finally, we show that the constructed dual polynomials satisfy the optimality conditions with high probability when the sample complexity $M$ is large enough.

\subsection{Optimality Conditions using Dual Polynomials}
We first certify the optimality of the primal problem \eqref{convex_demixing_rewritten} using the following proposition whose proof is in Appendix~\ref{proof_dual_certification}.
\begin{prop}\label{dual_certificate}
 $(\boldsymbol{x}_{1}^{\star}, \boldsymbol{x}_{2}^{\star})$ is the unique optimizer of \eqref{convex_demixing_rewritten} if there exists a vector $\boldsymbol{p} = \left[ p_{-2M}, \dots, p_{0}, \dots, p_{2M} \right]^{T}$ such that the dual polynomials $P(\tau)$ and $Q(\tau)$ constructed from it, represented as
\begin{equation}\label{dual_polynomial}
P\left(\tau\right)=\sum_{n=-2M}^{2M} p_{n}e^{j2\pi n\tau} , \quad  Q\left(\tau\right)=\sum_{n=-2M}^{2M}p_{n}\bar{g}_{n}e^{j2\pi n\tau}
\end{equation}
satisfy
\begin{equation} \label{conditions}
\begin{cases}
 P\left(\tau_{1k}\right) = \mathrm{sign}\left(a_{1k}\right), & \forall \tau_{1k} \in \Upsilon_{1} \\
  \left\vert P\left(\tau\right)\right\vert < 1, &\forall \tau\notin \Upsilon_{1} \\
  Q\left(\tau_{2k}\right) = \mathrm{sign}\left(a_{2k}\right), & \forall \tau_{2k} \in \Upsilon_{2} \\
  \left\vert Q\left(\tau\right)\right\vert < 1, &\forall \tau\notin\Upsilon_{2} \\
\end{cases},
\end{equation}
where the sign should be understood as the complex sign.
\end{prop}

\subsection{Constructing the Dual Certificate}\label{dual_construction_original}
Proposition~\ref{dual_certificate} suggests that if we can find a vector $\boldsymbol{p}$ to construct two dual polynomials $P(\tau)$ and $Q(\tau)$ in \eqref{dual_polynomial} that satisfy \eqref{conditions}, AtomicDemix is guaranteed to recover the ground truth. Our construction is inspired by \cite{candes2014towards,tang2014CSoffgrid},  based on use of the squared Fej\'er's kernel. However, since the two dual polynomials are coupled together, the construction is more involved.

Define the squared Fej\'er's kernel \cite{candes2014towards} as
\begin{equation}\label{func_K}
K (\tau )=\frac{1}{M}\sum_{n=-2M}^{2M}s_{n}e^{j2\pi n\tau},
\end{equation}
where $s_{n}=\frac{1}{M}\sum_{i=\max\left\{n-M,-M\right\}}^{\min\left\{n+M,M\right\}}\left(1-\left\vert\frac{i}{M}\right\vert\right)\left(1-\left\vert\frac{n}{M}-\frac{i}{M}\right\vert\right)$ with $\left\vert s_{n} \right\vert \le 1$. The value of $K\left(\tau\right)$ is nonnegative, attaining the peak at $\tau=0$ and decaying to zero rapidly with the increase of $|\tau|$.

We define two functions $K_{g}\left(\tau\right)$ and $K_{\bar{g}}\left(\tau\right)$ respectively as 
\begin{equation}\label{func_Kg}
K_{g}\left(\tau\right)=\frac{1}{M}\sum_{n=-2M}^{2M}s_{n}{g}_{n}e^{j2\pi n\tau}, \quad \mathrm{and}\quad
K_{\bar{g}}\left(\tau\right)=\frac{1}{M}\sum_{n=-2M}^{2M}s_{n} \bar{g}_{n} e^{j2\pi n\tau}.
\end{equation}

We then construct two polynomials $P\left(\tau\right)$ and $Q\left(\tau\right)$ as
\begin{equation}\label{func_Pf}
P\left(\tau\right)= \sum_{k=1}^{K_{1}}\alpha_{1k}K\left(\tau-\tau_{1k}\right)+\sum_{k=1}^{K_{1}}\beta_{1k}K'\left(\tau-\tau_{1k}\right)+\sum_{k=1}^{K_{2}}\alpha_{2k}K_{g}\left(\tau-\tau_{2k}\right)+\sum_{k=1}^{K_{2}}\beta_{2k}K_{g}'\left(\tau-\tau_{2k}\right),
\end{equation} 
and
\begin{equation}\label{func_Qf}
\begin{split}
Q\left(\tau\right)=\sum_{k=1}^{K_{1}}\alpha_{1k}K_{\bar{g}}\left(\tau-\tau_{1k}\right)+\sum_{k=1}^{K_{1}}\beta_{1k}K_{\bar{g}}'\left(\tau-\tau_{1k}\right)+\sum_{k=1}^{K_{2}}\alpha_{2k}K\left(\tau-\tau_{2k}\right)+\sum_{k=1}^{K_{2}}\beta_{2k}K'\left(\tau-\tau_{2k}\right),
\end{split}
\end{equation}
where $\tau_{1k}\in\Upsilon_{1}$ and $\tau_{2k}\in\Upsilon_{2}$. It is straightforward to validate that there exists a corresponding vector $\boldsymbol{p}$ such that \eqref{func_Pf} and \eqref{func_Qf} can be equivalently written in the form of \eqref{dual_polynomial}. Set the coefficients $\boldsymbol{\alpha}_{i}=\left[\alpha_{i1},\dots,\alpha_{iK_{i}}\right]^{T}$, $\boldsymbol{\beta}_{i}=\left[\beta_{i1},\dots,\beta_{iK_{i}}\right]^{T}$, for $i=1,2$ by solving the following equations:
\begin{equation}\label{equ_solve_coeff}
\begin{cases}
P\left(\tau_{1k}\right)  =  \mathrm{sign}\left(a_{1k}\right), &\quad \tau_{1k}\in\Upsilon_{1}, \\
P'\left(\tau_{1k}\right)  =  0, &\quad \tau_{1k}\in\Upsilon_{1}, \\
Q\left(\tau_{2k}\right)  =  \mathrm{sign}\left(a_{2k}\right), &\quad \tau_{2k}\in\Upsilon_{2}, \\
Q'\left(\tau_{2k}\right)  =  0, &\quad \tau_{2k}\in\Upsilon_{2}.
\end{cases}
\end{equation}
The above setting, if exists, immediately satisfies the first and third conditions in \eqref{conditions}. The rest of the proof is then to, under the condition of Theorem~\ref{theorem_main}, guarantee that a solution of \eqref{equ_solve_coeff} exists with high probability, and moreover, when existing, the solution satisfies the second and forth conditions in \eqref{conditions} with high probability, therefore completing the proof.

\begin{example}
Before proceeding, we demonstrate the above dual polynomial construction by an example. Set $M=32$. Let $K_{1}=4$ and $K_{2}=6$. We randomly generate the source locations $\Upsilon_{1}$ and $\Upsilon_{2}$ each satisfying the separation $\Delta \ge 1/M$. The amplitudes of the constructed $P\left(\tau\right)$ and $Q\left(\tau\right)$ are shown in Fig.~\ref{fig:Pf}, which indeed satisfy all the conditions in \eqref{conditions}.
\begin{figure}[h]
\centering
\begin{tabular}{cc}
\includegraphics[width=0.4\textwidth]{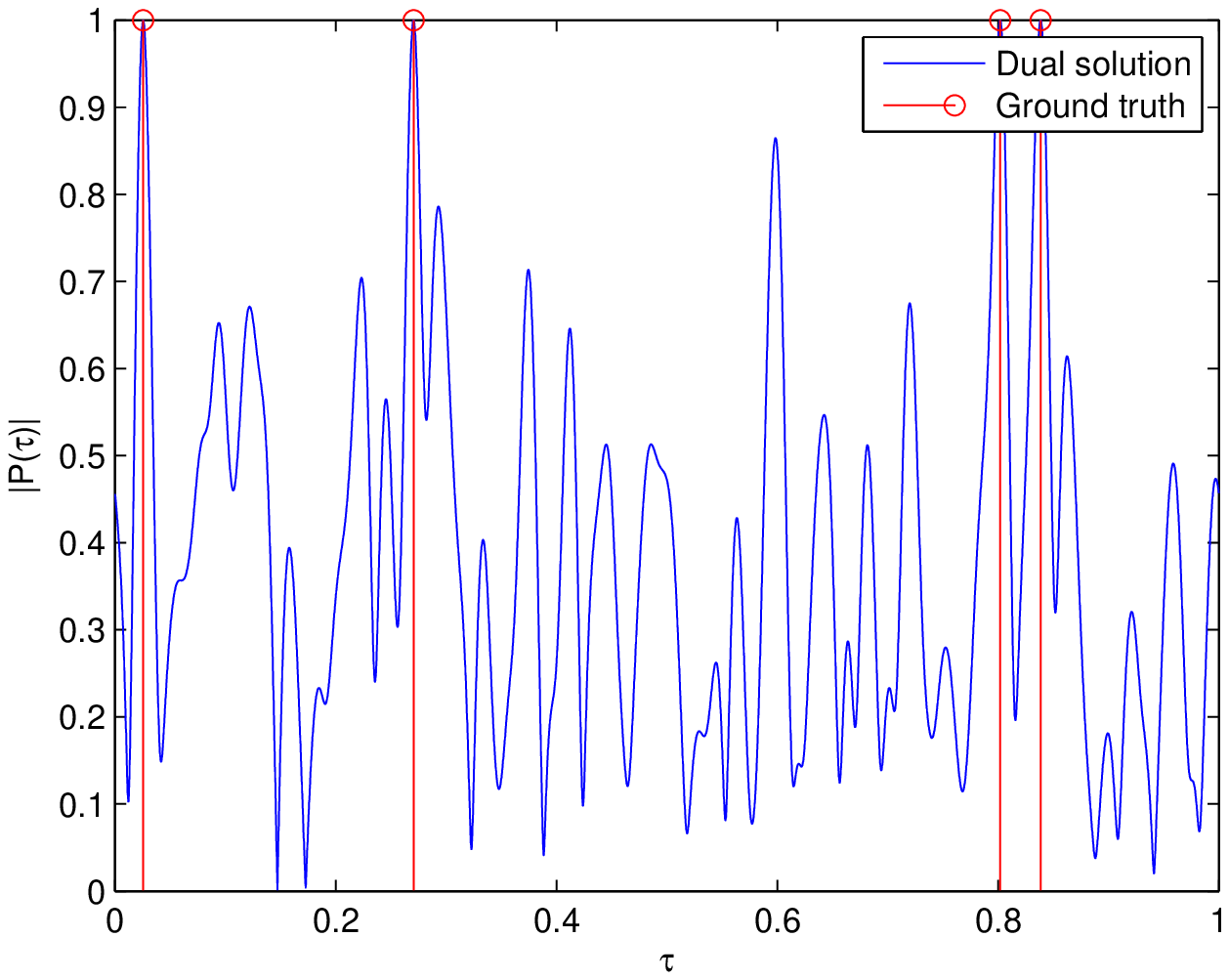} & \includegraphics[width=0.4\textwidth]{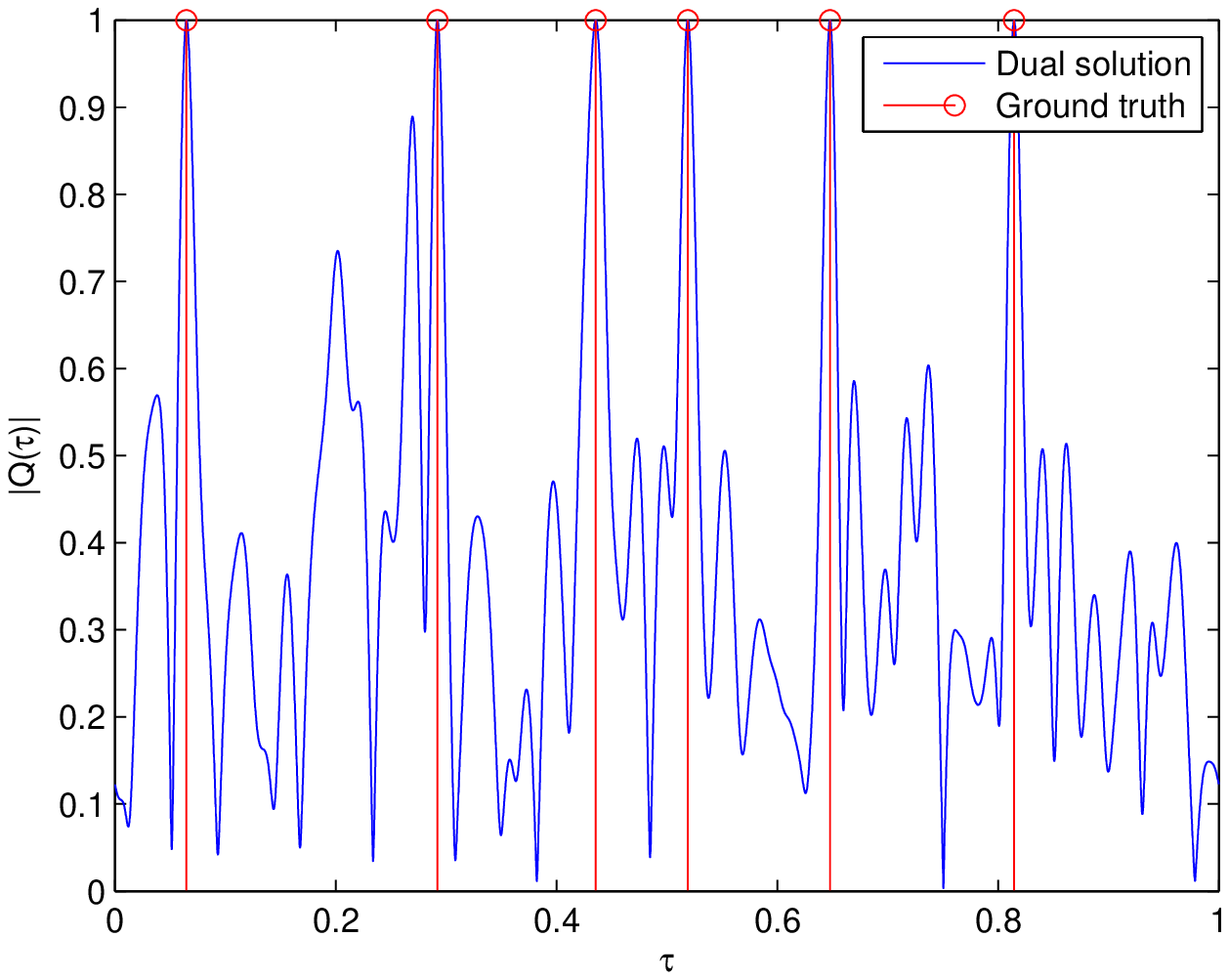} \\
(a) $| P(\tau)|$ & (b) $|Q(\tau)|$
\end{tabular}
\caption{The absolute values of the constructed dual polynomials $|P\left(\tau\right)|$ and $|Q(\tau)|$ following \eqref{equ_solve_coeff} with respect to $\tau\in[0,1)$.}
\label{fig:Pf}
\end{figure}
\end{example}

\subsection{Invertibility of \eqref{equ_solve_coeff}}
We first show that the solution of \eqref{equ_solve_coeff} exists with high probability in this subsection. Let 
\begin{equation*}
\boldsymbol{u}_{i}=\left[\mathrm{sign}\left(a_{i1}\right),\ldots,\mathrm{sign}\left(a_{iK_{i}}\right)\right]^{T},
\end{equation*}
for $i=1,2$. Rewrite \eqref{equ_solve_coeff} into a matrix form as
\begin{equation}\label{coefficient_det}
\begin{bmatrix}
\boldsymbol{W}_{10} & \frac{1}{\sqrt{\left\vert K''\left(0\right)\right\vert}}\boldsymbol{W}_{11} & \boldsymbol{W}_{g0} & \frac{1}{\sqrt{\left\vert K''\left(0\right)\right\vert}}\boldsymbol{W}_{g1}\\
-\frac{1}{\sqrt{\left\vert K''\left(0\right)\right\vert}}\boldsymbol{W}_{11} & -\frac{1}{\left\vert K''\left(0\right)\right\vert}\boldsymbol{W}_{12} & -\frac{1}{\sqrt{\left\vert K''\left(0\right)\right\vert}}\boldsymbol{W}_{g1} & -\frac{1}{\left\vert K''\left(0\right)\right\vert}\boldsymbol{W}_{g2}\\
\boldsymbol{W}_{\bar{g}0} & \frac{1}{\sqrt{\left\vert K''\left(0\right)\right\vert}}\boldsymbol{W}_{\bar{g}1} & \boldsymbol{W}_{20} & \frac{1}{\sqrt{\left\vert K''\left(0\right)\right\vert}}\boldsymbol{W}_{21}\\
-\frac{1}{\sqrt{\left\vert K''\left(0\right)\right\vert}}\boldsymbol{W}_{\bar{g}1} & -\frac{1}{\left\vert K''\left(0\right)\right\vert}\boldsymbol{W}_{\bar{g}2} & -\frac{1}{\sqrt{\left\vert K''\left(0\right)\right\vert}}\boldsymbol{W}_{21} & -\frac{1}{\left\vert K''\left(0\right)\right\vert}\boldsymbol{W}_{22}\\
\end{bmatrix}
\begin{bmatrix}
\boldsymbol{\alpha}_{1}\\
\sqrt{\left\vert K''\left(0\right)\right\vert}\boldsymbol{\beta}_{1}\\
\boldsymbol{\alpha}_{2}\\
\sqrt{\left\vert K''\left(0\right)\right\vert}\boldsymbol{\beta}_{2}\\
\end{bmatrix}
=
\begin{bmatrix}
\boldsymbol{u}_{1}\\
\boldsymbol{0}\\
\boldsymbol{u}_{2}\\
\boldsymbol{0}
\end{bmatrix},
\end{equation}
where $K''(0)$ is a scaler, defined as
\begin{equation}\label{equ_K_twoderiv_zero}
K''\left(0\right) = - \frac{4}{3}\pi^{2}\left(M^{2}-1\right).
\end{equation}
The entries of $\boldsymbol{W}_{1i} \in \mathbb{C}^{K_1\times K_1}$, $\boldsymbol{W}_{gi} \in \mathbb{C}^{K_1\times K_2}$, $\boldsymbol{W}_{\bar{g}i} \in \mathbb{C}^{K_2\times K_1}$, and $\boldsymbol{W}_{2i} \in \mathbb{C}^{K_2\times K_2}$, $i=0,1,2$, are specified respectively as 
\begin{align*} 
\boldsymbol{W}_{1i}\left(l,k\right)&=K^{\left(i\right)}\left(\tau_{1l}-\tau_{1k}\right), \quad \boldsymbol{W}_{gi}\left(l,k\right)=K_{g}^{\left(i\right)}\left(\tau_{1l}-\tau_{2k}\right), \\
 \boldsymbol{W}_{\bar{g}i}\left(l,k\right)&=K_{\bar{g}}^{\left(i\right)}\left(\tau_{2l}-\tau_{1k}\right), \quad \boldsymbol{W}_{2i}\left(l,k\right)=K^{\left(i\right)}\left(\tau_{2l}-\tau_{2k}\right).
\end{align*}

For simplicity, we further introduce the following notations:
\begin{equation*}
\begin{split}
&\boldsymbol{W}_{1}=\begin{bmatrix}\boldsymbol{W}_{10} & \frac{1}{\sqrt{\left\vert K''\left(0\right)\right\vert}}\boldsymbol{W}_{11}\\ -\frac{1}{\sqrt{\left\vert K''\left(0\right)\right\vert}}\boldsymbol{W}_{11} & -\frac{1}{\left\vert K''\left(0\right)\right\vert}\boldsymbol{W}_{12}\end{bmatrix}, \quad \boldsymbol{W}_{g}=\begin{bmatrix}\boldsymbol{W}_{g0} & \frac{1}{\sqrt{\left\vert K''\left(0\right)\right\vert}}\boldsymbol{W}_{g1}\\ -\frac{1}{\sqrt{\left\vert K''\left(0\right)\right\vert}}\boldsymbol{W}_{g1} & -\frac{1}{\left\vert K''\left(0\right)\right\vert}\boldsymbol{W}_{g2}\end{bmatrix},\\
&\boldsymbol{W}_{\bar{g}}=\begin{bmatrix}\boldsymbol{W}_{\bar{g}0} & \frac{1}{\sqrt{\left\vert K''\left(0\right)\right\vert}}\boldsymbol{W}_{\bar{g}1} \\ -\frac{1}{\sqrt{\left\vert K''\left(0\right)\right\vert}}\boldsymbol{W}_{\bar{g}1} & -\frac{1}{\left\vert K''\left(0\right)\right\vert}\boldsymbol{W}_{\bar{g}2}\end{bmatrix}, \quad  \boldsymbol{W}_{2}=\begin{bmatrix}\boldsymbol{W}_{20} & \frac{1}{\sqrt{\left\vert K''\left(0\right)\right\vert}}\boldsymbol{W}_{21}\\
-\frac{1}{\sqrt{\left\vert K''\left(0\right)\right\vert}}\boldsymbol{W}_{21} & -\frac{1}{\left\vert K''\left(0\right)\right\vert}\boldsymbol{W}_{22} \end{bmatrix},
\end{split}
\end{equation*}
and $\boldsymbol{W} = \begin{bmatrix} 
\bW_1 & \bW_g \\
\bW_{\bar{g}} & \bW_2
\end{bmatrix}$. Moreover, we have $\boldsymbol{W}_{g}=\boldsymbol{W}_{\bar{g}}^H$. The diagonal blocks $\bW_i$ of $\boldsymbol{W}$ are deterministic and well-conditioned if the separation $\Delta$ is not so small. This is formalized in the following proposition.
\begin{prop}\cite[Proposition 4.1]{tang2014CSoffgrid}
Suppose $\Delta \ge1/M$, then both $\boldsymbol{W}_{1}$ and $\boldsymbol{W}_{2}$ are invertible and satisfy the following
\begin{align}
\left\Vert\boldsymbol{I}-\boldsymbol{W}_{i}\right\Vert & \le 0.3623,\label{i_min_W}\\
\left\Vert\boldsymbol{W}_{i}\right\Vert &\le 1.3623,\label{single_W}\\
\left\Vert\boldsymbol{W}_{i}^{-1}\right\Vert &\le 1.568\label{inv_W},
\end{align}
for $i=1,2$, where $\left\Vert\cdot\right\Vert$ represents the matrix operator norm.
\end{prop}

The off-diagonal block $\boldsymbol{W}_g$ is a random matrix with respect to $\boldsymbol{g}$, which can be written as
\begin{equation}\label{decomposition_Wg}
\boldsymbol{W}_{g}=\frac{1}{M}\sum_{n=-2M}^{2M}s_{n}g_{n}\boldsymbol{e}_{1}\left(n\right)\boldsymbol{e}_{2}^H\left(n\right)=\sum_{n=-2M}^{2M}\boldsymbol{E}_{n},
\end{equation}
where
\begin{equation}\label{def_e1_e2}
\boldsymbol{e}_{1}\left(n\right)=\begin{bmatrix}e^{j2\pi n\tau_{11}}\\ e^{j2\pi n\tau_{12}}\\ \vdots \\ e^{j2\pi n\tau_{1K_{1}}}\\ -\frac{j2\pi n}{\sqrt{\left\vert K''\left(0\right)\right\vert}}e^{j2\pi n\tau_{11}}\\ -\frac{j2\pi n}{\sqrt{\left\vert K''\left(0\right)\right\vert}}e^{j2\pi n\tau_{12}} \\ \vdots \\ -\frac{j2\pi n}{\sqrt{\left\vert K''\left(0\right)\right\vert}}e^{j2\pi n\tau_{1K_{1}}} \end{bmatrix} \in\mathbb{C}^{2K_{1}}, \quad \boldsymbol{e}_{2}\left(n\right)=\begin{bmatrix}e^{j2\pi n\tau_{21}}\\ e^{j2\pi n\tau_{22}}\\ \vdots \\ e^{j2\pi n\tau_{2K_{2}}}\\ -\frac{j2\pi n}{\sqrt{\left\vert K''\left(0\right)\right\vert}}e^{j2\pi n\tau_{21}}\\ -\frac{j2\pi n}{\sqrt{\left\vert K''\left(0\right)\right\vert}}e^{j2\pi n\tau_{22}} \\ \vdots \\ -\frac{j2\pi n}{\sqrt{\left\vert K''\left(0\right)\right\vert}}e^{j2\pi n\tau_{2K_{2}}} \end{bmatrix} \in\mathbb{C}^{2K_{2}},
\end{equation}
and
\begin{equation}
\boldsymbol{E}_{n}=\frac{1}{M}s_{n}g_{n}\boldsymbol{e}_{1}\left(n\right)\boldsymbol{e}_{2}^H\left(n\right)
\end{equation}
is a zero-mean random matrix with $\mathbb{E}\left[\boldsymbol{E}_{n}\right]=\frac{1}{M}s_{n}\mathbb{E}\left[g_{n}\right]\boldsymbol{e}_{1}\left(n\right)\boldsymbol{e}_{2}^H\left(n\right)=\boldsymbol{0}$ since $\mathbb{E}\left[g_{n}\right] =\mathbb{E}\left[e^{j2\pi\phi_{n}}\right]=0$. We have $\bW_g$ is a sum of independent zero-mean random matrices with $\mathbb{E}\left[\bW_g\right] =0$. The following proposition establishes the spectral norm of $\boldsymbol{W}_{g}$ is bounded with high probability, whose proof is given in Appendix~\ref{proof_invertibility}.
\begin{prop}\label{prop:invertibility}
Assume $M\geq 4$. Let $\delta\in(0,0.6376)$ and $\eta\in(0,1)$, then $\mathbb{P}\left\{\|\bW_g\|\geq \delta\right\}\leq \eta$ provided that 
\begin{equation}\label{sample_complexity}
M\ge \frac{46}{\delta^{2}}K_{\max}\log\left(\frac{2\left(K_{1}+K_{2}\right)}{\eta}\right).
\end{equation}
\end{prop}

Denote the event $\mathcal{E}_{\delta}=\{\|\bW_g \| \leq \delta \}$, which holds with probability at least $1-\eta$ if \eqref{sample_complexity} holds, following \prettyref{prop:invertibility}. Assume $\mathcal{E}_{\delta}$ holds for some $0<\delta<0.6376$ and $\Delta\geq 1/M$, then  
\begin{align*}
\left\Vert\boldsymbol{I}-\boldsymbol{W} \right\Vert&\leq \left\Vert\boldsymbol{I}-\begin{bmatrix}\boldsymbol{W}_{1}&\boldsymbol{0}\\ \boldsymbol{0}& \boldsymbol{W}_{2}\end{bmatrix}\right\Vert+\left\Vert\begin{bmatrix}\boldsymbol{0}&\boldsymbol{W}_{g}\\ \boldsymbol{W}_{\bar{g}}& \boldsymbol{0}\end{bmatrix}\right\Vert\\
&\leq \max_{i=1,2} \| \bI -\bW_i \| + \|\bW_g\| \\
&\le 0.3623 + \delta < 1,
\end{align*}
which yields that $\bW$ is invertible under $\mathcal{E}_{\delta}$. Equivalently, under $\mathcal{E}_{\delta}$ the solution to \eqref{equ_solve_coeff} exists. Write $\boldsymbol{W}^{-1}$ as 
\begin{equation*}
\boldsymbol{W}^{-1}=\begin{bmatrix}\boldsymbol{L}_{1} & \boldsymbol{R}_{1} & \boldsymbol{L}_{g} & \boldsymbol{R}_{g}\\ \boldsymbol{L}_{\bar{g}} & \boldsymbol{R}_{\bar{g}} & \boldsymbol{L}_{2} & \boldsymbol{R}_{2}\end{bmatrix}, 
\end{equation*}
where $\boldsymbol{L}_i, \boldsymbol{R}_i  \in \mathbb{C}^{2K_i \times K_i}$ for $i=1,2$, $\boldsymbol{L}_g, \boldsymbol{R}_g \in \mathbb{C}^{2K_1\times K_2}$ and $\boldsymbol{L}_{\bar{g}}, \boldsymbol{R}_{\bar{g}} \in \mathbb{C}^{2K_2\times K_1}$. We can then invert \eqref{coefficient_det} and obtain
\begin{equation}\label{dual_polynomial_coefficients}
\begin{bmatrix}  
\boldsymbol{\alpha}_{1}\\
\sqrt{\left\vert K''\left(0\right)\right\vert}\boldsymbol{\beta}_{1}\\
\boldsymbol{\alpha}_{2}\\
\sqrt{\left\vert K''\left(0\right)\right\vert}\boldsymbol{\beta}_{2}\\
\end{bmatrix}
=\boldsymbol{W}^{-1}
\begin{bmatrix}
\boldsymbol{u}_{1}\\
\boldsymbol{0}\\
\boldsymbol{u}_{2}\\
\boldsymbol{0}
\end{bmatrix},
\end{equation}
which gives 
\begin{equation*}
\begin{bmatrix}
\boldsymbol{\alpha}_{1}\\
\sqrt{\left\vert K''\left(0\right)\right\vert}\boldsymbol{\beta}_{1}\\
\end{bmatrix}=\boldsymbol{L}_{1}\boldsymbol{u}_{1}+\boldsymbol{L}_{g}\boldsymbol{u}_{2}
\quad
\mathrm{and} 
\quad 
\begin{bmatrix}
\boldsymbol{\alpha}_{2}\\
\sqrt{\left\vert K''\left(0\right)\right\vert}\boldsymbol{\beta}_{2}\\
\end{bmatrix}=\boldsymbol{L}_{\bar{g}}\boldsymbol{u}_{1}+\boldsymbol{L}_{2}\boldsymbol{u}_{2}.
\end{equation*} 

\subsection{Bounding the Dual Polynomials}
The rest of the proof is then given \eqref{dual_polynomial_coefficients}, we need to verify that 
$\left\vert P\left(\tau\right)\right\vert < 1$, $\forall \tau\notin \Upsilon_{1}$ and similarly, $\left\vert Q\left(\tau\right)\right\vert < 1, \forall \tau\notin \Upsilon_{2}$. Since the expressions for $P(\tau)$ and $Q(\tau)$ are very similar, it is sufficient to only establish the above for $P(\tau)$.

Recall the form of $P\left(\tau\right)$ in \eqref{func_Pf}, the $l$th derivative of $P(\tau)$ can be represented as 
\begin{equation}
P^{\left(l\right)}\left(\tau\right)=\sum_{k=1}^{K_{1}}\alpha_{1k}K^{\left(l\right)}\left(\tau-\tau_{1k}\right)+\sum_{k=1}^{K_{1}}\beta_{1k}K^{\left(l+1\right)}\left(\tau-\tau_{1k}\right)+\sum_{k=1}^{K_{2}}\alpha_{2k}K_{g}^{\left(l\right)}\left(\tau-\tau_{2k}\right)+\sum_{k=1}^{K_{2}}\beta_{2k}K_{g}^{\left(l+1\right)}\left(\tau-\tau_{2k}\right),
\end{equation} 
which can be rewritten as
\begin{align}
\frac{1}{\sqrt{\left\vert K''\left(0\right)\right\vert}^{l}}P^{\left(l\right)}\left(\tau\right)&
=\sum_{k=1}^{K_{1}}\alpha_{1k}\frac{1}{\sqrt{\left\vert K''\left(0\right)\right\vert}^{l}}K^{\left(l\right)}\left(\tau-\tau_{1k}\right)+\sum_{k=1}^{K_{1}}\sqrt{\left\vert K''\left(0\right)\right\vert}\beta_{1k}\frac{1}{\sqrt{\left\vert K''\left(0\right)\right\vert}^{l+1}}K^{\left(l+1\right)}\left(\tau-\tau_{1k}\right) \nonumber \\
& +\sum_{k=1}^{K_{2}}\alpha_{2k}\frac{1}{\sqrt{\left\vert K''\left(0\right)\right\vert}^{l}}K_{g}^{\left(l\right)}\left(\tau-\tau_{2k}\right)+\sum_{k=1}^{K_{2}}\sqrt{\left\vert K''\left(0\right)\right\vert}\beta_{2k}\frac{1}{\sqrt{\left\vert K''\left(0\right)\right\vert}^{l+1}}K_{g}^{\left(l+1\right)}\left(\tau-\tau_{2k}\right) \nonumber \\
&=\boldsymbol{v}_{1l}^{H}\left(\tau\right)\begin{bmatrix}\boldsymbol{\alpha}_{1}\\ \sqrt{\left\vert K''\left(0\right)\right\vert}\boldsymbol{\beta}_{1}\end{bmatrix} + \boldsymbol{v}_{2l}^{H}\left(\tau\right)\begin{bmatrix}\boldsymbol{\alpha}_{2}\\ \sqrt{\left\vert K''\left(0\right)\right\vert}\boldsymbol{\beta}_{2}\end{bmatrix},
\end{align} 
where
\begin{equation*}
\bar{\boldsymbol{v}}_{1l}\left(\tau\right)=\frac{1}{\sqrt{\left\vert K''\left(0\right)\right\vert}^{l}}\begin{bmatrix}K^{\left(l\right)}\left(\tau-\tau_{11}\right) \\ K^{\left(l\right)}\left(\tau-\tau_{12}\right) \\ \vdots\\K^{\left(l\right)}\left(\tau-\tau_{1K_{1}}\right) \\ \frac{1}{\sqrt{\left\vert K''\left(0\right)\right\vert}}K^{\left(l+1\right)}\left(\tau-\tau_{11}\right) \\\frac{1}{\sqrt{\left\vert K''\left(0\right)\right\vert}}K^{\left(l+1\right)}\left(\tau-\tau_{12}\right) \\\vdots\\\frac{1}{\sqrt{\left\vert K''\left(0\right)\right\vert}}K^{\left(l+1\right)}\left(\tau-\tau_{1K_{1}}\right) \end{bmatrix},
\bar{\boldsymbol{v}}_{2l}\left(\tau\right)=\frac{1}{\sqrt{\left\vert K''\left(0\right)\right\vert}^{l}}\begin{bmatrix}K_{g}^{\left(l\right)}\left(\tau-\tau_{21}\right) \\ K_{g}^{\left(l\right)}\left(\tau-\tau_{22}\right) \\ \vdots\\K_{g}^{\left(l\right)}\left(\tau-\tau_{2K_{2}}\right) \\ \frac{1}{\sqrt{\left\vert K''\left(0\right)\right\vert}}K_{g}^{\left(l+1\right)}\left(\tau-\tau_{21}\right)  \\\frac{1}{\sqrt{\left\vert K''\left(0\right)\right\vert}}K_{g}^{\left(l+1\right)}\left(\tau-\tau_{22}\right) \\\vdots\\\frac{1}{\sqrt{\left\vert K''\left(0\right)\right\vert}}K_{g}^{\left(l+1\right)}\left(\tau-\tau_{2K_{2}}\right) \end{bmatrix},
\end{equation*}
and $K''\left(0\right)$ is the scaler defined in \eqref{equ_K_twoderiv_zero}. Using the forms of $K(\tau)$ and $K_g(\tau)$, we can rewrite the above as 
\begin{equation*}
\begin{split}
&\boldsymbol{v}_{1l}\left(\tau\right)=\frac{1}{M}\sum_{n=-2M}^{2M}s_{n} \left(\frac{-j2\pi n}{\sqrt{\left\vert K''\left(0\right)\right\vert}}\right)^{l} e^{-j2\pi n\tau}\boldsymbol{e}_{1}\left(n\right),\\
&\boldsymbol{v}_{2l}\left(\tau\right)=\frac{1}{M}\sum_{n=-2M}^{2M}s_{n}g_{n} \left(\frac{-j2\pi n}{\sqrt{\left\vert K''\left(0\right)\right\vert}}\right)^{l} e^{-j2\pi n\tau}\boldsymbol{e}_{2}\left(n\right),
\end{split}
\end{equation*}
where $\boldsymbol{e}_1(n)$ and $\boldsymbol{e}_2(n)$ are defined in \eqref{def_e1_e2}. Then $\frac{1}{\sqrt{\left\vert K''\left(0\right)\right\vert}^{l}}P^{\left(l\right)}\left(\tau\right)$ can be rewritten as 
\begin{align}
\frac{1}{\sqrt{\left\vert K''\left(0\right)\right\vert}^{l}}P^{\left(l\right)}\left(\tau\right)&=\boldsymbol{v}_{1l}^{H}\left(\tau\right)\left(\boldsymbol{L}_{1}\boldsymbol{u}_{1}+\boldsymbol{L}_{g}\boldsymbol{u}_{2}\right)+\boldsymbol{v}_{2l}^{H}\left(\tau\right)\left(\boldsymbol{L}_{\bar{g}}\boldsymbol{u}_{1}+\boldsymbol{L}_{2}\boldsymbol{u}_{2}\right) \label{plug1}\\
&=\langle\boldsymbol{u}_{1},\boldsymbol{L}_{1}^{H}\boldsymbol{v}_{1l}\left(\tau\right)\rangle + \langle\boldsymbol{u}_{2},\boldsymbol{L}_{g}^{H}\boldsymbol{v}_{1l}\left(\tau\right)\rangle + \langle\boldsymbol{u}_{1},\boldsymbol{L}_{\bar{g}}^{H}\boldsymbol{v}_{2l}\left(\tau\right)\rangle + \langle\boldsymbol{u}_{2},\boldsymbol{L}_{2}^{H}\boldsymbol{v}_{2l}\left(\tau\right)\rangle ,\label{plug2}
\end{align}
where \eqref{plug1} follows from \eqref{dual_polynomial_coefficients}. Let
\begin{equation*}
\boldsymbol{W}_{\mu} = \mathbb{E}\left[\bW\right] =\begin{bmatrix}\mathbb{E}\left[\boldsymbol{W}_{1}\right]&\mathbb{E}\left[\boldsymbol{W}_{g}\right]\\ \mathbb{E}\left[\boldsymbol{W}_{\bar{g}}\right]& \mathbb{E}\left[\boldsymbol{W}_{2}\right]\end{bmatrix}=\begin{bmatrix}\boldsymbol{W}_{1}&\boldsymbol{0}\\ \boldsymbol{0}& \boldsymbol{W}_{2}\end{bmatrix},
\end{equation*}
and 
\begin{equation*}
\boldsymbol{W}_{\mu}^{-1}=\begin{bmatrix}\boldsymbol{W}_{1}^{-1}&\boldsymbol{0}\\ \boldsymbol{0}& \boldsymbol{W}_{2}^{-1}\end{bmatrix}=\begin{bmatrix}\boldsymbol{L}_{\mu 1} & \boldsymbol{R}_{\mu 1} & \boldsymbol{0} & \boldsymbol{0}\\ \boldsymbol{0} & \boldsymbol{0} & \boldsymbol{L}_{\mu 2} & \boldsymbol{R}_{\mu 2}\end{bmatrix},
\end{equation*}
where $\boldsymbol{L}_{\mu i} \in \mathbb{C}^{2K_i\times K_i}$ and $\boldsymbol{R}_{\mu i} \in \mathbb{C}^{2K_i\times K_i}$, $i=1,2$. We can then further rewrite \eqref{plug2} as
\begin{align}
\frac{1}{\sqrt{\left\vert K''\left(0\right)\right\vert}^{l}}P^{\left(l\right)}\left(\tau\right)&=\langle\boldsymbol{u}_{1},\boldsymbol{L}_{\mu 1}^{H}\boldsymbol{v}_{1l}\left(\tau\right)\rangle  + \langle\boldsymbol{u}_{1},\left(\boldsymbol{L}_{1}-\boldsymbol{L}_{\mu 1}\right)^{H}\boldsymbol{v}_{1l}\left(\tau\right)\rangle + \langle\boldsymbol{u}_{2},\boldsymbol{L}_{g}^{H}\boldsymbol{v}_{1l}\left(\tau\right)\rangle \nonumber\\
&\quad + \langle\boldsymbol{u}_{1},\boldsymbol{L}_{\bar{g}}^{H}\boldsymbol{v}_{2l}\left(\tau\right)\rangle + \langle\boldsymbol{u}_{2},\boldsymbol{L}_{2}^{H}\boldsymbol{v}_{2l}\left(\tau\right)\rangle.\label{derivative_P}
\end{align}
Denote 
\begin{equation*}
\frac{1}{\sqrt{\left\vert K''\left(0\right)\right\vert}^{l}}P_{\mu}^{\left(l\right)}\left(\tau\right) =  \langle\boldsymbol{u}_{1},\boldsymbol{L}_{\mu 1}^{H}\boldsymbol{v}_{1l}\left(\tau\right)\rangle.
\end{equation*}
Our proof proceeds in the following steps:
\begin{itemize}
\item {\textbf {Step 1:}} show that $\frac{1}{\sqrt{\left\vert K''\left(0\right)\right\vert}^{l}}P^{\left(l\right)}\left(\tau\right)$ is bounded around $\frac{1}{\sqrt{\left\vert K''\left(0\right)\right\vert}^{l}}P_{\mu}^{\left(l\right)}\left(\tau\right)$ for a set of grid points $\Upsilon_{\mathrm{grid}}$;
\item {\textbf {Step 2:}} show that $\frac{1}{\sqrt{\left\vert K''\left(0\right)\right\vert}^{l}}P^{\left(l\right)}\left(\tau\right)$ is uniformly bounded around $\frac{1}{\sqrt{\left\vert K''\left(0\right)\right\vert}^{l}}P_{\mu}^{\left(l\right)}\left(\tau\right)$ for all $\tau\in[0,1)$;
\item {\textbf {Step 3:}} finally, show that $\left\vert P\left(\tau\right)\right\vert < 1$, $\forall \tau\notin \Upsilon_{1}$.
\end{itemize}

\subsubsection{Proof of Step 1} 
Here the goal is to bound the last four residual terms in \eqref{derivative_P} with high probability on a set of uniform grid points $\tau \in \Upsilon_{\mathrm{grid}}$ from $[0,1)$ whose size will be specified later. We first record the following useful lemma whose proof is given in Appendix~\ref{proof_bound_inver}.
\begin{lemma}\label{bound_inver}
Under the event $\mathcal{E}_\delta$ for some $\delta\in(0,1/4]$, we have
\begin{align*}
&\left\Vert\boldsymbol{L}_{i}\right\Vert\le 2\left\Vert\boldsymbol{W}_{\mu}^{-1}\right\Vert, \quad\mathrm{for}\quad i=1,2,\\
&\left\Vert\boldsymbol{L}_{i}-\boldsymbol{L}_{\mu i}\right\Vert\le 2\left\Vert\boldsymbol{W}_{\mu}^{-1}\right\Vert^{2}\delta, \quad \mathrm{for}\quad i=1,2,\\
&\left\Vert\boldsymbol{L}_{g}\right\Vert\le 2\left\Vert\boldsymbol{W}_{\mu}^{-1}\right\Vert^{2}\delta \le 0.8\left\Vert\boldsymbol{W}_{\mu}^{-1}\right\Vert,\\
&\left\Vert\boldsymbol{L}_{\bar{g}}\right\Vert\le 2\left\Vert\boldsymbol{W}_{\mu}^{-1}\right\Vert^{2}\delta \le 0.8\left\Vert\boldsymbol{W}_{\mu}^{-1}\right\Vert.
\end{align*}
\end{lemma}

When the signs of the coefficients $a_{ik}$'s are arbitrary, the last four terms in \eqref{derivative_P} can be bounded by 
\begin{align*}
\left|\langle\boldsymbol{u}_{1},\left(\boldsymbol{L}_{1}-\boldsymbol{L}_{\mu 1}\right)^{H}\boldsymbol{v}_{1l}\left(\tau\right)\rangle\right| &\le \left\Vert \boldsymbol{u}_{1}\right\Vert_{2}\left\Vert \left(\boldsymbol{L}_{1}-\boldsymbol{L}_{\mu 1}\right)^{H}\boldsymbol{v}_{1l}\left(\tau\right) \right\Vert_{2}\le C_{1}\sqrt{K_{1}}\delta,\\
\left|\langle\boldsymbol{u}_{2},\boldsymbol{L}_{g}^{H}\boldsymbol{v}_{1l}\left(\tau\right)\rangle\right| & \le \left\Vert \boldsymbol{u}_{2}\right\Vert_{2}\left\Vert \boldsymbol{L}_{g}^{H}\boldsymbol{v}_{1l}\left(\tau\right) \right\Vert_{2}\le C_{2}\sqrt{K_{2}}\delta, \\
|\langle\boldsymbol{u}_{1},\boldsymbol{L}_{\bar{g}}^{H}\boldsymbol{v}_{2l}\left(\tau\right)\rangle| &\le \left\Vert \boldsymbol{u}_{1} \right\Vert_{2}\left\Vert \boldsymbol{L}_{\bar{g}}^{H}\boldsymbol{v}_{2l}\left(\tau\right) \right\Vert_{2},\\
 |\langle\boldsymbol{u}_{2},\boldsymbol{L}_{2}^{H}\boldsymbol{v}_{2l}\left(\tau\right)\rangle| &\le \left\Vert \boldsymbol{u}_{2} \right\Vert_{2}\left\Vert \boldsymbol{L}_{2}^{H}\boldsymbol{v}_{2l}\left(\tau\right) \right\Vert_{2},
 \end{align*}
where the last steps of the first two inequalities follow from Lemma~\ref{bound_inver}, and $\left\Vert\boldsymbol{v}_{1l}\left(\tau\right)\right\Vert_{2}\le C$ for some numerical constant $C$ \cite[Lemma 4.9]{tang2014CSoffgrid}. By setting $\delta$ properly and we can obtain the bound on $M$ using Proposition~\ref{prop:invertibility}, Lemma 4.6 and Lemma 4.7 in \cite{tang2014CSoffgrid}. When the signs of the coefficients $a_{ik}$'s are random, we can provide a tighter bound by applying the Hoeffding's inequality, which follows similarly as the proof of \cite[Lemma 4.8 and 4.9]{tang2014CSoffgrid}. We have the following proposition.
\begin{prop}\label{bound_L1v1}
Suppose $\Delta\ge 1/M$. There exists a numerical constant $C$ such that 
\begin{equation*}
M\ge C\max\left\{\log^{2}\left(\frac{\left\vert \Upsilon_{\mathrm{grid}}\right\vert}{\eta}\right), \frac{K_{\max}}{\epsilon^{2}}\log\left(\frac{\left\vert \Upsilon_{\mathrm{grid}}\right\vert}{\eta}\right),  \frac{K_{\max}^{2}}{\epsilon^{2}}\log{\left(\frac{K_{1}+K_{2}}{\eta}\right)} \right\},
\end{equation*}
or additionally, if the signs of the coefficients $a_{ik}$'s are i.i.d. generated from a symmetric distribution on the complex unit circle, there exists a numerical constant $C$ such that 
\begin{equation*}
M\ge C \max\left\{  \frac{1}{\epsilon^{2}}\log^2\left(\frac{\left\vert\Upsilon_{\mathrm{grid}}\right\vert}{\eta}\right),  \frac{1}{\epsilon^{2}}K_{\max}\log{\left(\frac{K_{1}+K_{2}}{\eta}\right)}\log{\left(\frac{\left\vert\Upsilon_{\mathrm{grid}}\right\vert}{\eta}\right)}  \right\},
\end{equation*}
where $\left\vert\Upsilon_{\mathrm{grid}}\right\vert$ is the grid size, then we have
\begin{align*}
 \sup_{\tau_{d}\in\Upsilon_{\mathrm{grid}}}\left\vert\langle\boldsymbol{u}_{1},\left(\boldsymbol{L}_{1}-\boldsymbol{L}_{\mu 1}\right)^{H}\boldsymbol{v}_{1l}\left(\tau_{d}\right)\rangle\right\vert & < \epsilon, \ l=0,1,2,3; \\
 \sup_{\tau_{d}\in\Upsilon_{\mathrm{grid}}}\left\vert\langle\boldsymbol{u}_{2},\boldsymbol{L}_{g}^{H}\boldsymbol{v}_{1l}\left(\tau_{d}\right)\rangle\right\vert &< \epsilon, \ l=0,1,2,3; \\
 \sup_{\tau_{d}\in\Upsilon_{\mathrm{grid}}}\left\vert\langle\boldsymbol{u}_{1},\boldsymbol{L}_{\bar{g}}^{H}\boldsymbol{v}_{2l}\left(\tau_{d}\right)\rangle\right\vert &< \epsilon, \ l=0,1,2,3;\\
 \sup_{\tau_{d}\in\Upsilon_{\mathrm{grid}}}\left\vert\langle\boldsymbol{u}_{2},\boldsymbol{L}_{2}^{H}\boldsymbol{v}_{2l}\left(\tau_d\right)\rangle\right\vert & < \epsilon, \ l=0,1,2,3, 
\end{align*}
hold with probability at least $ 1-8\eta$. 
\end{prop}
Denote the event
\begin{equation*}
\mathcal{E}_{1}=\left\{\sup_{\tau_{d}\in\Upsilon_{\mathrm{grid}}}\left\vert\frac{1}{\sqrt{\left\vert K''\left(0\right)\right\vert}^{l}}P^{\left(l\right)}\left(\tau_{d}\right)-\frac{1}{\sqrt{\left\vert K''\left(0\right)\right\vert}^{l}}P_{\mu}^{\left(l\right)}\left(\tau_{d}\right)\right\vert\le\frac{\epsilon}{3}, \ l=0,1,2,3\right\},
\end{equation*}
for some $\epsilon>0$. Then by rescaling the constants, it is straightforward that $\mathcal{E}_1$ holds with probability at least $1-\eta$ as soon as the conditions in Proposition~\ref{bound_L1v1} are met.

\subsubsection{Proof of Step 2}
We have shown that the differences between $\frac{1}{\sqrt{\left\vert K''\left(0\right)\right\vert}^{l}}P^{\left(l\right)}\left(\tau\right)$ and $\frac{1}{\sqrt{\left\vert K''\left(0\right)\right\vert}^{l}}P_{\mu}^{\left(l\right)}\left(\tau\right)$ are bounded on a finite grid. In this step we extend this statement to the continuous domain $\tau\in[0,1)$ by assigning the size of $\Upsilon_{\mathrm{grid}}$ properly. This is given in the following proposition whose proof is given in Appendix~\ref{proof_bound_continuous}.
\begin{prop}\label{bound_continuous}
Suppose $\Delta\ge 1/M$. There exists a numerical constant $C$ such that 
\begin{equation*}
M\ge C\max{\left\{\log^{2}{\left(\frac{M\left( {K_{1}}+ {K_{2}}\right)}{\epsilon\eta}\right)},  \frac{1}{\epsilon^{2}}K_{\max}\log{\left(\frac{M\left( {K_{1}}+ {K_{2}}\right)}{\epsilon\eta}\right)},  \frac{1}{\epsilon^{2}}K_{\max}^{2}\log{\left(\frac{K_{1}+K_{2}}{\eta}\right)} \right\}},
\end{equation*}
or additionally, if the signs of the coefficients $a_{ik}$'s are i.i.d. generated from a symmetric distribution on the complex unit circle, there exists a numerical constant $C$ such that 
\begin{equation*}
M\ge C\max{\left\{ \frac{1}{\epsilon^{2}}\log^{2}{\left(\frac{M\left( {K_{1}}+ {K_{2}}\right)}{\epsilon\eta}\right)}, \frac{1}{\epsilon^{2}}K_{\max}\log{\left(\frac{K_{1}+K_{2}}{\eta}\right)}\log{\left(\frac{M\left( {K_{1}}+ {K_{2}}\right)}{\epsilon\eta}\right)}\right\}},
\end{equation*}
then we have
\begin{equation*}
\mathbb{P}\left\{\left\vert\frac{1}{\sqrt{\left\vert K''\left(0\right)\right\vert}^{l}}P^{\left(l\right)}\left(\tau\right) - \frac{1}{\sqrt{\left\vert K''\left(0\right)\right\vert}^{l}}P_{\mu}^{\left(l\right)}\left(\tau\right)\right\vert\le\epsilon,\ \forall\tau\in[0,1),\ l=0,1,2,3\right\}\ge 1-\eta.
\end{equation*}
\end{prop}

\subsubsection{Proof of Step 3}
This step follows essentially the same procedure as those in \cite[Lemma 4.13 and 4.14]{tang2014CSoffgrid}, where we divide $[0,1)$ into
\begin{equation}\label{def_upsilon} 
\Upsilon^i_{\mathrm{near}}=\cup_{k=1}^{K_{i}} \Upsilon^{i,k}_{\mathrm{near}}=\cup_{k=1}^{K_{i}}\left[\tau_{ik}-\tau_{s},\tau_{ik}+\tau_{s}\right], \quad \mathrm{and}\quad \Upsilon_{\mathrm{far}}^i=[0,1)\backslash\Upsilon^i_{\mathrm{near}},
\end{equation}
for $i=1,2$, where $\tau_{s}=8.245\times 10^{-2}/M$. Then conditioned on the event in Proposition~\ref{bound_continuous} one can bound $|P(\tau)|<1$ in $\Upsilon^1_{\mathrm{near}} \backslash \Upsilon_1$ and $\Upsilon_{\mathrm{far}}^1$ respectively following straightforward calculus. We shall omit the details and refer interested readers to \cite[Lemma 4.13 and 4.14]{tang2014CSoffgrid}. We have the following proposition.

\begin{prop}\label{bound_value_far}
Suppose $\Delta\ge 1/M$. There exists a numerical constant $C$ such that 
\begin{equation*}
M\ge C\max{\left\{\log^{2}{\left(\frac{M\left( {K_{1}}+ {K_{2}}\right)}{\eta}\right)}, K_{\max}\log{\left(\frac{M\left( {K_{1}}+ {K_{2}}\right)}{\eta}\right)}, K_{\max}^{2}\log{\left(\frac{K_{1}+K_{2}}{\eta}\right)} \right\}},
\end{equation*}
or additionally, if the signs of the coefficients $a_{ik}$'s are i.i.d. generated from a symmetric distribution on the complex unit circle, there exists a numerical constant $C$ such that 
\begin{equation*}
M\ge C\max{\left\{\log^{2}{\left(\frac{M\left( {K_{1}}+ {K_{2}}\right) }{\eta}\right)},K_{\max}\log{\left(\frac{K_{1}+K_{2}}{\eta}\right)}\log{\left(\frac{M\left( {K_{1}}+ {K_{2}}\right) }{\eta}\right)}\right\}},
\end{equation*}
then we have 
\begin{align*}
&\left\vert P\left(\tau\right)\right\vert  \le1- C_{p}M^2\left(\tau-\tau_{1k}\right)^{2} < 1, \quad \tau\in\Upsilon_{\mathrm{near}}^{1,k} \backslash \left\{ \tau_{1k} \right\},  \quad k=1,\dots,K_{1},\\
 & \left\vert P\left(\tau\right)-\mathrm{sign}\left(a_{1k}\right)\right\vert  \le C_{p}'M^{2}\left(\tau-\tau_{1k}\right)^{2}, \quad  \tau\in\Upsilon_{\mathrm{near}}^{1,k},  \quad k=1,\dots,K_{1}, \\
& \left\vert P\left(\tau\right)\right\vert  \le 1-C_{p}'' <1, \quad \tau\in\Upsilon_{\mathrm{far}}^1,
\end{align*}
with probability at least $1-\eta$, where $C_{p}$, $C_{p}^{\prime}$ and $C_{p}^{\prime\prime}$ are some positive numerical constants.
\end{prop}

\subsection{Finishing the Proof}
The proof of Theorem \ref{theorem_main} is now complete since we have established that $P(\tau)$ and $Q(\tau)$ constructed in \eqref{func_Pf} and \eqref{func_Qf} are indeed valid dual certificates under the condition of Theorem~\ref{theorem_main}.

\section{Proof of Theorem~\ref{theorem_noisy}}\label{sec::proof_theorem_noisy}

We first provide a proposition on optimality conditions of \eqref{algorithm_noisy_model}, which is proved in Appendix~\ref{proof_optimal_cond_noise}.
\begin{prop}\label{optimal_cond_noise}
$\left\{\hat{\bx}_{1},\hat{\bx}_{2}\right\}$ is the minimizer of \eqref{algorithm_noisy_model} if and only if the following holds:
\begin{align*}
\left\Vert\by-\left(\hat{\bx}_{1}+\bg\odot\hat{\bx}_{2}\right)\right\Vert_{\mathcal{A}}^{\star} & \le \lambda_{w}, \\
\left\Vert\bar{\bg}\odot\left(\by-\left(\hat{\bx}_{1}+\bg\odot\hat{\bx}_{2}\right)\right)\right\Vert_{\mathcal{A}}^{\star} & \le \lambda_{w},\\
\langle\by-\left(\hat{\bx}_{1}+\bg\odot\hat{\bx}_{2}\right),\hat{\bx}_{1}+\bg\odot\hat{\bx}_{2}\rangle_{\mathbb{R}} &=\lambda_{w}\left\Vert\hat{\bx}_{1}\right\Vert_{\mathcal{A}}+\lambda_{w}\left\Vert\hat{\bx}_{2}\right\Vert_{\mathcal{A}}.
\end{align*}
\end{prop}

Let $\be_{1}=\hat{\bx}_{1}-\bx_{1}^{\star}$ and $\be_{2}=\hat{\bx}_{2}-\bx_{2}^{\star}$. Moreover, let $\nu_{1}$ and $\nu_{2}$ be the corresponding representing measures \cite{candes2014towards,tang2013near} of $\be_{1}$ and $\be_{2}$, respectively, which are given as
$$ \be_{1}=\int_{0}^{1}\bc\left(\tau\right)\nu_{1}\left(d\tau\right), \quad \be_{2}=\int_{0}^{1}\bc\left(\tau\right)\nu_{2}\left(d\tau\right).$$ 
Therefore, we have $\|\be_i\|_{\cA}=\|\nu_i\|_{TV}$, $i=1,2$, where $\|\cdot\|_{TV}$ is the total variation norm of the representing measure. Define
\begin{align*}
 I_{i,0}^{k}&=\left\vert\int_{\Upsilon_{\mathrm{near}}^{i,k}}\nu_{i}\left(d\tau\right)\right\vert, \\ 
I_{i,1}^{k} &= \left(4M+1\right) \left\vert\int_{\Upsilon_{\mathrm{near}}^{i,k}} \left(\tau-\tau_{ik}\right) \nu_{i}\left(d\tau\right)\right\vert,\\
I_{i,2}^{k} &=\frac{\left(4M+1\right)^2}{2}\int_{\Upsilon_{\mathrm{near}}^{i,k}}\left(\tau-\tau_{ik}\right)^{2}\left\vert \nu_{i}\right\vert\left(d\tau\right), 
\end{align*}
and $I_{i,j}=\sum_{k=1}^{K_{i}}I_{i,j}^{k}$ for $j=0,1,2$ and $i=1,2$, where $\Upsilon^{i,k}_{\mathrm{near}}$ and $\Upsilon_{\mathrm{far}}^i$ are defined in \eqref{def_upsilon}. We have the following proposition whose proof can be found in Appendix~\ref{proof_bound_error}.
\begin{prop}\label{bound_error}
Assume the noise is bounded as $\|\bw\|_2^2\leq \sigma_w^2$. Set $\lambda_{w} = C_{w}\sigma_{w}\sqrt{4M+1}$, for some constant $C_{w} > 1$ large enough, then we have
\begin{align}
 \| \be_1\|_2+ \| \be_2\|_2 &  \leq  \sqrt{4M+1} \sum_{i=1}^{2}  \left(    \left\| P_{\Upsilon^i_{\mathrm{far}}}(\nu_i)\right\|_{TV}   + \sum_{j=0}^2 I_{i,j} \right), \label{estimation_error} \\
 \| \be_1 + \bg\odot \be_2 \|_2 & \leq  \sqrt{2\lambda_{w}  \sum_{i=1}^{2}\left(    \left\| P_{\Upsilon^{i}_{\mathrm{far}}}(\nu_i)\right\|_{TV}   + \sum_{j=0}^2 I_{i,j} \right)}. 
\end{align}
\end{prop}

Hence the rest is to provide an upper bound on the term $ \sum_{i=1}^{2}\left(    \left\| P_{\Upsilon^{i}_{\mathrm{far}}}(\nu_i)\right\|_{TV}   + \sum_{j=0}^2 I_{i,j} \right)$. We have the following proposition to control the sum value of zeroth moment terms $\sum_{i=1}^2 I_{i,0}$ and the sum value of first moment terms $\sum_{i=1}^2 I_{i,1}$, whose proof is given in Appendix~\ref{sec::proof_bound_zero_first_moment}.
\begin{prop}\label{bound_zero_first_moment}
Under the conditions in Theorem~\ref{theorem_noisy}, there exist some numerical constants $C_{0}$ and $C_{1}$, such that
\begin{align*}
\sum_{i=1}^2 I_{i,0}\le C_{0}\left(  \lambda_{w} \sqrt{ \frac{K_{\max}^3 \log M}{M}}  + \sum_{i=1}^{2}I_{i,2} + \sum_{i=1}^2 \left\| P_{\Upsilon^{i}_{\mathrm{far}}}(\nu_i)\right\|_{TV} \right),\\
\sum_{i=1}^2 I_{i,1}\le C_{1}\left(  \lambda_{w} \sqrt{ \frac{K_{\max}^3 \log M}{M}} +\sum_{i=1}^2 I_{i,2} +  \sum_{i=1}^2  \left\| P_{\Upsilon^{i}_{\mathrm{far}}}(\nu_i)\right\|_{TV} \right),
\end{align*}
for $i=1,2$, with high probability given in Theorem~\ref{theorem_noisy}.
\end{prop}

What remains is to bound $\sum_{i=1}^2   \left\| P_{\Upsilon^{i}_{\mathrm{far}}}(\nu_i)\right\|_{TV}  + \sum_{i=1}^2 I_{i,2} $, which is given in the following proposition proved in Appendix~\ref{proof_bound_left_together}.
\begin{prop}\label{bound_left_together}
Under the conditions in Theorem~\ref{theorem_noisy}, there exists a numerical constant $C$, such that 
\begin{equation*}
\sum_{i=1}^2   \left\| P_{\Upsilon^{i}_{\mathrm{far}}}(\nu_i)\right\|_{TV}  + \sum_{i=1}^2 I_{i,2} \leq  C\lambda_{w} \sqrt{ \frac{K_{\max}^3 \log M}{M}}
\end{equation*} 
holds with high probability given in Theorem~\ref{theorem_noisy}.
\end{prop}

Combining Propositions~\ref{bound_error}, \ref{bound_zero_first_moment} and \ref{bound_left_together}, there exists some constant $C$ such that
\begin{equation*}
\frac{1}{\sqrt{4M+1}}\left(\| \be_1\|_2 + \| \be_2\|_2\right)\le C\frac{K_{\max}^{\frac{3}{2}}\sqrt{\log{M}}\lambda_{w}}{\sqrt{M}} \le C_{1}\sigma_{w}\sqrt{K_{\max}^{3}\log{M}},
\end{equation*}
and
\begin{equation*} 
\frac{1}{\sqrt{4M+1}}\| \be_1 + \bg\odot \be_2 \|_2 \le  \frac{1}{\sqrt{4M+1}} \sqrt{2\lambda_{w}C\frac{\sqrt{K_{\max}^{3}\log{M}}\lambda_{w}}{\sqrt{M}} } \leq C_{2} \sigma_w \left(\frac{K_{\max}^3\log{M}}{M}\right)^{1/4} .
\end{equation*}

\section{Conclusions}\label{sec::conclusion}
 
We propose a convex optimization method based on atomic norm minimization to super-resolve two point source models from the measurements of their superposition, where each point source signal is
convolved with a different low-pass point spread function. It
is demonstrated, with high probability, that the point source
locations of each modality can be simultaneously determined
perfectly in the noise-free setting, from a near-optimal number of measurements when each point source signal satisfies a mild separation condition, and the point spread functions are randomly generated in the frequency domain. The proposed algorithm is also robust in the presence of bounded noise.

Our algorithmic framework and the proof methodology can be extended straightforwardly to handle more than two modalities when all of the modalities obey the conditions set forth in the current paper. There are a few possible future research directions. In applications such as multi-user detection, only a small number of users are active out of all the possible users. It will then be of great interest to
simultaneously identify a small set of active users as well as identify their corresponding point source signals. In addition, it will also be of interest to develop performance
guarantees of the proposed algorithm under milder conditions of
the point spread functions, for example when they are deterministic
but weakly correlated.

\section*{Acknowledgements}
This work is supported in part by the ONR Young Investigator Program Award N00014-15-1-2387, NSF Award CCF-1527456, and the ORAU Ralph E. Powe Junior Faculty Enhancement Award.

\bibliography{bibToeplitz}
\bibliographystyle{IEEEtran}

\appendix

\section{Useful Lemmas}\label{appendix_useful_lemmas}

\begin{lemma}\cite[noncommutative Bernstein's inequality]{tropp2012user} \label{bernstein}
Let $\{\boldsymbol{E}_{n}\}$ be a finite sequence of independent, random matrices with dimensions $d_{1}\times d_{2}$. Suppose that each random matrix satisfies 
\begin{equation*}
\mathbb{E}\left[\boldsymbol{E}_{n}\right]=\boldsymbol{0},\quad \mathrm{and} \quad \left\Vert\boldsymbol{E}_{n}\right\Vert\le R \ \ \mathrm{almost \ surely}.
\end{equation*}
Define 
\begin{equation*}
\sigma^{2}=\max\left\{\left\Vert\sum_{n}\mathbb{E}\left[\boldsymbol{E}_{n}\boldsymbol{E}_{n}^H\right]\right\Vert, \left\Vert\sum_{n}\mathbb{E}\left[\boldsymbol{E}_{n}^H\boldsymbol{E}_{n}\right]\right\Vert \right\}.
\end{equation*}
Then for any $ t\ge 0$,
\begin{equation*}
\mathbb{P}\left\{\left\Vert\sum_{n}\boldsymbol{E}_{n}\right\Vert\ge t \right\}\le\left(d_{1}+d_{2}\right)\cdot \exp\left(\frac{-t^{2}/2}{\sigma^{2}+Rt/3}\right).
\end{equation*}
\end{lemma}

\begin{lemma}\cite[Bernstein's polynomial inequality]{schaeffer1941inequalities}\label{bernsteins_poly_inequality}
Suppose $F\left(z\right)$ is a polynomial of degree $N$ with complex coefficients, then there exists 
$$\sup_{\left\vert z\right\vert\le 1}\left\vert F'\left(z\right)\right\vert\le N\cdot\sup_{\left\vert z\right\vert\le 1}\left\vert F\left(z\right)\right\vert.$$
\end{lemma}

\begin{lemma}\cite[Hoeffding's inequality]{hoeffding1963probability}\label{hoeffding_inequality}
Let the components of $\bu\in\mathbb{C}^{N}$ be sampled i.i.d. from a symmetric distribution on the complex unit circle, $\bw\in\mathbb{C}^{N}$, and $t$ be a positive real number. Then 
\begin{equation*}
\mathbb{P}\left\{\left\vert\langle\bu,\bw\rangle\right\vert\ge t\right\} \le 4 e^{-\frac{t^{2}}{4\left\Vert\bw\right\Vert_{2}^{2}}}.
\end{equation*}
\end{lemma}

\section{Proof of Dual Problem \eqref{convex_demixing_dual}}\label{proof_dual}
The Lagrangian function of \eqref{convex_demixing_rewritten} is given as
\begin{equation*}
L\left(\boldsymbol{x}_{1},\boldsymbol{x}_{2},\boldsymbol{p}\right)=\left\Vert\boldsymbol{x}_{1}\right\Vert_{\mathcal{A}} + \left\Vert\boldsymbol{x}_{2}\right\Vert_{\mathcal{A}}+\langle \boldsymbol{p},\boldsymbol{y}-\boldsymbol{x}_{1}-\boldsymbol{g}\odot\boldsymbol{x}_{2}\rangle_{\mathbb{R}},
\end{equation*}
whose infimum over $\bx_1$ and $\bx_2$ can be found as
\begin{align}
D\left(\boldsymbol{p}\right)&=\inf_{\boldsymbol{x}_{1},\boldsymbol{x}_{2}}L\left(\boldsymbol{x}_{1},\boldsymbol{x}_{2},\boldsymbol{p}\right) \nonumber \\
&=\inf_{\boldsymbol{x}_{1},\boldsymbol{x}_{2}}\{\left\Vert\boldsymbol{x}_{1}\right\Vert_{\mathcal{A}} -\langle\boldsymbol{p},\boldsymbol{x}_{1}\rangle_{\mathbb{R}}+ \left\Vert\boldsymbol{x}_{2}\right\Vert_{\mathcal{A}}-\langle\boldsymbol{p},\boldsymbol{g}\odot\boldsymbol{x}_{2}\rangle_{\mathbb{R}}+\langle\boldsymbol{p},\boldsymbol{y}\rangle_{\mathbb{R}}\} \nonumber \\
&=\inf_{\boldsymbol{x}_{1},\boldsymbol{x}_{2}}\{\left\Vert\boldsymbol{x}_{1}\right\Vert_{\mathcal{A}} -\langle\boldsymbol{p},\boldsymbol{x}_{1}\rangle_{\mathbb{R}}+ \left\Vert\boldsymbol{x}_{2}\right\Vert_{\mathcal{A}}-\langle\bar{\boldsymbol{g}}\odot\boldsymbol{p},\boldsymbol{x}_{2}\rangle_{\mathbb{R}}+\langle\boldsymbol{p},\boldsymbol{y}\rangle_{\mathbb{R}}\} \nonumber \\
& = \inf_{\boldsymbol{x}_{1} }\{\left\Vert\boldsymbol{x}_{1}\right\Vert_{\mathcal{A}} -\langle\boldsymbol{p},\boldsymbol{x}_{1}\rangle_{\mathbb{R}} \} + \inf_{\boldsymbol{x}_{2}} \{ \left\Vert\boldsymbol{x}_{2}\right\Vert_{\mathcal{A}}-\langle\bar{\boldsymbol{g}}\odot\boldsymbol{p},\boldsymbol{x}_{2}\rangle_{\mathbb{R}} \} +\langle\boldsymbol{p},\boldsymbol{y}\rangle_{\mathbb{R}}. \label{dual_opt}
\end{align}
Plugging into \eqref{dual_opt} the facts that 
$$\inf_{\boldsymbol{x}_{i}}\{\left\Vert\boldsymbol{x}_{i}\right\Vert_{\mathcal{A}} -\langle\boldsymbol{p},\boldsymbol{x}_{i}\rangle_{\mathbb{R}}\}= \left\{\begin{array}{cc} 
0 , &  \left\Vert\boldsymbol{p}\right\Vert_{\mathcal{A}}^{\star}\leq1 \\
-\infty, & \mathrm{otherwise}
\end{array}\right. ,$$
for $i=1,2$, we can have the dual problem of \eqref{convex_demixing_rewritten} as given in \eqref{convex_demixing_dual}.

\section{Proof of $\Upsilon_{1}\subseteq\hat{\Upsilon}_{1}$ and $\Upsilon_{2}\subseteq\hat{\Upsilon}_{2}$}\label{proof_fre_recovery}
If $\Upsilon_{1}\backslash\hat{\Upsilon}_{1}\neq\emptyset$ or $\Upsilon_{2}\backslash\hat{\Upsilon}_{2}\neq\emptyset$, there exists $ \vert\hat{P}\left(\tau\right) \vert<1$ for $\tau\in\Upsilon_{1}\backslash\hat{\Upsilon}_{1}$ or $ \vert\hat{Q}\left(\tau\right) \vert<1$ for $\tau\in\Upsilon_{2}\backslash\hat{\Upsilon}_{2}$. Then we have
\begin{align*}
\langle\hat{\bp},\by\rangle_{\mathbb{R}}&=\langle\hat{\bp},\bx_{1}^{\star}\rangle_{\mathbb{R}} + \langle\hat{\bp},\bg\odot\bx_{2}^{\star}\rangle_{\mathbb{R}}\\
&=\langle\hat{\bp},\sum_{k=1}^{K_{1}}a_{1k}\bc\left(\tau_{1k}\right)\rangle_{\mathbb{R}} + \langle\bar{\bg}\odot\hat{\bp},\sum_{k=1}^{K_{2}}a_{2k}\bc\left(\tau_{2k}\right)\rangle_{\mathbb{R}}\\
&=\sum_{\tau_{1k}\in\Upsilon_{1}\cap\hat{\Upsilon}_{1}}\mathrm{Re}\left(\bar{a}_{1k}\hat{P}\left(\tau_{1k}\right)\right) + \sum_{\tau_{1k}\in\Upsilon_{1}\backslash\hat{\Upsilon}_{1}}\mathrm{Re}\left(\bar{a}_{1k}\hat{P}\left(\tau_{1k}\right)\right) \\
&\quad+ \sum_{\tau_{2k}\in\Upsilon_{2}\cap\hat{\Upsilon}_{2}}\mathrm{Re}\left(\bar{a}_{2k}\hat{Q}\left(\tau_{2k}\right)\right) + \sum_{\tau_{2k}\in\Upsilon_{2}\backslash\hat{\Upsilon}_{2}}\mathrm{Re}\left(\bar{a}_{2k}\hat{Q}\left(\tau_{2k}\right)\right)\\
&<\sum_{\tau_{1k}\in\Upsilon_{1}\cap\hat{\Upsilon}_{1}}\left\vert a_{1k}\right\vert + \sum_{\tau_{1k}\in\Upsilon_{1}\backslash\hat{\Upsilon}_{1}}\left\vert a_{1k}\right\vert + \sum_{\tau_{2k}\in\Upsilon_{2}\cap\hat{\Upsilon}_{2}}\left\vert a_{2k}\right\vert + \sum_{\tau_{2k}\in\Upsilon_{2}\backslash\hat{\Upsilon}_{2}}\left\vert a_{2k}\right\vert\\
&=\left\Vert\bx_{1}^{\star}\right\Vert_{\mathcal{A}} + \left\Vert\bx_{2}^{\star}\right\Vert_{\mathcal{A}},
\end{align*}
where the strict inequality violates strong duality. Therefore, $\Upsilon_{1}\subseteq\hat{\Upsilon}_{1}$ and $\Upsilon_{2}\subseteq\hat{\Upsilon}_{2}$.

\section{Proof of Proposition~\ref{dual_certificate}}\label{proof_dual_certification}
\begin{proof}
Since
\begin{align*}
&\left\Vert\boldsymbol{p}\right\Vert_{\mathcal{A}}^{\star}=\sup_{\tau\in[0,1)}\left\vert\sum_{n=-2M}^{2M}p_{n}e^{j2\pi n\tau}\right\vert=\sup_{\tau\in[0,1)}\left\vert P\left(\tau\right)\right\vert\le 1,\\
&\left\Vert\bar{\boldsymbol{g}}\odot\boldsymbol{p}\right\Vert_{\mathcal{A}}^{\star}=\sup_{\tau\in[0,1)}\left\vert\sum_{n=-2M}^{2M}\bar{g}_{n}p_{n}e^{j2\pi n\tau}\right\vert=\sup_{\tau\in[0,1)}\left\vert Q\left(\tau\right)\right\vert\le 1,
\end{align*}
the vector $\boldsymbol{p}$ satisfying \eqref{conditions} is dual feasible. First,
\begin{align*}
\langle\boldsymbol{p},\boldsymbol{y} \rangle_{\mathbb{R}}&=\langle\boldsymbol{p},\boldsymbol{x}_{1}^{\star}+\boldsymbol{g}\odot\boldsymbol{x}_{2}^{\star}\rangle_{\mathbb{R}}\\
&=\langle\boldsymbol{p},\boldsymbol{x}_{1}^{\star}\rangle_{\mathbb{R}}+\langle\bar{\boldsymbol{g}}\odot\boldsymbol{p},\boldsymbol{x}_{2}^{\star}\rangle_{\mathbb{R}}\\
&=\sum_{k=1}^{K_{1}}\mathrm{Re}\left(\bar{a}_{1k}\sum_{n=-2M}^{2M}p_{n}e^{j2\pi n\tau_{1k}}\right)+\sum_{k=1}^{K_{2}}\mathrm{Re}\left(\bar{a}_{2k}\sum_{n=-2M}^{2M}\bar{g}_{n}p_{n}e^{j2\pi n\tau_{2k}}\right)\\
&=\sum_{k=1}^{K_{1}}\mathrm{Re}\left(\bar{a}_{1k}\mathrm{sign}\left(a_{1k}\right)\right)+\sum_{k=1}^{K_{2}}\mathrm{Re}\left(\bar{a}_{2k}\mathrm{sign}\left(a_{2k}\right)\right)\\
&=\sum_{k=1}^{K_{1}}\left\vert a_{1k}\right\vert+\sum_{k=1}^{K_{2}}\left\vert a_{2k}\right\vert \ge \left\Vert\boldsymbol{x}_{1}^{\star}\right\Vert_{\mathcal{A}} + \left\Vert\boldsymbol{x}_{2}^{\star}\right\Vert_{\mathcal{A}}.
\end{align*}
Also, we have
\begin{equation*}
\langle\boldsymbol{p},\boldsymbol{y} \rangle_{\mathbb{R}}=\langle\boldsymbol{p},\boldsymbol{x}_{1}^{\star}\rangle_{\mathbb{R}}+\langle\bar{\boldsymbol{g}}\odot\boldsymbol{p},\boldsymbol{x}_{2}^{\star}\rangle_{\mathbb{R}}\le \left\Vert \boldsymbol{p}\right\Vert_{\mathcal{A}}^{\star}\left\Vert\boldsymbol{x}_{1}^{\star}\right\Vert_{\mathcal{A}}+\left\Vert \bar{\boldsymbol{g}}\odot\boldsymbol{p}\right\Vert_{\mathcal{A}}^{\star}\left\Vert\boldsymbol{x}_{2}^{\star}\right\Vert_{\mathcal{A}}\leq \left\Vert\boldsymbol{x}_{1}^{\star}\right\Vert_{\mathcal{A}} + \left\Vert\boldsymbol{x}_{2}^{\star}\right\Vert_{\mathcal{A}},
\end{equation*}
which gives $\langle\boldsymbol{p},\boldsymbol{y}\rangle_{\mathbb{R}}=\left\Vert\boldsymbol{x}_{1}^{\star}\right\Vert_{\mathcal{A}} + \left\Vert\boldsymbol{x}_{2}^{\star}\right\Vert_{\mathcal{A}}$. This implies that $\boldsymbol{p}$ is a dual optimal solution of \eqref{convex_demixing_dual}, and that $\boldsymbol{x}_{1}^{\star}$ and $\boldsymbol{x}_{2}^{\star}$ are the primal optimal solutions of \eqref{convex_demixing_rewritten}. 

Now validate the uniqueness of $\boldsymbol{x}_{1}^{\star}$ and $\boldsymbol{x}_{2}^{\star}$. Suppose there is a different optimal solution of \eqref{convex_demixing_rewritten}, which can be written as $\hat{\bx}_{i} = \sum_{k=1}^{\hat{K}_i} \hat{a}_{ik} \bc(\hat{\tau}_{ik})$, where $\hat{\Upsilon}_i =\{\hat{\tau}_{ik}| k=1,\ldots, \hat{K}_i\}$, and $\left\Vert\hat{\boldsymbol{x}}_{i}\right\Vert_{\mathcal{A}}=\sum_{k=1}^{\hat{K}_{i}}\left\vert \hat{a}_{ik}\right\vert$ for $i=1,2$, and it satisfies $\by = \hat{\bx}_1 + \bg \odot \hat{\bx}_2$. If $\hat{\Upsilon}_i =  {\Upsilon}_i$ for $i=1,2$, we have $\hat{\bx}_1=\bx_1^\star$ and $\hat{\bx}_2=\bx_2^\star$ straightforwardly. We then consider the case when at least $\hat{\Upsilon}_i \neq  {\Upsilon}_i$ for some $i$. We have
\begin{align*}
\langle\boldsymbol{p}, \boldsymbol{y}\rangle_{\mathbb{R}}&=\langle\boldsymbol{p},\hat{\boldsymbol{x}}_{1}+\boldsymbol{g}\odot\hat{\boldsymbol{x}}_{2}\rangle_{\mathbb{R}}\\
&=\langle\boldsymbol{p},\hat{\boldsymbol{x}}_{1}\rangle_{\mathbb{R}}+\langle\bar{\boldsymbol{g}}\odot\boldsymbol{p},\hat{\boldsymbol{x}}_{2}\rangle_{\mathbb{R}}\\
&=\sum_{\hat{\tau}_{1k}\in\hat{\Upsilon}_{1}\cap\Upsilon_1}\mathrm{Re}\left( \bar{\hat{a}}_{1k} \sum_{n=-2M}^{2M}p_{n}e^{j2\pi n\hat{\tau}_{1k}}\right)+\sum_{\hat{\tau}_{1k}\in\hat{\Upsilon}_{1}\backslash\Upsilon_1}\mathrm{Re}\left(\bar{\hat{a}}_{1k} \sum_{n=-2M}^{2M}p_{n}e^{j2\pi n\hat{\tau}_{1k}}\right)\\
&\quad +\sum_{\hat{\tau}_{2k}\in\hat{\Upsilon}_{2}\cap\Upsilon_2}\mathrm{Re}\left(\bar{\hat{a}}_{2k} \sum_{n=-2M}^{2M}\bar{g}_{n}p_{n}e^{j2\pi n\hat{\tau}_{2k}}\right)+\sum_{\hat{\tau}_{2k}\in\hat{\Upsilon}_{2}\backslash\Upsilon_2}\mathrm{Re}\left(\bar{\hat{a}}_{2k} \sum_{n=-2M}^{2M}\bar{g}_{n}p_{n}e^{j2\pi n\hat{\tau}_{2k}}\right)\\
&\le \sum_{\hat{\tau}_{1k}\in\hat{\Upsilon}_{1}\cap\Upsilon_1}\left\vert \hat{a}_{1k}\right\vert + \sum_{\hat{\tau}_{1k}\in\hat{\Upsilon}_{1}\backslash\Upsilon_1}\mathrm{Re}\left(\bar{\hat{a}}_{1k}P(\hat{\tau}_{1k})\right) + \sum_{\hat{\tau}_{2k}\in\hat{\Upsilon}_{2}\cap\Upsilon_2}\left\vert \hat{a}_{2k}\right\vert + \sum_{\hat{\tau}_{2k}\in\hat{\Upsilon}_{2}\backslash\Upsilon_2}\mathrm{Re}\left(\bar{\hat{a}}_{2k} Q(\hat{\tau}_{2k})\right)\\
&< \sum_{\hat{\tau}_{1k}\in\hat{\Upsilon}_{1}\cap\Upsilon_1}\left\vert \hat{a}_{1k}\right\vert + \sum_{\hat{\tau}_{1k}\in\hat{\Upsilon}_{1}\backslash\Upsilon_1}\left\vert \hat{a}_{1k}\right\vert + \sum_{\hat{\tau}_{2k}\in\hat{\Upsilon}_{2}\cap\Upsilon_2}\left\vert \hat{a}_{2k}\right\vert + \sum_{\hat{\tau}_{2k}\in\hat{\Upsilon}_{2}\backslash\Upsilon_2}\left\vert \hat{a}_{2k}\right\vert\\
&=\left\Vert\hat{\boldsymbol{x}}_{1}\right\Vert_{\mathcal{A}}+\left\Vert\hat{\boldsymbol{x}}_{2}\right\Vert_{\mathcal{A}},
\end{align*}
which violates strong duality. Thus $(\boldsymbol{x}_{1}^{\star}, \boldsymbol{x}_{2}^{\star})$ is the unique primal optimal solution of \eqref{convex_demixing_rewritten}.
\end{proof}

\section{Proof of Proposition~\ref{prop:invertibility}} \label{proof_invertibility}

\begin{proof}
To apply Lemma~\ref{bernstein} to \eqref{decomposition_Wg}, we first bound $\left\Vert\boldsymbol{E}_{n}\right\Vert$ as
\begin{align*}
\left\Vert\boldsymbol{E}_{n}\right\Vert & = \left\Vert \frac{1}{M}s_{n}g_{n}\boldsymbol{e}_{1}\left(n\right)\boldsymbol{e}_{2}^{H}\left(n\right)\right\Vert\\
&=\frac{1}{M}\left\vert s_{n}\right\vert\sqrt{K_{1}+\frac{K_{1}}{\left\vert K''\left(0\right)\right\vert}\left(2\pi n\right)^{2}}\sqrt{K_{2}+\frac{K_{2}}{\left\vert K''\left(0\right)\right\vert}\left(2\pi n\right)^{2}}\\
&\le \frac{1}{M}  \left(\max_{\left\vert n\right\vert\le2M} \left\vert s_n\right\vert \right)  \sqrt{K_{1}K_{2}}\left(1+\max_{\left\vert n\right\vert\le2M}\frac{\left(2\pi n\right)^{2}}{\left\vert K''\left(0\right)\right\vert}\right)\\
&\le 14\frac{\sqrt{K_{1}K_{2}}}{M}: =  R, \quad \mathrm{for}\ M\ge 4,
\end{align*}
where $ \max_{\left\vert n\right\vert\le2M} \left\vert s_n\right\vert \leq 1$, and $\left(1+\max_{\left\vert n\right\vert\le2M}\frac{\left(2\pi n\right)^{2}}{\left\vert K''\left(0\right)\right\vert}\right)=1+\frac{12M^{2}}{M^{2}-1}\le 14$, for $M\ge 4$ \cite{tang2014CSoffgrid}. 

Furthermore,
\begin{equation*}
\begin{split}
\left\Vert\sum_{n=-2M}^{2M}\mathbb{E}\left[\boldsymbol{E}_{n}\boldsymbol{E}_{n}^{H}\right]\right\Vert&=\left\Vert\sum_{n=-2M}^{2M}\mathbb{E}\left[\frac{1}{M}s_{n}g_{n}\boldsymbol{e}_{1}\left(n\right)\boldsymbol{e}_{2}^{H}\left(n\right)\cdot\frac{1}{M}s_{n} \bar{g}_{n} \boldsymbol{e}_{2}\left(n\right)\boldsymbol{e}_{1}^{H}\left(n\right)\right]\right\Vert\\
&=\left\Vert\sum_{n=-2M}^{2M}\frac{1}{M^{2}}s_{n}^{2}K_{2}\left(1+\frac{\left(2\pi n\right)^{2}}{\left\vert K''\left(0\right)\right\vert}\right)\boldsymbol{e}_{1}\left(n\right)\boldsymbol{e}_{1}^{H}\left(n\right)\right\Vert\\
&\le\frac{1}{M}K_{2}\left(1+\max_{\left\vert n\right\vert\le2M}\frac{\left(2\pi n\right)^{2}}{\left\vert K''\left(0\right)\right\vert}\right)   \left( \max_{\left\vert n\right\vert\le2M} \left\vert s_n\right\vert\right)   \left\Vert\sum_{n=-2M}^{2M}\frac{1}{M}s_{n}\boldsymbol{e}_{1}\left(n\right)\boldsymbol{e}_{1}^{H}\left(n\right)\right\Vert\\
&\le \frac{14}{M}K_{2}\left\Vert\boldsymbol{W}_{1}\right\Vert\le 20\frac{K_{2}}{M},
\end{split}
\end{equation*}
where the last inequality follows from \eqref{single_W}. Similarly we can obtain $\left\Vert\sum_{n=-2M}^{2M}\mathbb{E}\left[\boldsymbol{E}_{n}^{H}\boldsymbol{E}_{n}\right]\right\Vert \le 20\frac{K_{1}}{M}$. Hence,
\begin{equation*}
\sigma^{2}=\max\left\{\left\Vert\sum_{n}\mathbb{E}\left[\boldsymbol{E}_{n}\boldsymbol{E}_{n}^{H}\right]\right\Vert, \left\Vert\sum_{n}\mathbb{E}\left[\boldsymbol{E}_{n}^{H}\boldsymbol{E}_{n}\right]\right\Vert\right\} =\frac{20}{M}K_{\max}.
\end{equation*}

Apply Lemma~\ref{bernstein}, for $0<\delta<0.6376$, then we have
\begin{equation}
\begin{split}
\mathbb{P}\{\left\Vert\boldsymbol{W}_{g}\right\Vert\ge \delta\}&\le 2\left(K_{1}+K_{2}\right)\exp\left(\frac{-\delta^{2}/2}{\frac{20}{M}K_{\max} +  \frac{14\delta}{3M}\sqrt{K_{1}K_{2}}}\right)\\
&\le 2\left(K_{1}+K_{2}\right)\exp\left(-\frac{\delta^{2}M}{46K_{\max}}\right) \leq \eta,
\end{split}
\end{equation}
if  $M\ge \frac{46}{\delta^{2}}K_{\max}\log\left(\frac{2\left(K_{1}+K_{2}\right)}{\eta}\right)$.
\end{proof}

\section{Proof of Lemma~\ref{bound_inver}}\label{proof_bound_inver}
\begin{proof}
For both invertible $\boldsymbol{A}$ and $\boldsymbol{B}$ that satisfy $\left\Vert\boldsymbol{A}-\boldsymbol{B}\right\Vert\left\Vert\boldsymbol{B}^{-1}\right\Vert\le 1/2$, it has \cite{tang2014CSoffgrid}
\begin{equation*}
\left\Vert\boldsymbol{A}^{-1}\right\Vert\le 2\left\Vert\boldsymbol{B}^{-1}\right\Vert,\quad \mathrm{and}\quad \left\Vert\boldsymbol{A}^{-1}-\boldsymbol{B}^{-1}\right\Vert\le2\left\Vert\boldsymbol{B}^{-1}\right\Vert^{2}\left\Vert\boldsymbol{A}-\boldsymbol{B}\right\Vert.
\end{equation*}
Applying the above to $\boldsymbol{A} = \boldsymbol{W}$ and $\boldsymbol{B} = \boldsymbol{W}_{\mu}$, from \eqref{inv_W}, we have $\left\Vert\boldsymbol{W}_{\mu}^{-1}\right\Vert\le 1.568$. Under the event $\mathcal{E}_\delta$, $\left\Vert\boldsymbol{W}-\boldsymbol{W}_{\mu}\right\Vert= \left\Vert\boldsymbol{W}_{g}\right\Vert\le\delta$. Therefore as soon as $\delta\le \frac{1}{4}\le\frac{1}{2\left\Vert\boldsymbol{W}_{\mu}^{-1}\right\Vert}$\footnote{This choice of $\delta$ is not unique but good enough for our purpose.}, we have
\begin{align*}
\left\Vert\boldsymbol{W}^{-1}\right\Vert & \le 2\left\Vert\boldsymbol{W}_{\mu}^{-1}\right\Vert, \\
\left\Vert\boldsymbol{W}^{-1}-\boldsymbol{W}_{\mu}^{-1}\right\Vert & \le 2\left\Vert\boldsymbol{W}_{\mu}^{-1}\right\Vert^{2}\left\Vert\boldsymbol{W}-\boldsymbol{W}_{\mu}\right\Vert\le 2\left\Vert\boldsymbol{W}_{\mu}^{-1}\right\Vert^{2}\delta.
\end{align*}
Finally, because the operator norm of a matrix dominates that of its submatrices, we have
\begin{align*}
&\left\Vert\boldsymbol{L}_{i}-\boldsymbol{L}_{\mu i}\right\Vert\le 2\left\Vert\boldsymbol{W}_{\mu}^{-1}\right\Vert^{2}\delta, \quad \mathrm{for}\quad i=1,2,\\
&\left\Vert\boldsymbol{L}_{g}\right\Vert\le 2\left\Vert\boldsymbol{W}_{\mu}^{-1}\right\Vert^{2}\delta, \\ &\left\Vert\boldsymbol{L}_{\bar{g}}\right\Vert\le 2\left\Vert\boldsymbol{W}_{\mu}^{-1}\right\Vert^{2}\delta, \end{align*}
and $\left\Vert\boldsymbol{L}_{i}\right\Vert\le \left\Vert\boldsymbol{L}_{i} - \boldsymbol{L}_{\mu i}\right\Vert +\left\Vert\boldsymbol{L}_{\mu i}\right\Vert  \le2\left\Vert\boldsymbol{W}_{\mu}^{-1}\right\Vert^{2}\delta + \left\Vert\boldsymbol{W}_{\mu}^{-1}\right\Vert\leq   2\left\Vert\boldsymbol{W}_{\mu}^{-1}\right\Vert$ for $ i=1,2$ where we have used $\left\Vert\boldsymbol{W}_{\mu}^{-1}\right\Vert\leq 1.568$ and $\delta\leq 1/4$.
\end{proof}

\section{Proof of Proposition~\ref{bound_continuous}}\label{proof_bound_continuous}

\begin{proof}
Conditioned on the event $\mathcal{E}_{\delta}$ with $\delta\in(0,1/4]$, we have
\begin{align}
&\quad \left\vert\frac{1}{\sqrt{\left\vert K''\left(0\right)\right\vert}^{l}}P^{\left(l\right)}\left(\tau\right)\right\vert \nonumber \\
&\le \left\vert\langle\boldsymbol{u}_{1},\boldsymbol{L}_{1}^{H}\boldsymbol{v}_{1l}\left(\tau\right)\rangle\right\vert + \left\vert\langle\boldsymbol{u}_{2},\boldsymbol{L}_{g}^{H}\boldsymbol{v}_{1l}\left(\tau\right)\rangle\right\vert + \left\vert\langle\boldsymbol{u}_{1},\boldsymbol{L}_{\bar{g}}^{H}\boldsymbol{v}_{2l}\left(\tau\right)\rangle\right\vert + \left\vert\langle\boldsymbol{u}_{2},\boldsymbol{L}_{2}^{H}\boldsymbol{v}_{2l}\left(\tau\right)\rangle\right\vert \nonumber\\
&\le\left\Vert\bu_{1}\right\Vert_{2}\left\Vert\boldsymbol{L}_{1} \right\Vert\left\Vert\bv_{1l}\left(\tau\right)\right\Vert_{2} + \left\Vert\bu_{2}\right\Vert_{2}\left\Vert\boldsymbol{L}_{g}\right\Vert\left\Vert\bv_{1l}\left(\tau\right)\right\Vert_{2} + \left\Vert\bu_{1}\right\Vert_{2}\left\Vert\boldsymbol{L}_{\bar{g}}\right\Vert\left\Vert\bv_{2l}\left(\tau\right)\right\Vert_{2} + \left\Vert\bu_{2}\right\Vert_{2}\left\Vert\boldsymbol{L}_{2}\right\Vert\left\Vert\bv_{2l}\left(\tau\right)\right\Vert_{2} \nonumber\\
&\le \sqrt{K_{1}}\cdot 2\left\Vert\boldsymbol{W}_{\mu}^{-1}\right\Vert\cdot \left(4M+1\right)\frac{1}{M}4^{l+1}\sqrt{K_{1}} + \sqrt{K_{2}}\cdot 0.8\left\Vert\boldsymbol{W}_{\mu}^{-1}\right\Vert\cdot \left(4M+1\right)\frac{1}{M}4^{l+1}\sqrt{K_{1}}\nonumber\\
&\quad + \sqrt{K_{1}}\cdot 0.8\left\Vert\boldsymbol{W}_{\mu}^{-1}\right\Vert\cdot \left(4M+1\right)\frac{1}{M}4^{l+1}\sqrt{K_{2}} + \sqrt{K_{2}}\cdot 2\left\Vert\boldsymbol{W}_{\mu}^{-1}\right\Vert\cdot \left(4M+1\right)\frac{1}{M}4^{l+1}\sqrt{K_{2}} \label{boundeq}\\
&\le C \left( {K_{1}}+ {K_{2}}\right) ,\nonumber
\end{align}
for some universal constant $C$. In \eqref{boundeq}, we applied Lemma~\ref{bound_inver}, $\left\Vert\bu_{i}\right\Vert_{2}=\sqrt{K_{i}}$, for $i=1,2$, and
\begin{equation}\label{bound_v1}
\begin{split}
\left\Vert\bv_{1l}\left(\tau\right)\right\Vert_{2}&=\left\Vert\frac{1}{M}\sum_{n=-2M}^{2M}s_{n} \left(\frac{-j2\pi n}{\sqrt{\left\vert K''\left(0\right)\right\vert}}\right)^{l} e^{-j2\pi n\tau}\boldsymbol{e}_{1}\left(n\right)\right\Vert_{2}\\
&\le\frac{1}{M}\left(4M+1\right)  \left(  \max_{\left\vert n\right\vert\le2M} \left\vert s_n\right\vert \right)  \left(\max_{\left\vert n\right\vert\le 2M}\left\vert\frac{j2\pi n}{\sqrt{\left\vert K''\left(0\right)\right\vert}}\right\vert^{l}\right) \left(\max_{\left\vert n\right\vert\le 2M} \left\Vert\boldsymbol{e}_{1}\left(n\right)\right\Vert_{2}\right) \\
&\le \frac{1}{M}\left(4M+1\right)4^{l}\sqrt{K_{1}}\max_{\left\vert n\right\vert\le 2M}\sqrt{1+\frac{\left(2\pi n\right)^{2}}{\left\vert K''\left(0\right)\right\vert}}\\
&\le \frac{1}{M}\left(4M+1\right)4^{l}\sqrt{14K_{1}},
\end{split}
\end{equation}
and similarly 
\begin{equation*}
\left\Vert\bv_{2l}\left(\tau\right)\right\Vert_{2}\le \frac{1}{M}\left(4M+1\right)4^{l}\sqrt{14K_{2}}.
\end{equation*}
In \eqref{bound_v1} we have used $ \max_{\left\vert n\right\vert\le2M} \left\vert s_n\right\vert \leq 1$, $\max_{\left\vert n\right\vert\le 2M} \left\vert\frac{j2\pi n}{\sqrt{\left\vert K''\left(0\right)\right\vert}}\right\vert  \le 4$, for $M\ge 2$ and $\left(1+\max_{\left\vert n\right\vert\le2M}\frac{\left(2\pi n\right)^{2}}{\left\vert K''\left(0\right)\right\vert}\right)\le 14$, for $M\ge 4$. 

Using Lemma~\ref{bernsteins_poly_inequality}, we have
\begin{align*}
\left\vert\frac{1}{\sqrt{\left\vert K''\left(0\right)\right\vert}^{l}}P^{\left(l\right)}\left(\tau_{a}\right) - \frac{1}{\sqrt{\left\vert K''\left(0\right)\right\vert}^{l}}P^{\left(l\right)}\left(\tau_{b}\right)\right\vert
&\le \left\vert e^{j2\pi\tau_{a}}-e^{j2\pi\tau_{b}}\right\vert\sup_{\tau\in\left[0,1\right]}\left\vert\frac{\partial \frac{1}{\sqrt{\left\vert K''\left(0\right)\right\vert}^{l}}P^{\left(l\right)}\left(e^{j2\pi\tau}\right)}{\partial e^{j2\pi\tau}}\right\vert\\
&\le \left\vert e^{j2\pi\tau_{a}}-e^{j2\pi\tau_{b}}\right\vert\cdot 2M\sup_{\tau\in\left[0,1\right]}\left\vert\frac{1}{\sqrt{\left\vert K''\left(0\right)\right\vert}^{l}}P^{\left(l\right)}\left(\tau\right)\right\vert\\
&\le 4\pi\left\vert\tau_{a}-\tau_{b}\right\vert\cdot 2M\cdot C \left( {K_{1}}+ {K_{2}}\right).
\end{align*}
Note that similar bounds also hold for $P_{\mu}^{(l)}(\tau)$. Conditioned on the event $\mathcal{E}_{\delta}\cap\mathcal{E}_{1}$ with $\delta\in(0,1/4]$, we have 
\begin{align*}
&\quad\left\vert\frac{1}{\sqrt{\left\vert K''\left(0\right)\right\vert}^{l}}P^{\left(l\right)}\left(\tau\right) - \frac{1}{\sqrt{\left\vert K''\left(0\right)\right\vert}^{l}}P_{\mu}^{\left(l\right)}\left(\tau\right)\right\vert\\
&\le \left\vert\frac{1}{\sqrt{\left\vert K''\left(0\right)\right\vert}^{l}}P^{\left(l\right)}\left(\tau\right) - \frac{1}{\sqrt{\left\vert K''\left(0\right)\right\vert}^{l}}P^{\left(l\right)}\left(\tau_{d}\right)\right\vert+\left\vert\frac{1}{\sqrt{\left\vert K''\left(0\right)\right\vert}^{l}}P^{\left(l\right)}\left(\tau_{d}\right) - \frac{1}{\sqrt{\left\vert K''\left(0\right)\right\vert}^{l}}P_{\mu}^{\left(l\right)}\left(\tau_{d}\right)\right\vert\\
&\quad+\left\vert\frac{1}{\sqrt{\left\vert K''\left(0\right)\right\vert}^{l}}P_{\mu}^{\left(l\right)}\left(\tau_{d}\right) - \frac{1}{\sqrt{\left\vert K''\left(0\right)\right\vert}^{l}}P_{\mu}^{\left(l\right)}\left(\tau\right)\right\vert\\
&\le 4\pi\left\vert\tau-\tau_{d}\right\vert\cdot 2M\cdot C \left( {K_{1}}+ {K_{2}}\right)  + \frac{\epsilon}{3} + 4\pi\left\vert\tau_{d}-\tau\right\vert\cdot 2M\cdot C \left( {K_{1}}+ {K_{2}}\right),
\end{align*}
for any $\tau\in\left[0,1\right]$, where $\tau_{d}\in \Upsilon_{\mathrm{grid}}$. By setting the grid size $\left\vert\Upsilon_{\mathrm{grid}}\right\vert=\left\lceil \frac{24\pi CM\left( {K_{1}}+ {K_{2}}\right)}{\epsilon} \right\rceil$, we have $|\tau_d -\tau|\leq\frac{\epsilon}{24\pi CM (K_1 + K_2)}$, which yields
$$ \left\vert\frac{1}{\sqrt{\left\vert K''\left(0\right)\right\vert}^{l}}P^{\left(l\right)}\left(\tau\right) - \frac{1}{\sqrt{\left\vert K''\left(0\right)\right\vert}^{l}}P_{\mu}^{\left(l\right)}\left(\tau\right)\right\vert \leq  {\epsilon} . $$
By plugging the grid size and modifying the condition on $M$, the proof is complete.
\end{proof}

\section{Proof of Proposition~\ref{optimal_cond_noise}}\label{proof_optimal_cond_noise}
\begin{proof}
Denote $f\left(\bx_{1},\bx_{2}\right)= \frac{1}{2}\left\Vert\by-\bx_{1}-\bg\odot\bx_{2}\right\Vert_{2}^{2}+\lambda_{w}\left(\left\Vert\bx_{1}\right\Vert_{\mathcal{A}}+\left\Vert\bx_{2}\right\Vert_{\mathcal{A}}\right)$ as the objective function of \eqref{algorithm_noisy_model}. Since $\left\{\hat{\bx}_{1},\hat{\bx}_{2}\right\}$ is the minimizer of \eqref{algorithm_noisy_model}, for all $\alpha_{w}\in\left(0,1\right]$ and all $\left\{\tilde{\bx}_{1},\tilde{\bx}_{2}\right\}$, we have
\begin{equation*}
f\left(\alpha_{w}\tilde{\bx}_{1}+\left(1-\alpha_{w}\right)\hat{\bx}_{1},\alpha_{w}\tilde{\bx}_{2}+\left(1-\alpha_{w}\right)\hat{\bx}_{2}\right)\ge f\left(\hat{\bx}_{1},\hat{\bx}_{2}\right).
\end{equation*}
This is equivalent to the following
\begin{equation*}
\begin{split}
&\alpha_{w}^{-1}\lambda_{w}\left(\left\Vert \hat{\bx}_{1} + \alpha_{w}(\tilde{\bx}_{1}- \hat{\bx}_{1}  )\right\Vert_{\mathcal{A}}-\left\Vert\hat{\bx}_{1}\right\Vert_{\mathcal{A}}\right)+\alpha_{w}^{-1}\lambda_{w}\left(\left\Vert\hat{\bx}_{2}+\alpha_{w}\left(\tilde{\bx}_{2}-\hat{\bx}_{2}\right) \right\Vert_{\mathcal{A}}-\left\Vert\hat{\bx}_{2}\right\Vert_{\mathcal{A}}\right)\\
&\ge \langle\by-\left(\hat{\bx}_{1}+\bg\odot\hat{\bx}_{2}\right),\left(\tilde{\bx}_{1}-\hat{\bx}_{1}\right)+\bg\odot\left(\tilde{\bx}_{2}-\hat{\bx}_{2}\right)\rangle_{\mathbb{R}}-\frac{1}{2}\alpha_{w}\left\Vert\hat{\bx}_{1}-\tilde{\bx}_{1} +\bg\odot (\hat{\bx}_{2}- \tilde{\bx}_{2}) \right\Vert_{2}^{2}.
\end{split}
\end{equation*}

As the atomic norm $\left\Vert\cdot\right\Vert_{\mathcal{A}}$ is convex, the following inequalities hold:
\begin{equation*}
\begin{split}
&\left\Vert\tilde{\bx}_{1}\right\Vert_{\mathcal{A}}-\left\Vert\hat{\bx}_{1}\right\Vert_{\mathcal{A}}\ge\alpha_{w}^{-1}\left(\left\Vert\hat{\bx}_{1}+\alpha_{w}\left(\tilde{\bx}_{1}-\hat{\bx}_{1}\right)\right\Vert_{\mathcal{A}}-\left\Vert\hat{\bx}_{1}\right\Vert_{\mathcal{A}}\right), \\
&\left\Vert\tilde{\bx}_{2}\right\Vert_{\mathcal{A}}-\left\Vert\hat{\bx}_{2}\right\Vert_{\mathcal{A}}\ge\alpha_{w}^{-1}\left(\left\Vert\hat{\bx}_{2}+\alpha_{w}\left(\tilde{\bx}_{2}-\hat{\bx}_{2}\right)\right\Vert_{\mathcal{A}}-\left\Vert\hat{\bx}_{2}\right\Vert_{\mathcal{A}}\right), \\
\end{split}
\end{equation*}
which can be plugged into the previous inequality to obtain
\begin{equation*}
\begin{split}
&\lambda_{w}\left(\left\Vert\tilde{\bx}_{1}\right\Vert_{\mathcal{A}}+\left\Vert\tilde{\bx}_{2}\right\Vert_{\mathcal{A}}-\left\Vert\hat{\bx}_{1}\right\Vert_{\mathcal{A}}-\left\Vert\hat{\bx}_{2}\right\Vert_{\mathcal{A}}\right)\\
&\ge \langle\by-\left(\hat{\bx}_{1}+\bg\odot\hat{\bx}_{2}\right),\left(\tilde{\bx}_{1}-\hat{\bx}_{1}\right)+\bg\odot\left(\tilde{\bx}_{2}-\hat{\bx}_{2}\right)\rangle_{\mathbb{R}}-\frac{1}{2}\alpha_{w}\left\Vert\hat{\bx}_{1}-\tilde{\bx}_{1}+\bg\odot\left(\hat{\bx}_{2}-\tilde{\bx}_{2}\right)\right\Vert_{2}^{2}.
\end{split}
\end{equation*}
\reviseA{
Set $\alpha_{w}\to 0$, we can obtain that $\left\{\hat{\bx}_{1},\hat{\bx}_{2}\right\}$ is the minimizer of \eqref{algorithm_noisy_model} only if for all $\left\{\tilde{\bx}_{1},\tilde{\bx}_{2}\right\}$, there exists
\begin{equation}\label{equ_iff_optimalcond_noisy}
\lambda_{w}\left(\left\Vert\tilde{\bx}_{1}\right\Vert_{\mathcal{A}}+\left\Vert\tilde{\bx}_{2}\right\Vert_{\mathcal{A}}-\left\Vert\hat{\bx}_{1}\right\Vert_{\mathcal{A}}-\left\Vert\hat{\bx}_{2}\right\Vert_{\mathcal{A}}\right)\ge \langle\by-\left(\hat{\bx}_{1}+\bg\odot\hat{\bx}_{2}\right),\left(\tilde{\bx}_{1}-\hat{\bx}_{1}\right)+\bg\odot\left(\tilde{\bx}_{2}-\hat{\bx}_{2}\right)\rangle_{\mathbb{R}}.
\end{equation}

On the other hand, if \eqref{equ_iff_optimalcond_noisy} holds for all $\left\{\tilde{\bx}_{1},\tilde{\bx}_{2}\right\}$, we have
\begin{align*}
f\left(\tilde{\bx}_{1}, \tilde{\bx}_{2}\right)
& = \frac{1}{2}\left\Vert\by-\tilde{\bx}_{1}-\bg\odot\tilde{\bx}_{2}\right\Vert_{2}^{2}+\lambda_{w}\left(\left\Vert\tilde{\bx}_{1}\right\Vert_{\mathcal{A}}+\left\Vert\tilde{\bx}_{2}\right\Vert_{\mathcal{A}}\right)\\
& = \frac{1}{2} \left\Vert \by - \hat{\bx}_{1} - \bg\odot\hat{\bx}_{2}  +  \hat{\bx}_{1} + \bg\odot\hat{\bx}_{2} - \tilde{\bx}_{1} - \bg\odot \tilde{\bx}_{2} \right\Vert_{2}^{2}  +  \lambda_{w} \left(\left\Vert\hat{\bx}_{1}\right\Vert_{\mathcal{A}}+\left\Vert\hat{\bx}_{2}\right\Vert_{\mathcal{A}}\right) \\
& \quad + \lambda_{w}\left(\left\Vert\tilde{\bx}_{1}\right\Vert_{\mathcal{A}}+\left\Vert\tilde{\bx}_{2}\right\Vert_{\mathcal{A}} -  \left\Vert\hat{\bx}_{1}\right\Vert_{\mathcal{A}} - \left\Vert\hat{\bx}_{2}\right\Vert_{\mathcal{A}} \right)\\
& = \frac{1}{2}\left\Vert\by-\hat{\bx}_{1}-\bg\odot\hat{\bx}_{2}\right\Vert_{2}^{2} +  \lambda_{w} \left(\left\Vert\hat{\bx}_{1}\right\Vert_{\mathcal{A}}+\left\Vert\hat{\bx}_{2}\right\Vert_{\mathcal{A}}\right) + \frac{1}{2} \left\Vert   \hat{\bx}_{1} + \bg\odot\hat{\bx}_{2} - \tilde{\bx}_{1} - \bg\odot \tilde{\bx}_{2} \right\Vert_{2}^{2}\\
&\quad + \langle \by-\hat{\bx}_{1}-\bg\odot\hat{\bx}_{2},   \hat{\bx}_{1} + \bg\odot\hat{\bx}_{2} - \tilde{\bx}_{1} - \bg\odot \tilde{\bx}_{2}  \rangle_{\mathbb{R}}  + \lambda_{w}\left(\left\Vert\tilde{\bx}_{1}\right\Vert_{\mathcal{A}}+\left\Vert\tilde{\bx}_{2}\right\Vert_{\mathcal{A}} -  \left\Vert\hat{\bx}_{1}\right\Vert_{\mathcal{A}} - \left\Vert\hat{\bx}_{2}\right\Vert_{\mathcal{A}} \right)\\
& \ge f\left(\hat{\bx}_{1}, \hat{\bx}_{2}\right) + \frac{1}{2} \left\Vert   \hat{\bx}_{1} + \bg\odot\hat{\bx}_{2} - \tilde{\bx}_{1} - \bg\odot \tilde{\bx}_{2} \right\Vert_{2}^{2}\\
& \ge f\left(\hat{\bx}_{1}, \hat{\bx}_{2}\right).
\end{align*}
Therefore, \eqref{equ_iff_optimalcond_noisy} holds if and only if $\left\{\hat{\bx}_{1},\hat{\bx}_{2}\right\}$ is the minimizer of \eqref{algorithm_noisy_model}.

Furthermore, we can rewrite \eqref{equ_iff_optimalcond_noisy} by moving all the terms containing $\left\{\tilde{\bx}_{1},\tilde{\bx}_{2}\right\}$ onto one side as
\begin{align}
& \lambda_{w}\left(\left\Vert\hat{\bx}_{1}\right\Vert_{\mathcal{A}}+\left\Vert\hat{\bx}_{2}\right\Vert_{\mathcal{A}}\right) -  \langle\by-\left(\hat{\bx}_{1}+\bg\odot\hat{\bx}_{2}\right),\hat{\bx}_{1}+\bg\odot\hat{\bx}_{2}\rangle_{\mathbb{R}} \nonumber\\
& \le \lambda_{w} \left\Vert\tilde{\bx}_{1}\right\Vert_{\mathcal{A}} -   \langle\by-\left(\hat{\bx}_{1}+\bg\odot\hat{\bx}_{2}\right),\tilde{\bx}_{1}\rangle_{\mathbb{R}} 
+ \lambda_{w}\left\Vert\tilde{\bx}_{2}\right\Vert_{\mathcal{A}} -   \langle\by-\left(\hat{\bx}_{1}+\bg\odot\hat{\bx}_{2}\right),  \bg\odot\tilde{\bx}_{2}\rangle_{\mathbb{R}}. \label{equ_iff_optimalcond_noisy_rewrite}
\end{align}
Since \eqref{equ_iff_optimalcond_noisy_rewrite} holds for all $\left\{\tilde{\bx}_{1},\tilde{\bx}_{2}\right\}$, \eqref{equ_iff_optimalcond_noisy_rewrite} still holds if taking infimum on the right-hand side with respect to $\left\{\tilde{\bx}_{1},\tilde{\bx}_{2}\right\}$. That is 
\begin{align*}
&\lambda_{w}\left(\left\Vert\hat{\bx}_{1}\right\Vert_{\mathcal{A}}+\left\Vert\hat{\bx}_{2}\right\Vert_{\mathcal{A}}\right)-\langle\by-\left(\hat{\bx}_{1}+\bg\odot\hat{\bx}_{2}\right),\hat{\bx}_{1}+\bg\odot\hat{\bx}_{2}\rangle_{\mathbb{R}}\\
&\le \inf_{\tilde{\bx}_{1}}\left\{ \lambda_{w}\left\Vert\tilde{\bx}_{1}\right\Vert_{\mathcal{A}}-\langle\by-\left(\hat{\bx}_{1}+\bg\odot\hat{\bx}_{2}\right),\tilde{\bx}_{1}\rangle_{\mathbb{R}} \right\} + \inf_{\tilde{\bx}_{2}}\left\{ \lambda_{w}\left\Vert\tilde{\bx}_{2}\right\Vert_{\mathcal{A}}-\langle\by-\left(\hat{\bx}_{1}+\bg\odot\hat{\bx}_{2}\right),\bg\odot\tilde{\bx}_{2}\rangle_{\mathbb{R}}\right\}.
\end{align*}
}
Plugging in the facts that 
\begin{equation*}
\inf_{\tilde{\boldsymbol{x}}_{i}}\{\left\Vert\tilde{\boldsymbol{x}}_{i}\right\Vert_{\mathcal{A}} -\langle\by-\left(\hat{\bx}_{1}+\bg\odot\hat{\bx}_{2}\right),\tilde{\bx}_{i}\rangle_{\mathbb{R}}\}= \left\{\begin{array}{cc} 
0 , &  \left\Vert\by-\left(\hat{\bx}_{1}+\bg\odot\hat{\bx}_{2}\right)\right\Vert_{\mathcal{A}}^{\star}\leq1 \\
-\infty, & \mathrm{otherwise}
\end{array}\right. ,\mathrm{for}\ i=1,2,
\end{equation*}
we have
\begin{equation*}
\langle\by-\left(\hat{\bx}_{1}+\bg\odot\hat{\bx}_{2}\right),\hat{\bx}_{1}+\bg\odot\hat{\bx}_{2}\rangle_{\mathbb{R}}\ge\lambda_{w}\left\Vert\hat{\bx}_{1}\right\Vert_{\mathcal{A}}+\lambda_{w}\left\Vert\hat{\bx}_{2}\right\Vert_{\mathcal{A}},
\end{equation*}
as well as
\begin{equation*}
\left\Vert\by-\left(\hat{\bx}_{1}+\bg\odot\hat{\bx}_{2}\right)\right\Vert_{\mathcal{A}}^{\star}\le \lambda_{w},\ \mathrm{and}\  \left\Vert\bar{\bg}\odot\left(\by-\left(\hat{\bx}_{1}+\bg\odot\hat{\bx}_{2}\right)\right)\right\Vert_{\mathcal{A}}^{\star}\le \lambda_{w}.
\end{equation*}
\end{proof}

\section{Proof of Proposition~\ref{bound_error}} \label{proof_bound_error}
\begin{proof}

We first record a useful lemma from \cite{tang2013near}.
\begin{lemma}\cite[Lemma 1]{tang2013near}\label{bound_trig}
For any $2m$th-order trigonometric polynomial $X(\tau) =\langle \bx, \boldsymbol{c}(\tau)\rangle$, we have  
$$ \int_0^1 X(\tau)\nu_i(d\tau) \leq \| \bx\|_{\mathcal{A}}^{\star} \left(    \left\| P_{\Upsilon^{i}_{\mathrm{far}}}(\nu_i)\right\|_{TV}   + \sum_{j=0}^2 I_{i,j} \right), $$
for $i=1,2$, where $P_{A}\left(\nu_{i}\right)$ denote the projection of the measure $\nu_{i}$ on the support set $A$.
\end{lemma}
Setting $X(\tau)=\langle \be_i , \bc(\tau)\rangle$ for $i=1,2$ in Lemma~\ref{bound_trig}, we obtain
\begin{align*}
\| \be_i\|_2^2  =\left\langle \be_i,  \int_{0}^{1}\bc\left(\tau\right)\nu_{i}\left(d\tau\right) \right\rangle & = \int_{0}^{1}\left\langle \be_i, \bc\left(\tau\right) \right\rangle \nu_{i}\left(d\tau\right)  \\
& \leq \| \be_i \|_{\mathcal{A}}^{\star} \left(    \left\| P_{\Upsilon^{i}_{\mathrm{far}}}(\nu_i)\right\|_{TV}   + \sum_{j=0}^2 I_{i,j} \right) \\
& \leq  \sqrt{4M+1} \|\be_i\|_2 \left(    \left\| P_{\Upsilon^{i}_{\mathrm{far}}}(\nu_i)\right\|_{TV}   + \sum_{j=0}^2 I_{i,j} \right),
\end{align*}
where we used the fact \reviseA{$\| \be_i \|_{\mathcal{A}}^{\star}  = \sup_{\tau\in[0,1)} \left\vert \langle \be_{i},  \boldsymbol{c}\left(\tau\right) \rangle \right\vert \le \left\Vert \be_{i} \right\Vert_{2}  \left\Vert  \boldsymbol{c}\left(\tau\right) \right\Vert_{2} =  \sqrt{4M+1}\|\be_i\|_2$ following the Cauchy-Schwarz inequality}. This yields the estimation error of $\bx_{i}^{\star}$ in \eqref{estimation_error}. For the denoising error, first notice that,
\begin{align}
\left\Vert \be_{1}+\bg\odot\be_{2}  \right\Vert_{\cA}^{\star} 
&\le \left\Vert\by-\bx_{1}^{\star}-\bg\odot\bx_{2}^{\star}\right\Vert_{\mathcal{A}}^{\star} + \left\Vert\by-\hat{\bx}_{1}-\bg\odot\hat{\bx}_{2}\right\Vert_{\mathcal{A}}^{\star} \nonumber \\
&\le\left\Vert\boldsymbol{w}\right\Vert_{\mathcal{A}}^{\star}+\lambda_{w} \label{from_optimality}\\
&\le  2\lambda_{w}, \label{noise_atomic_bound}
\end{align}
where \eqref{from_optimality} follows from Proposition~\ref{optimal_cond_noise}, and \eqref{noise_atomic_bound} follows from \reviseA{$\left\Vert \bw \right\Vert_{\mathcal{A}}^{\star} = \sup_{\tau\in[0,1)} \left\vert \langle \bw,  \boldsymbol{c}\left(\tau\right) \rangle \right\vert \le \left\Vert \bw \right\Vert_{2}  \left\Vert  \boldsymbol{c}\left(\tau\right) \right\Vert_{2} \le \sigma_{w}\sqrt{4M+1}  = \lambda_{w}/C_{w}  \le \lambda_{w}$. Similarly, we have $\left\Vert \bar{\bg}\odot\bw \right\Vert_{\mathcal{A}}^{\star} \le \lambda_{w}/C_{w}  $ and consequently $\left\Vert\bar{\bg}\odot\be_{1}+\be_{2}\right\Vert_{\mathcal{A}}^{\star}\leq 2\lambda_{w}$.} Therefore, we have
\begin{align*}
 \| \be_1 + \bg\odot \be_2 \|_2^2 & = \langle\be_{1}+\bg\odot\be_{2},\be_{1}\rangle + \langle\bar{\bg}\odot\be_{1}+\be_{2},\be_{2}\rangle\\
&=\langle\be_{1}+\bg\odot\be_{2},\int_{0}^{1}\bc\left(\tau\right)\nu_1\left(d\tau\right)\rangle + \langle\bar{\bg}\odot\be_{1}+\be_{2},\int_{0}^{1}\bc\left(\tau\right)\nu_2 \left(d\tau\right)\rangle\\
&=\int_{0}^{1}\langle\be_{1}+\bg\odot\be_{2},\bc\left(\tau\right)\rangle \nu_{1}\left(d\tau\right) + \int_{0}^{1}\langle\bar{\bg}\odot\be_{1}+\be_{2},\bc\left(\tau\right)\rangle \nu_{2}\left(d\tau\right)\\
&\le\left\Vert\be_{1}+\bg\odot\be_{2}\right\Vert_{\mathcal{A}}^{\star}\left(    \left\| P_{\Upsilon^{1}_{\mathrm{far}}}(\nu_1)\right\|_{TV}   + \sum_{j=0}^2 I_{1,j} \right)  + \left\Vert\bar{\bg}\odot\be_{1}+\be_{2}\right\Vert_{\mathcal{A}}^{\star} \left(    \left\| P_{\Upsilon^{2}_{\mathrm{far}}}(\nu_2)\right\|_{TV}   + \sum_{j=0}^2 I_{2,j} \right)\\
& \leq 2\lambda_{w}  \sum_{i=1}^{2}\left(    \left\| P_{\Upsilon^{i}_{\mathrm{far}}}(\nu_i)\right\|_{TV}   + \sum_{j=0}^2 I_{i,j} \right),
\end{align*}
where we used \eqref{noise_atomic_bound} in the last inequality.
\end{proof}

\section{Proof of Proposition~\ref{bound_zero_first_moment}}\label{sec::proof_bound_zero_first_moment}

We first construct a pair of trigonometric polynomials $P_1\left(\tau\right)$ and $Q_{1}\left(\tau\right)$ with the following properties whose proof can be found in Appendix~\ref{sec::proof_constructed_poly_p1q1}.
\begin{lemma}\label{constructed_poly_p1q1}
Assume that $g_{n}=e^{j2\pi\phi_n}$'s are i.i.d. randomly generated from a uniform distribution on the complex unit circle with $\phi_n\sim\mathcal{U}[0,1]$. Provided that the separation $\Delta \ge 1/M$, there exists a numerical constant $C$ such that as soon as
\begin{equation*}
M\ge C \max\left\{\log^{2}{\left(\frac{M\left( {K_{1}}+ {K_{2}}\right) }{\eta}\right)}, K_{\max}\log{\left(\frac{M\left(K_{1}+K_{2}\right)}{\eta}\right)}, K_{\max}^{2}\log{\left(\frac{K_{1}+K_{2}}{\eta}\right)} \right\},
\end{equation*}
we can construct $P_{1}\left(\tau\right)=\sum_{n=-2M}^{2M}p_{1n}e^{j2\pi n\tau}$ and $Q_{1}\left(\tau\right)=\sum_{n=-2M}^{2M}p_{1n}\bar{g}_{n}e^{j2\pi n\tau}$ that satisfy
\begin{align*}
\left\vert P_1\left(\tau\right)-\mathrm{sign}\left(a_{1k}\right)\left(\tau-\tau_{1k}\right)\right\vert & \leq C_{p}M\left(\tau-\tau_{1k}\right)^{2},  \quad \tau\in\Upsilon_{\mathrm{near}}^{1,k}, \quad k=1,\dots,K_{1}, \\
 \left\vert P_{1}\left(\tau\right)\right\vert & \le \frac{C_{p}^{\prime}}{M}, \quad \tau\in\Upsilon_{\mathrm{far}}^{1}, \\
\left\vert Q_{1}\left(\tau\right)-\mathrm{sign}\left(a_{2k}\right)\left(\tau-\tau_{2k}\right)\right\vert & \leq C_{q}M\left(\tau-\tau_{2k}\right)^{2},  \quad \tau\in\Upsilon_{\mathrm{near}}^{2,k}, \quad k=1,\dots,K_{2}, \\
 \left\vert Q_{1}\left(\tau\right)\right\vert &\le \frac{C_{q}^{\prime}}{M},  \quad \tau\in\Upsilon_{\mathrm{far}}^2,
\end{align*}
with probability at least $1-\eta$, where $C_{p}$, $C_{p}^{\prime}$, $C_{q}$ and $C_{q}^{\prime}$ are numerical constants.
\end{lemma}

Furthermore, we derive the following useful lemma in Appendix~\ref{sec::proof_bound_combined_poly}.

\begin{lemma}\label{bound_combined_poly}
For $P(\tau)$ and $Q(\tau)$ constructed in Proposition~\ref{bound_value_far}, and $P_1(\tau)$ and $Q_1(\tau)$ constructed in Lemma~\ref{constructed_poly_p1q1}, there exist numerical constant $C$ and $C_{1}$ such that 
\begin{align}
&\left\vert \int_{0}^{1}P\left(\tau\right)\nu_{1}\left(d\tau\right) + \int_{0}^{1}Q\left(\tau\right)\nu_{2}\left(d\tau\right)\right\vert\le C\lambda_{w} \sqrt{ \frac{K_{\max}^3 \log M}{M}}, \label{bound_poly_eq1}\\
&\left\vert \int_{0}^{1}P_{1}\left(\tau\right)\nu_{1}\left(d\tau\right) + \int_{0}^{1}Q_{1}\left(\tau\right)\nu_{2}\left(d\tau\right)\right\vert\le C_1\lambda_{w} \sqrt{\frac{K_{\max}^3\log{M}}{M^3}},\label{bound_poly_eq2}
\end{align}
with high probability given in Theorem~\ref{theorem_noisy}.
\end{lemma} 

\begin{proof}
Consider the polar form 
\begin{equation*}
\left\vert\int_{\Upsilon_{\mathrm{near}}^{i,k}}\nu_{i}\left(d\tau\right)\right\vert = e^{-j\rho_{ik}} \int_{\Upsilon_{\mathrm{near}}^{i,k}}\nu_{i}\left(d\tau\right), \quad i=1,2,
\end{equation*}
then we can construct a pair of dual polynomials $P\left(\tau\right)$ and $Q\left(\tau\right)$ that interpolate a pair of point sources with $\mathrm{sign}\left(\tilde{a}_{ik}\right)=e^{-j\rho_{ik}}$, as in Proposition~\ref{bound_value_far}. Therefore, we have
\begin{equation*}
\begin{split}
I_{1,0}&=\sum_{k=1}^{K_{1}}\left\vert\int_{\Upsilon_{\mathrm{near}}^{1,k}}\nu_{1}\left(d\tau\right)\right\vert\\
&=\sum_{k=1}^{K_{1}}\int_{\Upsilon_{\mathrm{near}}^{1,k}} P\left(\tau\right) \nu_{1}\left(d\tau\right) + \sum_{k=1}^{K_{1}}\int_{\Upsilon_{\mathrm{near}}^{1,k}} \left(e^{-j\rho_{1k}} - P\left(\tau\right)\right) \nu_{1}\left(d\tau\right)\\
&=\int_{0}^{1}P\left(\tau\right)\nu_{1}\left(d\tau\right) - \int_{\Upsilon_{\mathrm{far}}^{1}}P\left(\tau\right)\nu_{1}\left(d\tau\right)+ \sum_{k=1}^{K_{1}}\int_{\Upsilon_{\mathrm{near}}^{1,k}} \left(e^{-j\rho_{1k}} - P\left(\tau\right)\right) \nu_{1}\left(d\tau\right).
\end{split}
\end{equation*}
Similarly,  
\begin{equation*}
I_{2,0}=\int_{0}^{1}Q\left(\tau\right)\nu_{2}\left(d\tau\right) - \int_{\Upsilon_{\mathrm{far}}^{2}}Q\left(\tau\right)\nu_{2}\left(d\tau\right) + \sum_{k=1}^{K_{2}}\int_{\Upsilon_{\mathrm{near}}^{2,k}} \left(e^{-j\rho_{2k}} - Q\left(\tau\right)\right) \nu_{2}\left(d\tau\right).
\end{equation*}
Now consider their sum, then we have
\begin{align}
\sum_{i=1}^{2}I_{i,0}&\le\left\vert \int_{0}^{1}P\left(\tau\right)\nu_{1}\left(d\tau\right) + \int_{0}^{1}Q\left(\tau\right)\nu_{2}\left(d\tau\right)\right\vert + \sum_{i=1}^2 \left\| P_{\Upsilon^{i}_{\mathrm{far}}}(\nu_i)\right\|_{TV} \nonumber\\
&\quad +  \sum_{k=1}^{K_{1}}\int_{\Upsilon_{\mathrm{near}}^{1,k}} C_p M^2 \left(\tau-\tau_{1k}\right)^{2} \left\vert\nu_{1}\right\vert\left(d\tau\right) +  \sum_{k=1}^{K_{2}}\int_{\Upsilon_{\mathrm{near}}^{2,k}} C_q M^2 \left(\tau-\tau_{2k}\right)^{2} \left\vert\nu_{2}\right\vert\left(d\tau\right)\label{eq_constructed_pq_bound_by_I2}\\
&\le\left\vert \int_{0}^{1}P\left(\tau\right)\nu_{1}\left(d\tau\right) + \int_{0}^{1}Q\left(\tau\right)\nu_{2}\left(d\tau\right)\right\vert + \sum_{i=1}^2 \left\| P_{\Upsilon^{i}_{\mathrm{far}}}(\nu_i)\right\|_{TV} + C_{2}\sum_{i=1}^{2}I_{i,2}\label{eq_defIi2}\\
&\le C\lambda_{w} \sqrt{ \frac{K_{\max}^3 \log M}{M}} + \sum_{i=1}^2 \left\| P_{\Upsilon^{i}_{\mathrm{far}}}(\nu_i)\right\|_{TV} + C_{2}\sum_{i=1}^{2}I_{i,2}\label{eq_bound_constructed_pq},
\end{align} 
where \eqref{eq_constructed_pq_bound_by_I2} follows from the triangle inequality and the properties of the dual polynomials in Proposition~\ref{bound_value_far}, \eqref{eq_defIi2} follows from the definition of $I_{i,2}$, and \eqref{eq_bound_constructed_pq} follows from Lemma~\ref{bound_combined_poly}.

Then, we consider bounding $\sum_{i=1}^{2}I_{i,1}$ in a similar way. Again, consider the polar form
\begin{equation*}
\left\vert\int_{\Upsilon_{\mathrm{near}}^{i,k}} \left(\tau-\tau_{ik}\right) \nu_{i}\left(d\tau\right)\right\vert = e^{-j\rho_{ik}}\int_{\Upsilon_{\mathrm{near}}^{i,k}} \left(\tau-\tau_{ik}\right) \nu_{i}\left(d\tau\right), \quad i=1,2,
\end{equation*}
then we can construct a pair of polynomials $P_{1}\left(\tau\right)$ and $Q_{1}\left(\tau\right)$ in the form of Lemma~\ref{constructed_poly_p1q1} by letting $\mathrm{sign}\left(\tilde{a}_{ik}\right)=e^{-j\rho_{ik}}$. Then we have
\begin{equation*}
\begin{split}
\frac{I_{1,1}}{4M+1}&=\sum_{k=1}^{K_{1}}\left\vert\int_{\Upsilon_{\mathrm{near}}^{1,k}} \left(\tau-\tau_{1k}\right) \nu_{1}\left(d\tau\right)\right\vert\\
&= \sum_{k=1}^{K_{1}}\int_{\Upsilon_{\mathrm{near}}^{1,k}} \left(e^{-j\rho_{1k}}\left(\tau-\tau_{1k}\right) - P_{1}\left(\tau\right)\right)\nu_{1}\left(d\tau\right) +  \sum_{k=1}^{K_{1}}\int_{\Upsilon_{\mathrm{near}}^{1,k}} P_{1}\left(\tau\right)\nu_{1}\left(d\tau\right)\\
&= \sum_{k=1}^{K_{1}}\int_{\Upsilon_{\mathrm{near}}^{1,k}} \left(e^{-j\rho_{1k}}\left(\tau-\tau_{1k}\right) - P_{1}\left(\tau\right)\right)\nu_{1}\left(d\tau\right)  + \int_{0}^{1} P_{1}\left(\tau\right)\nu_{1}\left(d\tau\right) - \int_{\Upsilon_{\mathrm{far}}^{1}} P_{1}\left(\tau\right)\nu_{1}\left(d\tau\right),
\end{split}
\end{equation*}
and
\begin{equation*}
\frac{I_{2,1}}{4M+1}= \sum_{k=1}^{K_{2}}\int_{\Upsilon_{\mathrm{near}}^{2,k}} \left(e^{-j\rho_{2k}}\left(\tau-\tau_{2k}\right) - Q_{1}\left(\tau\right)\right)\nu_{2}\left(d\tau\right)  + \int_{0}^{1} Q_{1}\left(\tau\right)\nu_{2}\left(d\tau\right) -  \int_{\Upsilon_{\mathrm{far}}^{2}} Q_{1}\left(\tau\right)\nu_{2}\left(d\tau\right).
\end{equation*}
Taking their sum, we have
\begin{align}
 \sum_{i=1}^{2}I_{i,1} &\leq (4M+1) \Big(\left\vert \int_{0}^{1}P_{1}\left(\tau\right)\nu_{1}\left(d\tau\right) + \int_{0}^{1}Q_{1}\left(\tau\right)\nu_{2}\left(d\tau\right)\right\vert +  \sum_{k=1}^{K_{1}}\int_{\Upsilon_{\mathrm{near}}^{1,k}} \left\vert e^{-j\rho_{1k}}\left(\tau-\tau_{1k}\right) - P_{1}\left(\tau\right)\right\vert\left\vert\nu_{1}\right\vert\left(d\tau\right)  \nonumber\\
&\quad +  \sum_{k=1}^{K_{2}}\int_{\Upsilon_{\mathrm{near}}^{2,k}} \left\vert e^{-j\rho_{2k}}\left(\tau-\tau_{2k}\right) - Q_{1}\left(\tau\right)\right\vert\left\vert\nu_{2}\right\vert\left(d\tau\right)   \Big) + C_3 \sum_{i=1}^2 \left\| P_{\Upsilon^{i}_{\mathrm{far}}}(\nu_i)\right\|_{TV} \nonumber\\
&\le C\lambda_{w} \sqrt{ \frac{K_{\max}^3 \log M}{M}}  + C_2\sum_{i=1}^{2}I_{i,2} + C_3 \sum_{i=1}^2 \left\| P_{\Upsilon^{i}_{\mathrm{far}}}(\nu_i)\right\|_{TV},\nonumber
\end{align}
where the first inequality follows from the triangle inequality and Lemma~\ref{constructed_poly_p1q1}, and the last inequality follows from Lemma~\ref{bound_combined_poly}, the definition of $I_{i,2}$ and Lemma~\ref{constructed_poly_p1q1}.
\end{proof}

\section{Proof of Proposition~\ref{bound_left_together}}\label{proof_bound_left_together}
\begin{proof}
Let $\hat{u}_{i}$ and $u_{i}^{\star}$ denote the representing measure of $\hat{\bx}_{i}$ and $\bx_{i}^{\star}$, then we have $\nu_i=\hat{u}_{i} - u_{i}^{\star} $, $i=1,2$. Since $\|u_{i}^{\star}\|_{TV}=\|\bx_i^{\star}\|_{\cA}$ and $\|\hat{u}_{i}\|_{TV}=\|\hat{\bx}_{i} \|_{\cA}$, from Proposition~\ref{optimal_cond_noise}, we have
\reviseA{
\begin{align}
\|\hat{u}_{1}\|_{TV} + \|\hat{u}_{2}\|_{TV} 
& = \|\hat{\bx}_{1} \|_{\cA} + \|\hat{\bx}_{2} \|_{\cA} \nonumber\\
& = \frac{1}{\lambda_{w}} \langle\by-\left(\hat{\bx}_{1}+\bg\odot\hat{\bx}_{2}\right),\be_{1}+\bg\odot\be_{2} \rangle_{\mathbb{R}}  \nonumber\\
&\quad + \frac{1}{\lambda_{w}} \langle\by-\left(\hat{\bx}_{1}+\bg\odot\hat{\bx}_{2}\right), \bx_{1}^{\star}  \rangle_{\mathbb{R}}  +  \frac{1}{\lambda_{w}} \langle\by-\left(\hat{\bx}_{1}+\bg\odot\hat{\bx}_{2}\right),\bg\odot\bx_{2}^{\star} \rangle_{\mathbb{R}}   \nonumber\\
&\le \frac{1}{\lambda_{w}} \langle\by-\left(\hat{\bx}_{1}+\bg\odot\hat{\bx}_{2}\right),\be_{1}+\bg\odot\be_{2}  \rangle_{\mathbb{R}} + \left\Vert\bx_{1}^{\star}\right\Vert_{\mathcal{A}} + \left\Vert\bx_{2}^{\star}\right\Vert_{\mathcal{A}} \nonumber\\
& = \frac{1}{\lambda_{w}} \langle \bx_1^{\star} + \boldsymbol{g} \odot \bx_2^{\star} + \boldsymbol{w} -\left(\hat{\bx}_{1}+\bg\odot\hat{\bx}_{2}\right),\be_{1}+\bg\odot\be_{2} \rangle_{\mathbb{R}} + \left\Vert\bx_{1}^{\star}\right\Vert_{\mathcal{A}} + \left\Vert\bx_{2}^{\star}\right\Vert_{\mathcal{A}} \nonumber\\
& = -\frac{1}{\lambda_{w}} \left\Vert   \be_{1} + \bg\odot\be_{2}  \right\Vert_{2}^{2} + \frac{1}{\lambda_{w}} \langle  \boldsymbol{w} ,\be_{1}+\bg\odot\be_{2} \rangle_{\mathbb{R}} + \left\Vert u_{1}^{\star}\right\Vert_{TV} + \left\Vert u_{2}^{\star}\right\Vert_{TV}  \nonumber\\
& \le \left\Vert u_{1}^{\star}\right\Vert_{TV} + \left\Vert u_{2}^{\star}\right\Vert_{TV} + \frac{1}{\lambda_{w}} \left\vert    \langle  \boldsymbol{w} ,\be_{1}+\bg\odot\be_{2} \rangle_{\mathbb{R}}  \right\vert. \label{inequ_represent_measure}
\end{align}
}
Then the last term in \eqref{inequ_represent_measure} can be bounded by 
\begin{align}
 \left\vert\langle\boldsymbol{w},\be_{1}+\bg\odot\be_{2}\rangle \right\vert 
&\le\left\vert\langle\boldsymbol{w},\be_{1}\rangle\right\vert + \left\vert\langle\bar{\bg}\odot\boldsymbol{w},\be_{2}\rangle \right\vert \nonumber \\
&=\left\vert\int_{0}^{1}\langle\boldsymbol{w},\bc\left(\tau\right)\rangle \nu_{1}\left(d\tau\right)\right\vert + \left\vert\int_{0}^{1}\langle\bar{\bg}\odot\boldsymbol{w},\bc\left(\tau\right)\rangle \nu_{2}\left(d\tau\right)\right\vert \nonumber \\
&\le \left\Vert \boldsymbol{w}  \right\Vert_{\cA}^{\star} \left(  \left\| P_{\Upsilon^{1}_{\mathrm{far}}}(\nu_1)\right\|_{TV} + \sum_{j=0}^2 I_{1,j}    \right) + \left\Vert \bar{\bg}\odot \boldsymbol{w}  \right\Vert_{\cA}^{\star} \left(  \left\| P_{\Upsilon^{2}_{\mathrm{far}}}(\nu_2)\right\|_{TV} +\sum_{j=0}^2 I_{2,j}      \right),\label{inequ_comb_error_noise} \\
& \leq  \frac{\lambda_{w}}{C_{w}} \sum_{i=1}^2 \left(  \left\| P_{\Upsilon^{i}_{\mathrm{far}}}(\nu_i)\right\|_{TV} + \sum_{j=0}^2 I_{i,j}    \right) , \label{inequ_comb_error_noisebound}
\end{align}
where \eqref{inequ_comb_error_noise} follows from Lemma~\ref{bound_trig}, and the last inequality \eqref{inequ_comb_error_noisebound} follows from $\left\Vert \boldsymbol{w}  \right\Vert_{\cA}^{\star} \leq   \lambda_{w}/C_{w}   $ and $\left\Vert \bar{\bg}\odot \boldsymbol{w}  \right\Vert_{\cA}^{\star} \leq  \lambda_{w}/C_{w}   $. Moreover, since
\begin{equation*}
\left\Vert\hat{u}_{i}\right\Vert_{TV} = \left\Vert u_{i}^{\star} + \nu_i \right\Vert_{TV} \geq  \left\Vert u_{i}^{\star}\right\Vert_{TV} - \left\Vert P_{\Upsilon_{i}}\left(\nu_{i}\right)\right\Vert_{TV} + \left\Vert P_{\Upsilon_{i}^{c}}\left(\nu_{i}\right)\right\Vert_{TV} , \quad i=1,2,
\end{equation*}
plugging this and \eqref{inequ_comb_error_noisebound} into \eqref{inequ_represent_measure}, we have
\begin{equation}\label{lowerbound} 
\sum_{i=1}^2\left\Vert P_{\Upsilon_{i}^{c}}\left(\nu_{i}\right)\right\Vert_{TV}  - \sum_{i=1}^2  \left\Vert P_{\Upsilon_{i}}\left(\nu_{i}\right)\right\Vert_{TV} \leq  \frac{1}{C_{w}} \sum_{i=1}^2 \left(  \left\| P_{\Upsilon^{i}_{\mathrm{far}}}(\nu_i)\right\|_{TV} + \sum_{j=0}^2 I_{i,j}    \right) .
\end{equation}

Set $P(\tau)$ and $Q(\tau)$ as a pair of polynomials that interpolate the conjugate sign of $P_{\Upsilon_1}(\nu_1)$ and $P_{\Upsilon_2}(\nu_2)$, respectively, whose existence is established in Proposition~\ref{bound_value_far}, then we have
\begin{align}
\sum_{i=1}^2  \left\Vert P_{\Upsilon_{i}}\left(\nu_{i}\right)\right\Vert_{TV} & =  \int_0^1 P(\tau) P_{\Upsilon_1}(\nu_1) (d\tau) + \int_0^1 Q(\tau) P_{\Upsilon_2}(\nu_2) (d\tau) \nonumber \\
& \leq \left| \int_0^1 P(\tau)  \nu_1 (d\tau)+\int_0^1 Q(\tau) \nu_2 (d\tau)  \right| + \left| \int_{\Upsilon_1^c} P(\tau)  \nu_1(d\tau) \right| + \left| \int_{\Upsilon_2^c} Q(\tau)  \nu_2(d\tau) \right|, \label{bound_pol}\\
& \le C\lambda_{w} \sqrt{ \frac{K_{\max}^3 \log M}{M}} + \left| \int_{\Upsilon_1^c} P(\tau)  \nu_1(d\tau) \right| + \left| \int_{\Upsilon_2^c} Q(\tau)  \nu_2(d\tau) \right| \nonumber,
\end{align}
where the first term in \eqref{bound_pol} can be bounded using Lemma~\ref{bound_combined_poly}. For the second term in \eqref{bound_pol}, \reviseA{according to the properties of $P(\tau)$ established in Proposition~\ref{bound_value_far}, we have
\begin{align}
\left| \int_{\Upsilon_1^c} P(\tau)  \nu_1(d\tau) \right|
& = \left\vert \int_{\Upsilon_{\mathrm{far}}^{1}}  P(\tau)  \nu_1(d\tau)  +  \sum_{k=1}^{K_{1}} \int_{\Upsilon_{\mathrm{near}}^{1,k} \backslash \left\{ \tau_{1k} \right\}} P(\tau)  \nu_1(d\tau) \right\vert \nonumber\\
& \le   \sum_{k=1}^{K_{1}} \left\vert \int_{\Upsilon_{\mathrm{near}}^{1,k} \backslash \left\{ \tau_{1k} \right\}} P(\tau)  \nu_1(d\tau) \right\vert  +  \left\vert \int_{\Upsilon_{\mathrm{far}}^{1}}  P(\tau)  \nu_1(d\tau)  \right\vert \nonumber\\
& \le \sum_{k=1}^{K_{1}}  \int_{\Upsilon_{\mathrm{near}}^{1,k} \backslash \left\{ \tau_{1k} \right\}} \left\vert P(\tau)  \right\vert  \left\vert \nu_1 \right\vert (d\tau)   +  \left(1 - C_{b}\right)   \int_{\Upsilon_{\mathrm{far}}^{1}} \left\vert \nu_1 \right\vert (d\tau) \nonumber\\
& \le \sum_{k=1}^{K_{1}}  \int_{\Upsilon_{\mathrm{near}}^{1,k} \backslash \left\{ \tau_{1k} \right\}} \left( 1 - C M^{2}\left(\tau - \tau_{1k}\right)^{2} \right) \left\vert \nu_1 \right\vert (d\tau)   +  \left(1 - C_{b}\right)  \int_{\Upsilon_{\mathrm{far}}^{1}} \left\vert \nu_1 \right\vert (d\tau) \nonumber\\
& = \sum_{k=1}^{K_{1}}  \int_{\Upsilon_{\mathrm{near}}^{1,k} \backslash \left\{ \tau_{1k} \right\}} \left\vert \nu_1 \right\vert (d\tau)  +   \int_{\Upsilon_{\mathrm{far}}^{1}} \left\vert \nu_1 \right\vert (d\tau) \nonumber\\
& \quad - C  \sum_{k=1}^{K_{1}}  \int_{\Upsilon_{\mathrm{near}}^{1,k} \backslash \left\{ \tau_{1k} \right\}}  M^{2}\left(\tau - \tau_{1k}\right)^{2}  \left\vert \nu_1 \right\vert (d\tau)  - C_{b}  \int_{\Upsilon_{\mathrm{far}}^{1}} \left\vert \nu_1 \right\vert (d\tau) \nonumber\\
& \le \left\Vert P_{\Upsilon_{1}^{c}} \left(\nu_{1}\right)\right\Vert_{TV}  -  C_{a}I_{1,2}  - C_{b} \left\| P_{\Upsilon_{\mathrm{far}}^1}( \nu_{1}) \right\|_{TV},
\end{align}
}
for some positive constants $C_a$ and $C_b$. A similar bound holds for the third term. Putting together, we have
\begin{equation}\label{inequ_tv_support}
\sum_{i=1}^2 \left\Vert P_{\Upsilon_{i}^{c}}\left(\nu_{i}\right)\right\Vert_{TV} -\sum_{i=1}^2 \left\Vert P_{\Upsilon_{i}}\left(\nu_{i}\right)\right\Vert_{TV}\geq \sum_{i=1}^2 \left( C_{a}I_{i,2} + C_{b} \left\| P_{\Upsilon_{\mathrm{far}}^i}( \nu_{i}) \right\|_{TV}  \right) - C\lambda_{w} \sqrt{ \frac{K_{\max}^3 \log M}{M}},
\end{equation}
which combined with \eqref{lowerbound} yields:
\begin{equation*}
 \frac{1}{C_{w}} \sum_{i=1}^2 \left(  \left\| P_{\Upsilon^{i}_{\mathrm{far}}}(\nu_i)\right\|_{TV} + \sum_{j=0}^2 I_{i,j}    \right) + C\lambda_{w} \sqrt{ \frac{K_{\max}^3 \log M}{M}} \geq \sum_{i=1}^2 \left( C_{a}I_{i,2} + C_{b} \left\| P_{\Upsilon_{\mathrm{far}}^i}( \nu_{i}) \right\|_{TV}  \right) .
 \end{equation*}
The proof is finished by reorganizing terms and plugging in Proposition~\ref{bound_zero_first_moment}, for a large enough constant $C_{w} > 1$.
\end{proof}

\section{Proof of Lemma~\ref{constructed_poly_p1q1}}\label{sec::proof_constructed_poly_p1q1}
Here we constructed the pair of polynomials $P_{1}\left(\tau\right)$ and $Q_{1}\left(\tau\right)$ using the same techniques as the ones in proof of Theorem~\ref{theorem_main}. Recall the definitions of $K\left(\tau\right)$, $K_{g}\left(\tau\right)$ and $K_{\bar{g}}\left(\tau\right)$ in \eqref{func_K} and \eqref{func_Kg}, and we construct two polynomials $P_{1}\left(\tau\right)$ and $Q_{1}\left(\tau\right)$ as
\begin{equation}\label{func_Pf_1}
P_{1}\left(\tau\right)=\sum_{k=1}^{K_{1}}\theta_{1k}K\left(\tau-\tau_{1k}\right)+\sum_{k=1}^{K_{1}}\psi_{1k}K'\left(\tau-\tau_{1k}\right)+\sum_{k=1}^{K_{2}}\theta_{2k}K_{g}\left(\tau-\tau_{2k}\right)+\sum_{k=1}^{K_{2}}\psi_{2k}K_{g}'\left(\tau-\tau_{2k}\right),
\end{equation} 
and
\begin{equation}\label{func_Qf_1}
Q_{1}\left(\tau\right)=\sum_{k=1}^{K_{1}}\theta_{1k}K_{\bar{g}}\left(\tau-\tau_{1k}\right)+\sum_{k=1}^{K_{1}}\psi_{1k}K_{\bar{g}}'\left(\tau-\tau_{1k}\right)+\sum_{k=1}^{K_{2}}\theta_{2k}K\left(\tau-\tau_{2k}\right)+\sum_{k=1}^{K_{2}}\psi_{2k}K'\left(\tau-\tau_{2k}\right),
\end{equation}
where $\tau_{1k}\in\Upsilon_{1}$ and $\tau_{2k}\in\Upsilon_{2}$. Set the coefficients $\boldsymbol{\theta}_{i}=\left[\theta_{i1},\dots,\theta_{iK_{i}}\right]^{T}$, $\boldsymbol{\psi}_{i}=\left[\psi_{i1},\dots,\psi_{iK_{i}}\right]^{T}$, for $i=1,2$ by solving the following set of equations
\begin{equation*}
\begin{cases}
P_{1}\left(\tau_{1k}\right)  =  0, &\quad \tau_{1k}\in\Upsilon_{1}, \\
P_{1}'\left(\tau_{1k}\right)  =  \mathrm{sign}\left(a_{1k}\right), &\quad \tau_{1k}\in\Upsilon_{1}, \\
Q_{1}\left(\tau_{2k}\right)  =  0, &\quad \tau_{2k}\in\Upsilon_{2}, \\
Q_{1}'\left(\tau_{2k}\right)  =  \mathrm{sign}\left(a_{2k}\right), &\quad \tau_{2k}\in\Upsilon_{2},
\end{cases}
\end{equation*}
which can be rewritten into a matrix form as
\begin{equation*}
\begin{bmatrix}
\boldsymbol{W}_{10} & \frac{1}{\sqrt{\left\vert K''\left(0\right)\right\vert}}\boldsymbol{W}_{11} & \boldsymbol{W}_{g0} & \frac{1}{\sqrt{\left\vert K''\left(0\right)\right\vert}}\boldsymbol{W}_{g1}\\
-\frac{1}{\sqrt{\left\vert K''\left(0\right)\right\vert}}\boldsymbol{W}_{11} & -\frac{1}{\left\vert K''\left(0\right)\right\vert}\boldsymbol{W}_{12} & -\frac{1}{\sqrt{\left\vert K''\left(0\right)\right\vert}}\boldsymbol{W}_{g1} & -\frac{1}{\left\vert K''\left(0\right)\right\vert}\boldsymbol{W}_{g2}\\
\boldsymbol{W}_{\bar{g}0} & \frac{1}{\sqrt{\left\vert K''\left(0\right)\right\vert}}\boldsymbol{W}_{\bar{g}1} & \boldsymbol{W}_{20} & \frac{1}{\sqrt{\left\vert K''\left(0\right)\right\vert}}\boldsymbol{W}_{21}\\
-\frac{1}{\sqrt{\left\vert K''\left(0\right)\right\vert}}\boldsymbol{W}_{\bar{g}1} & -\frac{1}{\left\vert K''\left(0\right)\right\vert}\boldsymbol{W}_{\bar{g}2} & -\frac{1}{\sqrt{\left\vert K''\left(0\right)\right\vert}}\boldsymbol{W}_{21} & -\frac{1}{\left\vert K''\left(0\right)\right\vert}\boldsymbol{W}_{22}\\
\end{bmatrix}
\begin{bmatrix}
\boldsymbol{\theta}_{1}\\
\sqrt{\left\vert K''\left(0\right)\right\vert}\boldsymbol{\psi}_{1}\\
\boldsymbol{\theta}_{2}\\
\sqrt{\left\vert K''\left(0\right)\right\vert}\boldsymbol{\psi}_{2}\\
\end{bmatrix}
=
\begin{bmatrix}
\boldsymbol{0}\\
-\frac{1}{\sqrt{\left\vert K''\left(0\right)\right\vert}}\boldsymbol{u}_{1}\\
\boldsymbol{0}\\
-\frac{1}{\sqrt{\left\vert K''\left(0\right)\right\vert}}\boldsymbol{u}_{2}
\end{bmatrix},
\end{equation*}
whose left-hand side matrix is the same as that in \eqref{coefficient_det}, called $\boldsymbol{W}$, where $K''\left(0\right)$ is the scaler defined in \eqref{equ_K_twoderiv_zero}. Therefore, following Proposition~\ref{prop:invertibility}, under the event $\mathcal{E}_{\delta}$, $\boldsymbol{W}$ is invertible, which gives 
\begin{equation*}
\begin{bmatrix}
\boldsymbol{\theta}_{1}\\
\sqrt{\left\vert K''\left(0\right)\right\vert}\boldsymbol{\psi}_{1}\\
\end{bmatrix}=-\frac{1}{\sqrt{\left\vert K''\left(0\right)\right\vert}}\left(\boldsymbol{R}_{1}\boldsymbol{u}_{1}+\boldsymbol{R}_{g}\boldsymbol{u}_{2}\right),
\quad
\mathrm{and} 
\quad 
\begin{bmatrix}
\boldsymbol{\theta}_{2}\\
\sqrt{\left\vert K''\left(0\right)\right\vert}\boldsymbol{\psi}_{2}\\
\end{bmatrix}=-\frac{1}{\sqrt{\left\vert K''\left(0\right)\right\vert}}\left(\boldsymbol{R}_{\bar{g}}\boldsymbol{u}_{1}+\boldsymbol{R}_{2}\boldsymbol{u}_{2}\right).
\end{equation*} 
\reviseA{And further we know 
\begin{equation*}
\frac{1}{\sqrt{\left\vert K''\left(0\right)\right\vert}^{l}}P_{1}^{\left(l\right)}\left(\tau\right) = - \frac{1}{\sqrt{\left\vert K''\left(0\right)\right\vert}} \boldsymbol{v}_{1l}^{H}\left(\tau\right)\left(\boldsymbol{R}_{1}\boldsymbol{u}_{1}+\boldsymbol{R}_{g}\boldsymbol{u}_{2}\right)  - \frac{1}{\sqrt{\left\vert K''\left(0\right)\right\vert}} \boldsymbol{v}_{2l}^{H}\left(\tau\right)\left(\boldsymbol{R}_{\bar{g}}\boldsymbol{u}_{1}+\boldsymbol{R}_{2}\boldsymbol{u}_{2}\right).
\end{equation*}
}

Under this choice, we will establish that $P_{1}(\tau)$ satisfies the properties in Lemma~\ref{constructed_poly_p1q1}, and $Q_{1}(\tau)$ will follow similarly. Denote 
\begin{equation*}
\frac{1}{\sqrt{\left\vert K''\left(0\right)\right\vert}^{l}}P_{\mu 1}^{\left(l\right)}\left(\tau\right) = -\frac{1}{\sqrt{\left\vert K''\left(0\right)\right\vert}}  \left\langle \boldsymbol{u}_{1},\boldsymbol{R}_{\mu 1}^{H}\boldsymbol{v}_{1l}\left(\tau\right)\right\rangle,
\end{equation*}
then it is straightforward to obtain the following proposition to bound the distance between $\frac{1}{\sqrt{\left\vert K''\left(0\right)\right\vert}^{l}}P_{1}^{\left(l\right)}\left(\tau\right)$ and $\frac{1}{\sqrt{\left\vert K''\left(0\right)\right\vert}^{l}}P_{\mu 1}^{\left(l\right)}\left(\tau\right)$, following essentially the same proof of Proposition~\ref{bound_continuous}.
\begin{lemma}\label{bound_continuous_for_p1}
Suppose $\Delta \ge 1/M$. There exists a numerical constant $C$ such that
\begin{equation*}
M\ge C\max{\left\{\log^{2}{\left(\frac{M\left(K_{1}+K_{2}\right)}{\eta}\right)}, \frac{1}{\epsilon^{2}}K_{\max}\log{\left(\frac{M\left(K_{1}+K_{2}\right)}{\eta}\right)},  \frac{1}{\epsilon^{2}}K_{\max}^{2}\log{\left(\frac{K_{1}+K_{2}}{\eta}\right)} \right\}},
\end{equation*}
then we have 
\begin{equation*}
\mathbb{P}\left\{\left\vert\frac{1}{\sqrt{\left\vert K''\left(0\right)\right\vert}^{l}}P_{1}^{\left(l\right)}\left(\tau\right) - \frac{1}{\sqrt{\left\vert K''\left(0\right)\right\vert}^{l}}P_{\mu 1}^{\left(l\right)}\left(\tau\right)\right\vert\le\frac{\epsilon}{4M+1},\ \forall\tau\in[0,1),\ l=0,1,2,3\right\}\ge 1-\eta.
\end{equation*}
\end{lemma}

When $\tau\in\Upsilon_{\mathrm{far}}^{1}$, since $\left\vert P_{\mu 1}\left(\tau\right)\right\vert\le \frac{C}{4M+1}$ for some numerical constant $C$ \cite[Lemma 2.7]{candes2013super}, under the event in Lemma~\ref{bound_continuous_for_p1}, we have
\begin{equation*}
\left\vert P_{1}\left(\tau\right)\right\vert\le \left\vert P_{\mu 1}\left(\tau\right)\right\vert + \frac{\epsilon}{M}\le\frac{C_1}{M}
\end{equation*}
for some numerical constant $C_1$. Next consider $\left\vert P_{1}\left(\tau\right)-\mathrm{sign}\left(a_{1k}\right)\left(\tau-\tau_{1k}\right)\right\vert$ when $\tau\in\Upsilon_{\mathrm{near}}^{1,k}$. Without loss of generality, assume $\tau_{1k}=0$.
Denote $Z\left(\tau\right)=\mathrm{sign}\left(a_{1k}\right)\tau- P_{1}\left(\tau\right)=Z_{R}\left(\tau\right)+jZ_{I}\left(\tau\right)$, where $Z_{R}\left(\tau\right)$ and $Z_{I}\left(\tau\right)$ are the real part and the imaginary part of $Z\left(\tau\right)$, respectively. Thus we have $Z_{R}\left(0\right)=0$, $Z_{R}^{\prime}\left(0\right)=0$,
$Z_{I}\left(0\right)=0$, and $Z_{I}^{\prime}\left(0\right)=0$. Similarly define $Z_{\mu}\left(\tau\right)=\mathrm{sign}\left(a_{1k}\right)\tau- P_{\mu 1}\left(\tau\right) = Z_{\mu R}\left(\tau\right) + j Z_{\mu R}\left(\tau\right)$, where $Z_{\mu R}\left(\tau\right)$ and $Z_{\mu I}\left(\tau\right)$ are the real part and the imaginary part of $Z_{\mu}\left(\tau\right)$, respectively. Since $\left\vert Z_{\mu R}''\left(\tau\right)\right\vert\le CM$ and $\left\vert Z_{\mu I}''\left(\tau\right)\right\vert\le CM$ for some constant $C$ from the proof of Lemma 6.1 in \cite{candes2013super}, combining with Lemma~\ref{bound_continuous_for_p1}, we can obtain $\left\vert Z_{R}''\left(\tau\right)\right\vert\le C_pM$ and $\left\vert Z_{I}''\left(\tau\right)\right\vert\le C_pM$ with numerical constant $C_{p}$. Then we have $\left\vert\mathrm{sign}\left(a_{1k}\right)\tau- P_{1}\left(\tau\right)\right\vert = \left\vert Z\left(\tau\right)\right\vert\le C_{p}M\tau^{2}$.

\section{Proof of Lemma~\ref{bound_combined_poly}}\label{sec::proof_bound_combined_poly}

We record the following lemma whose proof is given in Appendix~\ref{proof_bound_kg_kgderi}.

\begin{lemma}\label{bound_kg_kgderi}
Set $M\ge 4$. There exist numerical constants $C_{1}$, $C_{2}$ and $C_3$ such that we have 
$\left\vert K_{g}\left(\tau\right)\right\vert\le C_{1}\sqrt{\frac{\log{M}}{M}}$, and $\left\vert K'_{g}\left(\tau\right)\right\vert\le C_{2}\sqrt{M\log{M}}$ with probability at least $1- C_3(M^3\log M)^{-1/2}$.
\end{lemma}

\begin{proof}
Since $P(\tau)=\langle \bp, \bc(\tau) \rangle$, and $Q(\tau)=\langle \bar{\bg} \odot \bp, \bc(\tau) \rangle$, we have
\reviseA{
\begin{align}
\left\vert \int_{0}^{1}P\left(\tau\right)\nu_{1}\left(d\tau\right) + \int_{0}^{1}Q\left(\tau\right)\nu_{2}\left(d\tau\right)\right\vert
&= \left\vert \int_{0}^{1}\langle \boldsymbol{p}, \boldsymbol{c}\left(\tau\right)\rangle\nu_{1}\left(d\tau\right) + \int_{0}^{1}\langle \bar{\boldsymbol{g}}\odot\boldsymbol{p}, \boldsymbol{c}\left(\tau\right)\rangle\nu_{2}\left(d\tau\right)\right\vert \nonumber \\
&= \left\vert \langle \boldsymbol{p}, \int_{0}^{1}\boldsymbol{c}\left(\tau\right)\nu_{1}\left(d\tau\right)\rangle + \langle\boldsymbol{p},\int_{0}^{1}  \boldsymbol{g}\odot\boldsymbol{c}\left(\tau\right)\nu_{2}\left(d\tau\right)\rangle\right\vert \nonumber\\
&=\left\vert\langle\boldsymbol{p},\be_{1}+\bg\odot\be_{2}\rangle\right\vert  \nonumber\\
& = \left\vert \langle P(\tau), E(\tau) \rangle \right\vert \nonumber\\
&\le\left\Vert P\left(\tau\right)\right\Vert_{1}\left\Vert \be_{1}+\bg\odot\be_{2}  \right\Vert_{\cA}^{\star},\label{bound_PQ}
\end{align}
where $E(\tau) = \langle \be_{1}+\bg\odot\be_{2}, \bc(\tau) \rangle $ and $\left\Vert P\left(\tau\right)\right\Vert_{1}=\int_{0}^{1}\left\vert P\left(\tau\right)\right\vert d\tau$. Here the penultimate step follows from Parseval's identity, and the last inequality follows from H\"older's inequality.} Therefore, we need to bound $\left\Vert P\left(\tau\right)\right\Vert_{1}$. Recall 
\begin{equation*}
\begin{bmatrix}
\boldsymbol{\alpha}_{1}\\
\sqrt{\left\vert K''\left(0\right)\right\vert}\boldsymbol{\beta}_{1}\\
\end{bmatrix}=\boldsymbol{L}_{1}\boldsymbol{u}_{1}+\boldsymbol{L}_{g}\boldsymbol{u}_{2},
\quad
\mathrm{and}
\quad
\begin{bmatrix}
\boldsymbol{\alpha}_{2}\\
\sqrt{\left\vert K''\left(0\right)\right\vert}\boldsymbol{\beta}_{2}\\
\end{bmatrix}=\boldsymbol{L}_{\bar{g}}\boldsymbol{u}_{1}+\boldsymbol{L}_{2}\boldsymbol{u}_{2}
\end{equation*}
in \eqref{dual_polynomial_coefficients}. Define 
\begin{equation*}
\begin{bmatrix}
\boldsymbol{\alpha}_{\mu 1}\\
\sqrt{\left\vert K''\left(0\right)\right\vert}\boldsymbol{\beta}_{\mu 1}\\
\end{bmatrix}=\boldsymbol{L}_{\mu 1}\boldsymbol{u}_{1},
\quad
\mathrm{and}
\quad
\begin{bmatrix}
\boldsymbol{\alpha}_{\mu 2}\\
\sqrt{\left\vert K''\left(0\right)\right\vert}\boldsymbol{\beta}_{\mu 2}\\
\end{bmatrix}=\boldsymbol{L}_{\mu 2}\boldsymbol{u}_{2}.
\end{equation*}
From \cite[Lemma 2.2]{candes2014towards}, we have $\left\Vert \boldsymbol{\alpha}_{\mu i}\right\Vert_{\infty}\le C_{\alpha}$ and $\left\Vert \boldsymbol{\beta}_{\mu i}\right\Vert_{\infty}\le\frac{C_{\beta}}{M}$ for some constants $C_{\alpha}$ and $C_{\beta}$, $i=1,2$. Under the event $\mathcal{E}_{\delta}$ for $0<\delta\leq 1/4$ in Lemma~\ref{bound_inver}, we have
\begin{align*}
\left\Vert \begin{bmatrix}
\boldsymbol{\alpha}_{1}\\
\sqrt{\left\vert K''\left(0\right)\right\vert}\boldsymbol{\beta}_{1}\\
\end{bmatrix} - \begin{bmatrix}
\boldsymbol{\alpha}_{\mu 1}\\
\sqrt{\left\vert K''\left(0\right)\right\vert}\boldsymbol{\beta}_{\mu 1}\\
\end{bmatrix}
\right\Vert_{\infty}
&\le\left\Vert (\boldsymbol{L}_{1} - \boldsymbol{L}_{\mu 1} )\boldsymbol{u}_{1}\right\Vert_{\infty} + \left\Vert\boldsymbol{L}_{g}\boldsymbol{u}_{2}\right\Vert_{\infty}\\
&\le\left\Vert (\boldsymbol{L}_{1} - \boldsymbol{L}_{\mu 1} )\boldsymbol{u}_{1} \right\Vert_{2} + \left\Vert\boldsymbol{L}_{g}\boldsymbol{u}_{2}\right\Vert_{2}\\
&\le\left\Vert \boldsymbol{L}_{1} - \boldsymbol{L}_{\mu 1} \right\Vert \left\Vert \boldsymbol{u}_{1}\right\Vert_{2} + \left\Vert \boldsymbol{L}_{g} \right\Vert \left\Vert \boldsymbol{u}_{2}\right\Vert_{2}\le C \delta \sqrt{K_{\max}}.
\end{align*}
Therefore, we have $\left\Vert \boldsymbol{\alpha}_{1}\right\Vert_{\infty}\le C_{\alpha}'\sqrt{K_{\max}}$ and $\left\Vert \boldsymbol{\beta}_{1}\right\Vert_{\infty}\le\frac{C_{\beta}'}{M}\sqrt{K_{\max}}$ for some constants $C_{\alpha}'$ and $C_{\beta}'$. Similar bounds hold for $\left\Vert \boldsymbol{\alpha}_{2}\right\Vert_{\infty}$ and $\left\Vert \boldsymbol{\beta}_{2}\right\Vert_{\infty}$ as well. Then $\left\Vert P\left(\tau\right)\right\Vert_{1}$ can be bounded as below:
\begin{align*}
& \left\Vert P\left(\tau\right)\right\Vert_{1}=\int_{0}^{1}\left\vert P\left(\tau\right)\right\vert d\tau\\
&\le K_{1}\left\Vert \boldsymbol{\alpha}_{1}\right\Vert_{\infty}\int_{0}^{1}\left\vert K\left(\tau\right)\right\vert d\tau
      + K_{1}\left\Vert \boldsymbol{\beta}_{1}\right\Vert_{\infty}\int_{0}^{1}\left\vert K'\left(\tau\right)\right\vert d\tau + K_{2}\left\Vert \boldsymbol{\alpha}_{2}\right\Vert_{\infty}\int_{0}^{1}\left\vert K_{g}\left(\tau\right)\right\vert d\tau + K_{2}\left\Vert \boldsymbol{\beta}_{2}\right\Vert_{\infty}\int_{0}^{1}\left\vert K_{g}'\left(\tau\right)\right\vert d\tau\\
&\le K_{1} C_{\alpha}'\sqrt{K_{\max}} \frac{C}{M} + K_{1}\frac{C_{\beta}'}{M}\sqrt{K_{\max}} C + K_{2}C_{\alpha}'\sqrt{K_{\max}}C_{1}\sqrt{\frac{\log{M}}{M}} + K_{2}\frac{C_{\beta}'}{M}\sqrt{K_{\max}}C_{2}\sqrt{M\log{M}}\\
&\le C_{p}\sqrt{\frac{K_{\max}^3\log M}{M}},
\end{align*}
where we used $\int_{0}^{1}\left\vert K\left(\tau\right)\right\vert d\tau \le\frac{C}{M}$, $\int_{0}^{1}\left\vert K'\left(\tau\right)\right\vert d\tau\le C$ from \cite[Lemma 4]{tang2013near}, and $ \left\vert K_{g}\left(\tau\right)\right\vert \le C_{1}\sqrt{\frac{\log{M}}{M}}$, $\left\vert K_{g}'\left(\tau\right)\right\vert\le C_{2}\sqrt{M\log{M}}$ from Lemma~\ref{bound_kg_kgderi}. Plugging this into \eqref{bound_PQ} and combining \eqref{noise_atomic_bound}, we have proved \eqref{bound_poly_eq1}.

Next, we can write similarly that
\begin{equation}\label{bound_P1Q1}
\left\vert \int_{0}^{1}P_{1}\left(\tau\right)\nu_{1}\left(d\tau\right) + \int_{0}^{1}Q_{1}\left(\tau\right)\nu_{2}\left(d\tau\right)\right\vert\le\left\Vert P_{1}\left(\tau\right)\right\Vert_{1}\left\Vert \be_{1}+\bg\odot\be_{2}  \right\Vert_{\cA}^{\star},
\end{equation}
then it suffices to bound $\left\Vert P_{1}\left(\tau\right)\right\Vert_{1}$. Recall that
\begin{equation*}
\begin{bmatrix}
\boldsymbol{\theta}_{1}\\
\sqrt{\left\vert K''\left(0\right)\right\vert}\boldsymbol{\psi}_{1}\\
\end{bmatrix}=-\frac{1}{\sqrt{\left\vert K''\left(0\right)\right\vert}}\left(\boldsymbol{R}_{1}\boldsymbol{u}_{1}+\boldsymbol{R}_{g}\boldsymbol{u}_{2}\right),
\quad
\mathrm{and}
\quad
\begin{bmatrix}
\boldsymbol{\theta}_{2}\\
\sqrt{\left\vert K''\left(0\right)\right\vert}\boldsymbol{\psi}_{2}\\
\end{bmatrix}=-\frac{1}{\sqrt{\left\vert K''\left(0\right)\right\vert}}\left(\boldsymbol{R}_{\bar{g}}\boldsymbol{u}_{1}+\boldsymbol{R}_{2}\boldsymbol{u}_{2}\right)
\end{equation*}
in Appendix~\ref{sec::proof_constructed_poly_p1q1}. Define
\begin{equation*}
\begin{bmatrix}
\boldsymbol{\theta}_{\mu 1}\\
\sqrt{\left\vert K''\left(0\right)\right\vert}\boldsymbol{\psi}_{\mu 1}\\
\end{bmatrix}=-\frac{1}{\sqrt{\left\vert K''\left(0\right)\right\vert}}\boldsymbol{R}_{\mu 1}\boldsymbol{u}_{1},
\quad
\mathrm{and}
\quad
\begin{bmatrix}
\boldsymbol{\theta}_{\mu 2}\\
\sqrt{\left\vert K''\left(0\right)\right\vert}\boldsymbol{\psi}_{\mu 2}\\
\end{bmatrix}=-\frac{1}{\sqrt{\left\vert K''\left(0\right)\right\vert}}\boldsymbol{R}_{\mu 2}\boldsymbol{u}_{2}.
\end{equation*}
From \cite[Lemma 2.7]{candes2013super}, we have $\left\Vert \boldsymbol{\theta}_{\mu i}\right\Vert_{\infty}\le C_{\theta}/M$ and $\left\Vert \boldsymbol{\psi}_{\mu i}\right\Vert_{\infty}\le C_{\psi}/M^2$ for some constants $C_{\theta}$ and $C_{\psi}$, $i=1,2$. Following similar arguments as above, we have
$\left\Vert \boldsymbol{\theta}_{i}\right\Vert_{\infty}\le C_{\theta}'\sqrt{K_{\max}}/M$ and $\left\Vert \boldsymbol{\psi}_{i}\right\Vert_{\infty}\le C_{\psi}'\sqrt{K_{\max}}/M^2$, $i=1,2$. Hence $\left\Vert P_{1}\left(\tau\right)\right\Vert_{1}$ can be bounded as
\begin{align*}
& \left\Vert P_{1}\left(\tau\right)\right\Vert_{1}=\int_{0}^{1}\left\vert P_{1}\left(\tau\right)\right\vert d\tau\\
&\le K_{1}\left\Vert \boldsymbol{\theta}_{1}\right\Vert_{\infty}\int_{0}^{1}\left\vert K\left(\tau\right)\right\vert d\tau
      + K_{1}\left\Vert \boldsymbol{\psi}_{1}\right\Vert_{\infty}\int_{0}^{1}\left\vert K'\left(\tau\right)\right\vert d\tau+ K_{2}\left\Vert \boldsymbol{\theta}_{2}\right\Vert_{\infty}\int_{0}^{1}\left\vert K_{g}\left(\tau\right)\right\vert d\tau + K_{2}\left\Vert \boldsymbol{\psi}_{2}\right\Vert_{\infty}\int_{0}^{1}\left\vert K_{g}'\left(\tau\right)\right\vert d\tau\\
&\le K_{1} \frac{C_{\theta}'}{M}\sqrt{K_{\max}} \frac{C}{M} + K_{1} \frac{C_{\psi}'}{M^{2}}\sqrt{K_{\max}} C + K_{2} \frac{C_{\theta}'}{M}\sqrt{K_{\max}} C_{1}\sqrt{\frac{\log{M}}{M}} + K_{2}\frac{C_{\psi}'}{M^{2}}\sqrt{K_{\max}} C_{2}\sqrt{M\log{M}}\\
&\le C_{p}' \sqrt{\frac{K_{\max}^3 \log{M}}{M^3}}.
\end{align*}
Plugging this into \eqref{bound_P1Q1} and combining \eqref{noise_atomic_bound}, we have proved \eqref{bound_poly_eq2}.
\end{proof}

\section{Proof of Lemma~\ref{bound_kg_kgderi}}\label{proof_bound_kg_kgderi}
\begin{proof}
Suppose $M\geq 4$. For a fixed $\tau\in[0,1)$, applying the Hoeffding's inequality in Lemma~\ref{hoeffding_inequality}, we have
\begin{equation*}
\mathbb{P}\left\{\left\vert K_{g}\left(\tau \right)\right\vert\ge \zeta\right\}
=\mathbb{P}\left\{ \left\vert\frac{1}{M} \sum_{n=-2M}^{2M}s_{n}g_{n}e^{j2\pi n\tau}\right\vert\ge \zeta\right\} \le 4e^{-\frac{M^{2}\zeta^{2}}{4\sum_{n=-2M}^{2M}s_{n}^{2}}} \le 4e^{-\frac{M^{2}\zeta^{2}}{4\left(4M+1\right)}} \le 4e^{-\frac{M\zeta^{2}}{17}},
\end{equation*}
where we used $|s_n|\leq 1$. Let $\Upsilon_{\mathrm{grid}}=\left\{\tau_{d}\in[0,1)\right\}$ be a uniform grid of $[0,1)$ whose size will be determined later.  As a result of the union bound, we have
$$\mathbb{P}\left\{\sup_{\tau_{d}\in\Upsilon_{\mathrm{grid}}} \left\vert K_{g}\left(\tau_{d}\right)\right\vert\le\zeta\right\}\ge 1-4\left\vert\Upsilon_{\mathrm{grid}}\right\vert e^{-\frac{M\zeta^{2}}{17}}.$$
For any $\tau_a,\tau_b\in [0,1)$, following Lemma~\ref{bernsteins_poly_inequality} we have
\begin{align*}
\left\vert K_{g}\left(\tau_{a}\right) - K_{g}\left(\tau_{b}\right) \right\vert &\le \left\vert e^{j2\pi\tau_{a}} - e^{j2\pi\tau_{b}}\right\vert \sup_{\tau}\left\vert\frac{\partial K_{g}\left(\tau\right)}{\partial e^{j2\pi\tau}}\right\vert\le 4\pi\left\vert\tau_{a}-\tau_{b}\right\vert 2M \sup_{\tau}\left\vert K_{g}\left(\tau\right)\right\vert \le 40\pi M \left\vert\tau_{a}-\tau_{b}\right\vert,
\end{align*}
where the last inequality follows from $\left\vert K_{g}\left(\tau\right)\right\vert\le \frac{1}{M}\sqrt{\sum_{n=-2M}^{2M}s_{n}^{2}}\sqrt{\sum_{n=-2M}^{2M}\left\vert g_{n}e^{j2\pi n\tau}\right\vert^{2}}\le \frac{4M+1}{M}\le5$. By choosing the grid size such that for any $\tau\in[0,1)$, there exists a point $\tau_{d}\in\Upsilon_{\mathrm{grid}}$ satisfying $40\pi M \left\vert\tau-\tau_{d}\right\vert\le\zeta$, which means we can set $\left\vert\Upsilon_{\mathrm{grid}}\right\vert=\lceil\frac{40\pi M}{\zeta}\rceil$. Consequently, for any $\tau\in[0,1)$, we have
\begin{equation*}
\left\vert  K_{g}\left(\tau\right)\right\vert\le\left\vert  K_{g}\left(\tau\right) -  K_{g}\left(\tau_{d}\right) \right\vert + \left\vert  K_{g}\left(\tau_{d}\right)\right\vert \le 40\pi M \left\vert\tau -\tau_{d}\right\vert + \zeta\le 2\zeta,
\end{equation*}
with probability at least $1-4\left\vert\Upsilon_{\mathrm{grid}}\right\vert e^{-\frac{M\zeta^{2}}{17}}$. Choose $\zeta = \sqrt{\frac{51\log M}{M}}$, then we have 
$$ \mathbb{P}\left\{ |K_g(\tau)|\leq 2\sqrt{\frac{51\log M}{M}} \right\} \ge 1- 71 (M^3\log M)^{-1/2}. $$

Next consider $\left\vert K'_{g}\left(\tau\right)\right\vert$. For a fixed $\tau\in[0,1)$, applying the Hoeffding's inequality in Lemma~\ref{hoeffding_inequality}, we have
\begin{align*}
\mathbb{P}\left\{\left\vert K'_{g}\left(\tau \right)\right\vert\ge \zeta\right\}
& =\mathbb{P}\left\{\left\vert \frac{1}{M}\sum_{n=-2M}^{2M}s_{n}g_{n}e^{j2\pi n\tau}\left(j2\pi n\right)\right\vert\ge \zeta\right\} \\
& \le 4e^{-\frac{M^{2}\zeta^{2}}{4\sum_{n=-2M}^{2M}s_{n}^{2}\left(2\pi n\right)^{2}}}
\le 4e^{-\frac{\zeta^{2}}{321\pi M}}.
\end{align*}
Set $\Upsilon_{\mathrm{grid}}=\left\{\tau_{d}\in[0,1)\right\}$ be a uniform grid of $[0,1)$ whose size will be determined later.  As a result of the union bound, we have
\begin{equation*}
\mathbb{P}\left(\sup_{\tau_{d}\in\Upsilon_{\mathrm{grid}}} \left\vert K'_{g}\left(\tau_{d}\right)\right\vert\leq \zeta\right)\ge 1- 4\left\vert\Upsilon_{\mathrm{grid}}\right\vert e^{-\frac{\zeta^{2}}{321\pi M}}.
\end{equation*}
For any $\tau_a,\tau_b\in [0,1)$, following Lemma~\ref{bernsteins_poly_inequality} we have
\begin{equation*}
\left\vert K'_{g}\left(\tau_{a}\right) - K'_{g}\left(\tau_{b}\right) \right\vert \le 4\pi\left\vert\tau_{a}-\tau_{b}\right\vert 2M \sup_{\tau}\left\vert K'_{g}\left(\tau\right)\right\vert \le 88\pi^{2} M^{2} \left\vert\tau_{a}-\tau_{b}\right\vert,
\end{equation*} 
where in the last inequality we use $\left\vert K'_{g}\left(\tau\right)\right\vert\le\frac{1}{M}\sqrt{\sum_{n=-2M}^{2M}s_{n}^{2}}\sqrt{\sum_{n=-2M}^{2M}\left\vert g_{n}e^{j2\pi n\tau}\left(j2\pi n\right)\right\vert^{2}}\le 11\pi M$. Hence, by choosing the grid size such that for any $\tau\in[0,1)$, there exists a point $\tau_{d}\in\Upsilon_{\mathrm{grid}}$ satisfying $88\pi^{2} M^{2} \left\vert\tau - \tau_{d}\right\vert \leq \zeta$, which gives $\left\vert\Upsilon_{\mathrm{grid}}\right\vert=\lceil \frac{88\pi^{2}M^{2}}{\zeta}\rceil$. Then for any $\tau\in[0,1)$, we have
\begin{equation*}
\left\vert  K'_{g}\left(\tau\right)\right\vert\le\left\vert  K'_{g}\left(\tau\right) -  K'_{g}\left(\tau_{d}\right) \right\vert + \left\vert  K'_{g}\left(\tau_{d}\right)\right\vert \le 88\pi^{2} M^{2} \left\vert\tau -\tau_{d}\right\vert + \zeta\le 2\zeta
\end{equation*}
with probability at least $1-4\left\vert\Upsilon_{\mathrm{grid}}\right\vert e^{-\frac{\zeta^{2}}{321\pi M}}$. Choosing $\zeta=\sqrt{963\pi M\log M}$ gives 
\begin{equation*}
\mathbb{P}\left( \left\vert K'_{g}\left(\tau\right)\right\vert\leq 2\sqrt{963\pi M\log M} \right)\ge 1- 64(M^3\log M)^{-1/2}.
\end{equation*}
\end{proof}

\end{document}